\documentclass{article}

\usepackage[pagebackref,colorlinks,hypertexnames=false]{hyperref}
\usepackage{amsmath}
\usepackage{amsthm}
\usepackage{amssymb}
\usepackage[dvips,letterpaper,margin=1in,bottom=1in]{geometry}

\usepackage[utf8]{inputenc}
\usepackage[english]{babel}

\usepackage{bm}

\usepackage{subcaption}
\usepackage{tikz}
\usepackage{pgfplots}
\pgfplotsset{compat=1.12}

\pgfplotsset{
    discard unless two cols/.style n args={4}{
        x filter/.code={
            \edef\vala{\thisrow{#1}}
            \edef\wanta{#2}
            \edef\valb{\thisrow{#3}}
            \edef\wantb{#4}

            \ifx\vala\wanta
                \ifx\valb\wantb
                \else
                    
                \fi
            \else
                
            \fi
        }
    }
}

\usepackage{zref-clever}
\newcommand{\cref}[1]{\zcref{#1}}
\zcsetup{
    cap,
    nameinlink=false
}
\zcRefTypeSetup{equation}{
    Name-sg = Eq.,
    Name-pl = Eqs.,
    refbounds = {(,,,)}
}

\usepackage{mathtools}

\usepackage{algorithm}
\usepackage{algpseudocode}

\newcommand{\addtheorem}[1]{%
  \AddToHook{env/#1/begin}{%
    \zcsetup{countertype={theorem=#1}}%
  }%
  \zcRefTypeSetup{#1}{Name-sg={\MakeUppercase #1}}%
  \newtheorem{#1}[theorem]{\MakeUppercase #1}%
}

\newtheorem{theorem}{Theorem}[section]
\addtheorem{axiom}
\addtheorem{lemma}
\addtheorem{corollary}
\addtheorem{proposition}
\addtheorem{fact}
\addtheorem{remark}

\newcommand{\ketbra}[3][]{\left\lvert #2 \vphantom{#3} \right>_{#1}\!\!\left< #3 \vphantom{#2} \right\rvert}

\DeclarePairedDelimiter\rbra{\lparen}{\rparen}
\DeclarePairedDelimiter\sbra{\lbrack}{\rbrack}
\DeclarePairedDelimiter\cbra{\{}{\}}
\DeclarePairedDelimiter\abs{\lvert}{\rvert}
\DeclarePairedDelimiter\Abs{\lVert}{\rVert}
\DeclarePairedDelimiter\ceil{\lceil}{\rceil}
\DeclarePairedDelimiter\floor{\lfloor}{\rfloor}
\DeclarePairedDelimiter\ket{\lvert}{\rangle}
\DeclarePairedDelimiter\bra{\langle}{\rvert}
\DeclarePairedDelimiter\ave{\langle}{\rangle}

\DeclareMathOperator{\Span}{span}
\DeclareMathOperator{\tr}{tr}

\newcommand{\CC}{\mathbb{C}}
\newcommand{\EE}{\mathbb{E}}
\newcommand{\NN}{\mathbb{N}}
\newcommand{\RR}{\mathbb{R}}

\newcommand{\cA}{\mathcal{A}}
\newcommand{\cF}{\mathcal{F}}

\title{Weakly Measured Loops for Quantum Amplitude Amplification:\\
Oracle Savings and Adaptive Search with Unknown Target Probability
}
\author{
    Wang Fang\footnote{United Arab Emirates University, \url{Wang.Fang@uaeu.ac.ae}.}
    \and Chris Heunen\footnote{University of Edinburgh, \url{Chris.Heunen@ed.ac.uk}.}
}
\date{}

\begin{document}

\maketitle

\begin{abstract}
We present a family of quantum amplitude amplification algorithms that interleave ordinary Grover rotations with tunable weak measurements in loops.
A successful measurement yields the target state and stops the algorithm, while after a failure, the loop resumes from the post-measurement state without restarting.
We express the leading costs in expected Grover iterations as $p \to 0$.
For a known target probability $p$, an exact weak measurement-conditioned loop for states near the angle $\pi/4$ uses $(\pi/8+1/4+o(1))/\sqrt{p}$ iterations, improving on the standard Grover constant $\pi/4$ and on optimised restart Grover, and matching the corresponding infimum in our continuous-angle analysis.
When only a lower bound $0 < p_0 \leq p < 1/2$ is available, a discrete Lyapunov equation gives the exact expected oracle cost for every fixed measurement strength and a closed-form optimal strength.
Using the strength selected from $p_0$, the expected number of Grover iterations is at most $\frac{1}{2\sqrt{2}}(\frac{1}{\sqrt{p}}+\frac{1}{\sqrt{p_0}})$, where the expected cost decreases as the actual $p$ increases, and gives the leading constant $1/\sqrt{2}$ at the promise boundary $p = p_0$.
For completely unknown $p$, we identify measurement strength $\kappa(t)=\Theta(1/t)$ as the critical scale within the regular schedules considered here and analyse $\kappa_b(t)=\min\{1/2,b/t\}$.
For each fixed $b > 2$, we rigorously derive a closed-form Gamma-function expression $C(b)$ giving $(C(b)/2+o(1))/\sqrt{p}$ expected iterations.
Numerical minimisation of this explicit expression gives $b \approx 5.2$ and $C(b)/2 \approx 1.01$.
The results show that weak measurements preserve the $\Theta(1/\sqrt{p})$ search scale while adding an explicit fixed-point control mechanism and provable oracle savings.
\end{abstract}

\newpage

\tableofcontents
\newpage

\section{Introduction}
Grover search~\cite{Grover96} and quantum amplitude amplification~\cite{BHM+02} are canonical examples of the overcooking problem.
The optimal query complexity is $\Theta\rbra{1/\sqrt{p}}$~\cite{BBH+98},
the dynamics are confined to a two-dimensional plane, and the danger of applying too many rotations is well understood.

Let $p = \Abs*{\Pi_{\text{target}}\cA\ket{0^n}}^2$ denote the target probability, i.e., the initial probability of the target subspace, for a state-preparation unitary oracle $\cA$.
If $p$ is known exactly, one can choose an appropriate number of Grover iterations and measure at the end.
When needed, phase-matched variants can
further adjust the reflections so that the last rotation lands exactly in the target subspace, achieving success probability $1$~\cite{Hoyer00,Long01,BHM+02}.
If only a lower bound $p_0 \leq p$ is available, a fixed-point construction is more suitable.\footnote{In this discussion, we restrict attention to amplification procedures that retain the Grover-type quadratic speedup in $O\rbra{1/\sqrt{p_0}}$ for fixed target error, rather than to arbitrary fixed-point variants.}
Instead of repeatedly applying the same $\pi$-phase reflections, fixed-point quantum search~\cite{YLC14,GSL+19} uses a carefully chosen sequence of phase-modified reflections, whose overall action is robust for every $p \geq p_0$, thereby preventing overcooking from larger values of $p$.
When $p$ is unknown and no lower bound is available, a pre-determined number of iterations cannot be applied directly.
In this regime, one instead uses randomized stopping
times, geometrically increasing trial lengths with restarts, or amplitude estimation and quantum counting subroutines to achieve the
relevant scale $O\rbra{1/\sqrt{p}}$~\cite{BBH+98,BHM+02}.

Most previous work has focused on adjusting the phases of the Grover reflections and on choosing the number of iterations.
\emph{We study a different control mechanism: weak measurement}.
Here, a weak measurement of the target subspace is the two-outcome measurement with operators
\begin{equation*}
    M_0 = I-\Pi_{\text{target}}+\sqrt{1-\kappa}\,\Pi_{\text{target}},
    \qquad M_1 = \sqrt{\kappa}\,\Pi_{\text{target}},
    \qquad 0 < \kappa \leq 1.
\end{equation*}
Outcome $1$ certifies success and projects into the target subspace.
Outcome $0$ reduces the unnormalised target amplitude by $\sqrt{1-\kappa}$ while preserving the non-target amplitude.
For small $\kappa$, the resulting normalised state is close to the state before measurement.
All of our algorithms in this paper are based on the same elementary loop,
\begin{equation*}
    \text{weak measurement of the target qubit} \quad \longrightarrow\quad \text{on failure, one ordinary Grover rotation}.
\end{equation*}
If the measurement succeeds, the algorithm halts on a target state; a failed measurement with $\kappa < 1$ preserves both components of the Grover plane, allowing progress from preceding rotations to be retained.
The algorithm can therefore continue from the
post-measurement state rather than restart from scratch.

The motivation comes from a simple geometric observation that will be illustrated in Section~\ref{sec:pi_over_4_and_weak_measurement}.
Standard amplitude amplification is usually analysed globally by choosing the number of rotations so that the final angle is close to $\pi/2$.
We instead consider the local increase in target probability caused by one Grover rotation.
As shown in \cref{sec:pi_over_4_and_weak_measurement}, this increase is largest when the rotation is centred at the balanced angle $\pi/4$.
Near this region, a weak measurement detects the target with probability approximately $\kappa/2$, while
the failure branch can still be placed near the angle $\pi/4$ for the next rotation and weak measurement.

\emph{This observation suggests replacing the global question
``when should Grover search stop?'' by the local question
``how strongly should the state be measured after each rotation?''}
For known $p$, the strength can be chosen so that a failed measurement returns the state to a prescribed angle and one further Grover rotation restores the same checkpoint.
When only a lower bound $p_0$ is known, exact angle control is impossible, but a fixed strength can damp the dynamics uniformly over all $p \geq p_0$.
When $p$ is completely unknown, the measurement strength can instead be decreased adaptively as a function of the number of preceding failures.

Thus, our algorithms are genuine measurement-conditioned loops, which we call \emph{weakly measured loops} following previous work~\cite{AMH22}, rather than fixed-length sequences of Grover rotations. 
Their unitary component is always the ordinary Grover rotation, and the algorithmic design problem is entirely the choice of weak measurement strengths.
We solve this problem in three information regimes:
\begin{equation*}
    \text{known }p, \qquad \text{lower-bounded }p, \quad \text{and} \quad  \text{unknown }p,
\end{equation*}
The $\pi/4$ observation gives the exact geometric picture in the first regime and the guiding intuition for the remaining two. 
The later algorithms do not require the state to be steered explicitly to $\pi/4$; instead, they use the same weakly measured loops without knowledge of the true Grover angle.

\subsection{Our results}
Before illustrating our results, let us review the standard notions for quantum amplitude amplification.
We consider the standard two-dimensional subspace for quantum amplitude amplification~\cite{BHM+02}.
Let
\begin{equation}\label{eq:oracle_A}
    \cA\ket{0^n} = \cos\rbra{\theta}\ket{0}\ket{\psi_0} + \sin\rbra{\theta}\ket{1}\ket{\psi_1}, \qquad \theta = \arcsin\rbra{\sqrt{p}} \in \rbra{0,\pi/4},
\end{equation}
where $p$ is the probability for the target state $\ket{1}\ket{\psi_1}$.
The Grover rotation is given by
\begin{equation}\label{eq:grover_operator}
    G = -\cA S_0 \cA^{\dagger} S_{1},
\end{equation}
where $S_0 = I-2\ketbra{0^n}{0^n}$ reflects about the orthogonal complement of the all-zero input state, and $S_1 = \rbra*{\ketbra{0}{0}-\ketbra{1}{1}}\otimes I = I-2\Pi_{\text{target}}$ flips the phase of the target subspace.
$G$ rotates any state inside the subspace $\Span\cbra{\ket{0}\ket{\psi_0}, \ket{1}\ket{\psi_1}}$ by an angle $2\theta$ as 
\begin{equation}\label{eq:grover_rotation}
    G\ket{\phi_\alpha} = \ket{\phi_{\alpha+2\theta}},
\end{equation}
where we write
\begin{equation}\label{eq:alpha_state}
    \ket{\phi_\alpha} = \cos\rbra{\alpha}\ket{0}\ket{\psi_0} + \sin\rbra{\alpha}\ket{1}\ket{\psi_1}
\end{equation}
for any angle $\alpha \in \sbra{0,2\pi}$.

\paragraph{Cost conventions.}
Let $N_G$ denote the number of ordinary Grover iterations and $N_{\cA/\cA^\dagger}$ the total number of calls to the state-preparation oracle or its adjoint.
Each of our algorithms prepares the register once, using one oracle call, and each Grover iteration uses two further calls. Thus,
\begin{equation*}
    N_{\cA/\cA^\dagger} = 2N_G+1,
    \qquad \text{total angle supplied by Grover rotations} = 2\theta N_G,
\end{equation*}
The rotation angle is a cumulative cost, not the final state angle after measurement backaction.
An expected oracle cost $\rbra{c+o\rbra{1}}/\sqrt{p}$ therefore corresponds to $\rbra{c/2+o\rbra{1}}/\sqrt{p}$ expected Grover iterations, since $\theta \sim \sqrt{p}$.
The abstract uses Grover iterations; the formal results retain oracle counts because the input model supplies query access to $\cA$ and $\cA^\dagger$.

\paragraph{Exact schedule for known $p$.}
We first consider the case where the target probability $p$, and hence $\theta$, is known.
In this regime, we can easily control the angle $\alpha$ of the state $\ket{\phi_{\alpha}}$ because we know the exact $\theta$.
We rotate the initial state $\cA\ket{0^n}$ to the first checkpoint state with an angle at or slightly beyond $\pi/4$.
The weak measurement strength is then chosen so that, upon a failure outcome, the angle decreases by exactly $2\theta$.
One subsequent Grover rotation therefore restores the checkpoint, and the algorithm repeats this one-step loop until the target is observed.

\begin{theorem}
    [Oracle counts for known $p$; {\cref{thm:exact_known_p}, restated}]
    Given query access to the oracle $\cA$ and its inverse $\cA^\dagger$ in \cref{eq:oracle_A}, together with the target probability $0 < p \leq \sin^2\rbra{\frac{\pi}{8}}$ and $\theta = \arcsin\rbra{\sqrt{p}}$, there is a quantum algorithm that takes $p$ as input and outputs the target state $\ket{1}\ket{\psi_1}$.
    Its expected total number of calls to $\cA$ (and $\cA^{\dagger}$) is less than
    \begin{equation*}
        \rbra*{\frac{\pi}{4}+\frac{2\theta}{\sin\rbra{4\theta}}+\theta}\cdot \frac{1}{\theta}.
    \end{equation*}
\end{theorem}
As $p \to 0$, the leading constant in the $1/\theta \sim 1/\sqrt{p}$ scaling is $\pi/4+1/2$.
Equivalently, the leading constant in expected Grover iterations is $\pi/8+1/4$, compared with $\pi/4$ for standard Grover amplification.
Under our oracle-counting convention, standard Grover amplification has asymptotic cost $\frac{\pi}{2}\frac{1}{\theta}$.
The ratio of our leading constant to the standard Grover constant is therefore $\frac{\pi/4+1/2}{\pi/2} \approx 0.8183$.
Thus, the exact weak measurement schedule saves approximately $18\%$ of the oracle calls in the small-$p$ limit.
It also improves on the optimised restart-after-failure strategy of~\cite{BBH+98}, whose corresponding ratio is approximately $0.88$.

The continuous-angle calculation in
\cref{sec:continuous_heuristic} further shows that
$\pi/4+1/2$ is the infimum within the natural family of one-checkpoint protocols considered there.
In this sense, the $\pi/4$ region is not only a useful heuristic: it identifies the optimal limiting checkpoint for this class of weakly measured loops.

\paragraph{Fixed-point schedule for lower-bounded $p$.}
We next assume that $p$ is not known exactly, but that a lower bound $p_0 \leq p$ is available.
Exact control of the $\pi/4$ region is no longer possible, because the rotation angle $2\theta$ is unknown.
We therefore prepare $\cA\ket{0^n}$ and repeatedly perform a weak measurement of fixed strength $\kappa$, applying one Grover rotation after each failure.
Writing $x = \frac{1-\sqrt{1-\kappa}}{1+\sqrt{1-\kappa}}$,
the exact expected oracle cost, including preparation, is $Q\rbra{p,\kappa} = \frac{1}{x}+\frac{x}{2p}$.
This expression is minimised at $x = \sqrt{2p}$, giving the optimal strength
\begin{equation*}
    \kappa^*\rbra{\theta} = \frac{4\sqrt{2}\sin\rbra{\theta}}{\rbra*{1+\sqrt{2}\sin\rbra{\theta}}^2}
    = \frac{4\sqrt{2p}}{\rbra*{1+\sqrt{2p}}^2},
\end{equation*}
as proved in \cref{thm:fixed_point_expected_cost}. If only $p_0$ is known, we select the strength using $\theta_0 = \arcsin\rbra{\sqrt{p_0}}$, which leads to the
following result.

\begin{theorem}
    [Oracle counts for lower-bounded $p \geq p_0$; {\cref{thm:fixed_point_lower_bounded}, restated}]
    Given query access to $\cA$ and $\cA^\dagger$ satisfying \cref{eq:oracle_A}, and a known lower bound $0 < p_0 \leq p < 1/2$, there is a quantum algorithm that takes only $p_0$ as input and outputs the target state $\ket{1}\ket{\psi_1}$.
    Its expected total number of oracle calls is
    \begin{equation*}
        \frac{1}{\sqrt{2p_0}}+\frac{\sqrt{p_0}}{\sqrt{2}\,p}
        \leq \frac{1}{\sqrt{2}}\rbra*{\frac{1}{\sqrt{p_0}}+\frac{1}{\sqrt{p}}}.
    \end{equation*}
\end{theorem}
At $p = p_0$, the exact expected oracle cost is $\sqrt{2/p_0}$ and the expected number of Grover iterations is $1/\sqrt{2p_0}-1/2$. The leading ratio to standard Grover amplification is $2\sqrt{2}/\pi \approx 0.9003$, corresponding to a saving of approximately $10\%$.
More importantly, the expected cost automatically decreases when the true $p$ is larger than the promised lower bound.
Conventional fixed-point amplitude amplification~\cite{YLC14,GSL+19} instead chooses a predetermined query count from $p_0$ and the desired failure probability, and this count does not decrease when the actual $p$ is larger.

\paragraph{Adaptive schedule for unknown $p$.}
Our main result concerns the case in which $p$ is completely unknown.
We retain the weakly measured loop and choose the measurement strength in round $t\geq1$ using only the number $t-1$ of preceding failures.
Within the regular schedule classes specified in \cref{lem:kappa_order}, we identify
\begin{equation*}
    \kappa\rbra{t} = \Theta\rbra*{\frac{1}{t}}
\end{equation*}
as the critical decay scale for retaining the quadratic speedup.
This leads to the one-parameter family
\begin{equation*}
    \kappa_b\rbra{t} = \min\cbra*{\frac{1}{2},\frac{b}{t}}, \qquad t > 0
\end{equation*}
for measurement strengths, and the following result.

\begin{theorem}
    [Oracle counts for completely unknown $p$; {\cref{thm:adaptive_unknown_p}, restated}]
    Given query access to the oracle $\cA$ and its inverse $\cA^\dagger$ in \cref{eq:oracle_A}, for every fixed $b > 2$, there is a quantum algorithm that requires no information about $p$ and outputs the target state $\ket{1}\ket{\psi_1}$.
    If $Q_p\rbra{b}$ denotes its expected total number of oracle calls to $\cA$ (and $\cA^{\dagger}$), then
    \begin{equation*}
        \lim_{p \to 0}\sqrt{p}\, Q_p\rbra{b} = C\rbra{b} = \frac{b\sqrt{\pi}\Gamma\rbra*{\frac{b+2}{4}}^2\Gamma\rbra*{\frac{b-2}{2}}}{\rbra{b-1}\Gamma\rbra*{\frac{b}{4}}^2\Gamma\rbra*{\frac{b-1}{2}}},
    \end{equation*}
    where $\Gamma$ is the Gamma function~\cite[Chapter~12.4]{WW21}.
\end{theorem}
This gives an explicit leading constant for a weakly measured search loop that requires no prior knowledge of $p$.
In particular, $C\rbra{4} = 2\pi/3$, and numerical minimisation gives a minimum near $b \approx 5.2$, where $C\rbra{5.2} \approx 2.02$.
In Grover iterations, the leading constant is $C\rbra{b}/2$, approximately $1.01$ at this choice.
The oracle constant $C\rbra{5.2}$ is about $1.28$ times the $\pi/2$ leading oracle constant for standard Grover amplification when $p$ is known.
Since the standard Grover assumes knowledge of $p$, the more relevant comparison in the unknown-$p$ setting is with algorithms that also achieve $O\rbra{1/\sqrt{p}}$ without such knowledge.
Many weak-measurement, dissipative, and related search schemes already achieve this scaling; \emph{the distinction here is the leading constant factor analysis and the closed-form oracle constant for this discrete weakly measured loop}.

\subsection{Techniques and proof overview}
All three analyses exploit the fact that a failed weak measurement preserves the two-dimensional Grover geometry.
If the state immediately before measurement is
\begin{equation*} 
    \ket{\phi_{\alpha}}
    = \cos\rbra{\alpha}\ket{0}\ket{\psi_0} + \sin\rbra{\alpha}\ket{1}\ket{\psi_1},
\end{equation*}
then the failure branch of a weak measurement of strength $0 < \kappa < 1$ has angle
\begin{equation*}
    \beta = \arctan\rbra*{\sqrt{1-\kappa}\tan\rbra{\alpha}}.
\end{equation*}
A round of the algorithm is therefore a one-dimensional return map on the Grover angle, together with a probability of being killed by a successful outcome.
Expected oracle costs can consequently be written as sums of survival probabilities.

\paragraph{Exact return maps.}
For known $p$, we prescribe two angles
$0 < \beta < \alpha < \pi/2$ and invert the failure map for the angle:
\begin{equation*}
    \kappa\rbra{\alpha,\beta}
    =
    1-\frac{\tan^2\rbra{\beta}}{\tan^2\rbra{\alpha}}.
\end{equation*}
Choosing $\alpha-\beta = 2\theta$ makes one Grover iterate return the failure branch from $\beta$ to $\alpha$.
The loop then has a geometric survival law, so its expected cost is obtained exactly.
A continuous relaxation of this checkpoint construction yields the lower bound
\begin{equation*}
    E\rbra{\alpha,\beta} > \frac{\pi}{4}+\frac{1}{2},
\end{equation*}
whose infimum is approached as
$\alpha,\beta \to \pi/4$.
Here $E\rbra{\alpha,\beta}$ is the expected rotation measured from angle $0$ in the continuous relaxation.
Rounding the continuous checkpoint to the discrete Grover grid gives the exact schedule and the oracle bound in \cref{thm:exact_known_p}.

\paragraph{A Lyapunov equation for fixed strength.}
For the lower-bounded regime, exact angle tracking is unavailable.
A failed measurement followed by a Grover rotation acts as $R_\theta D_\kappa$ on the physical state, where $R_\theta$ is the Grover rotation and $D_\kappa$ is the measurement operator of the failure branch of the weak measurement.
For the analysis, we use coordinates rotated by $R_\theta^{-1}$, in which the unnormalised failure dynamics is
\begin{equation*}
    A_{\theta,\kappa} = D_\kappa R_\theta.
\end{equation*}
The prepared state has rotated coordinates $u_0 = \rbra*{\cos\rbra{\theta},-\sin\rbra{\theta}}^{\dagger}$.
The survival probability after $m$ failed measurements is $\Abs{A_{\theta,\kappa}^m u_0}^2$, so the expected cost is determined by
\begin{equation*}
    Q_{\theta,\kappa} = \sum_{m \geq 0}
    \rbra[\big]{A_{\theta,\kappa}^m}^{\dagger}A_{\theta,\kappa}^m.
\end{equation*}
The matrix $A_{\theta,\kappa}$ is Schur stable (that is, all eigenvalues lie inside the unit disc), and $Q_{\theta,\kappa}$ is the unique solution of the discrete Lyapunov equation~\cite{BCS18}
\begin{equation*}
    Q_{\theta,\kappa} - A_{\theta,\kappa}^{\dagger}
    Q_{\theta,\kappa}A_{\theta,\kappa} = I.
\end{equation*}
Solving this equation gives the full expected oracle cost $2u_0^\dagger Q_{\theta,\kappa}u_0-1$, including preparation.
The subtraction accounts for the absence of a Grover rotation after the successful measurement.
An arithmetic--geometric mean argument then gives the closed-form optimiser $\kappa^*\rbra{\theta}$.
Substituting the worst-case angle $\theta_0$ and exploiting monotonicity in the true angle yields the lower-bounded \cref{thm:fixed_point_lower_bounded}.
\paragraph{A continuous limit for unknown $p$.}
For a time-dependent strength, the survival probability is a product of conditional failure probabilities.
Combining this product with the angle recurrence identifies $\kappa\rbra{t} = \Theta\rbra{1/t}$ as the critical scale for a constant probability of success within $\floor*{\pi/\rbra{8\theta}}$ rounds as $p \to 0$, among the three asymptotic regimes treated by \cref{lem:kappa_order}.
This motivates $\kappa_b\rbra{t} = \min\cbra*{\frac{1}{2},\frac{b}{t}}$.

To determine the leading constant, we fix $b > 2$ and rescale discrete time by $\tau = t\theta$.
As $\theta \to 0$, the unnormalised two-dimensional recurrence $u_{t+1} = A_{\theta, \kappa_b\rbra{t+1}}u_t$ converges to the singular differential equation
\begin{equation*}
    z'\rbra{\tau} = \begin{pmatrix}
        0 & -2 \\
        2 & -\dfrac{b}{2\tau}
    \end{pmatrix}
    z\rbra{\tau} \quad \text{with} \quad z\rbra{0} = \begin{pmatrix}
        1\\0
    \end{pmatrix}.
\end{equation*}
The proof has three components.
We first construct the regular solution at the singular point $\tau = 0$.
We then prove uniform convergence of the discrete dynamics to this solution on compact time intervals.
Finally, an algebraic tail bound permits extension from compact intervals to the entire half-line.
Consequently,
\begin{equation*}
    \lim_{\theta \to 0}\theta \sum_{t \geq 0} \Abs*{u_t}^2 = \int_{0}^{\infty} \Abs*{z\rbra{\tau}}^2\mathrm{d}\tau.
\end{equation*}
Eliminating the second component of the differential equation gives a second-order differential equation.
We solve it in Fourier space and evaluate the
limiting integral using the Parseval--Plancherel identity together with Beta- and Gamma-function identities.
This calculation produces the closed form $C\rbra{b}$ in \cref{thm:adaptive_unknown_p}.

\subsection{Related work}
Grover search~\cite{Grover96} and amplitude amplification~\cite{BHM+02} provide the basic rotation framework used throughout this paper.
Exact amplification methods adjust the phases of reflections when the target probability is known~\cite{Hoyer00,Long01,BHM+02}.
For unknown target probability, Boyer, Brassard, H{\o}yer, and Tapp introduced randomised iteration counts and restart strategies that retain the quadratic speedup~\cite{BBH+98}.
Fixed-point amplification instead seeks robustness over a promised interval $p \geq p_0$.
Earlier fixed-point approaches combine amplification with measurements~\cite{AC12}.
The optimal phase schedule of Yoder, Low, and Chuang~\cite{YLC14} and later formulations through quantum singular-value transformation~\cite{GSL+19} achieve robustness by modifying the reflection phases.
In contrast, the algorithms in this paper keep the unitary part fixed to the ordinary Grover rotation and use the measurement strength, rather than reflection phases or stopping times, as the control design.

Our work belongs to a line in which measurement, damping, or open-system evolution is used as part of the search mechanism.
\emph{Critically damped quantum search}~\cite{Mizel09} introduced damping into Grover dynamics and chose the damping strength through a spectral analysis.
In contrast, the optimisation in the proof of \cref{thm:fixed_point_expected_cost} shows that the expected-cost-optimal strength is not the critically damped value, where the eigenvalues are all real~\cite{Mizel09}, but a strength for which the relevant eigenvalues remain complex.
The closest prior work to our unknown-$p$ result is Wang's dissipative quantum search, developed in Chapter~6 of \emph{Computational Problems Related to Open Quantum Systems}~\cite{Wang18}.
Wang gives an open-system search procedure that does not require prior knowledge of $p$ and proves $O\rbra{1/\sqrt{p}}$ query complexity for fixed error, but does not evaluate the asymptotic expected-cost constant considered here.
These works established that controlled nonunitary evolution can remove overcooking without destroying the quadratic speedup.

Weakly measured loops have also been studied from the viewpoint of quantum programming.
\emph{Weakly measured while loops: peeking at quantum states}~\cite{AMH22} investigates while loops in which intermediate weak measurements reveal information without fully collapsing the program state, while preserving the quantum speedup.
That work does not claim an algorithmic advantage over restart-based loops.
Its contribution is instead to the semantics and analysis of weakly measured quantum loops.
\emph{Differentiable quantum programming with unbounded loops}~\cite{FYW23} develops automatic differentiation methods for such quantum loop programs and uses gradient-based optimisation to search for effective measurement strengths.
Its numerical experiments indicate that weak-measurement-assisted Grover loops can outperform the standard Grover constant in suitable regimes.
The distinction of the present work is the exact strength optimisation and leading constant factor analysis for the weakly measured loops, rather than the existence of weakly measured loops themselves.

There is also a broader open-system and Zeno-effect literature showing that nonunitary dynamics can support search speedups.
Avron, Fraas, Graf, and Grech~\cite{AFG+10} studied dephasing Lindblad evolutions and showed that, with suitable scheduling and dephasing rates, Grover scaling can be recovered in a dephasing adiabatic search setting.
Pyshkin et al.~\cite{PGL+22} studied a nonunitary continuous-time Grover variant based on frequent Zeno-type measurements, deriving analytical bounds and discussing a tunable tradeoff between quantum and classical behavior.
More recently, \emph{Grover Speedup from Many Forms of the Zeno Effect}~\cite{BCD24} showed that several analogue Zeno mechanisms, including measurement, decoherence, and decay, can support Grover-like $\sqrt{N}$ speedups.
These results support the same broad viewpoint as ours: measurement and dissipation need not be treated only as noise.
They can also be algorithmic resources.
The models, however, are mostly continuous-time or analogue models, and there is no explicit constant factor indicating oracle savings over standard Grover amplification.

Our work also fits into the general theory of weak measurement.
Since the weak-measurement and weak-value formalism of Aharonov, Albert, and Vaidman~\cite{AAV88}, weak measurement has been viewed as a tunable tradeoff between information-gain and disturbance: one obtains partial information about a system while disturbing the state only partially.
This tradeoff has been quantified in several forms~\cite{FP96,FJ01,Busch09,CL12} and examined in the context of gentle measurements~\cite{AR19}.
Weak measurements also appear naturally in quantum feedback control~\cite{LJ00}, and in the theory of quantum trajectories, where measurement-conditioned evolution provides a stochastic description of open-system dynamics~\cite{B02}.
More broadly, weak measurements and measurement-conditioned nonunitary updates have recently appeared in algorithmic primitives beyond search.
In Gibbs sampling, weak measurement, for simulating a Lindbladian, has been used to implement quantum Metropolis-type updates that extract partial energy information without fully destroying coherence~\cite{JS24}.
Related developments appear in algorithms based on quantum trajectories and repeated-interaction models for simulating Lindbladian dynamics, where measurement-like updates are interleaved with coherent evolution to realize nonunitary processes efficiently~\cite{PSW25,BM25,BM26}.
These approaches share a common principle: weak measurements are used not merely for readout, but as an integral computational resource that allows partial progress to be retained and refined across iterations.
Our work brings this perspective into amplitude amplification in the discrete oracle model, providing explicit strength schedules and closed-form analyses of oracle cost in the known-$p$, lower-bounded-$p$, and unknown-$p$ regimes.

\subsection{Discussion and outlook}
The results in this paper concern amplitude amplification, but the underlying message is broader.
We use weak measurements as a control mechanism for an iterative quantum process whose coherent dynamics would otherwise over-rotate.
In the Grover setting, this control problem is one-dimensional after restricting to the usual invariant plane, and the measurement strength can be chosen explicitly.
This raises the natural question of whether similar adaptive measurement schedules can be used in other quantum algorithms with a repeated rotation, hitting, or amplification structure.

One immediate direction is quantum walk search.
Many quantum walk algorithms can be understood through an effective low-dimensional dynamics involving a marked subspace, a spectral gap, or a hitting-time scale.
When this scale is known, one can often choose a suitable running time; when it is unknown, the situation is analogous to amplitude amplification with unknown $p$.
It would be interesting to ask whether weak measurements of the marked subspace can replace or improve the usual choice of stopping time.
In particular, one may hope for adaptive measurement schedules that depend only on the number of failures and not on the unknown spectral gap or marked-state overlap.

A second direction is to study loops of quantum channels more systematically.
In the fixed-strength case, the expected cost is governed by a Lyapunov equation for the failure map.
In the adaptive case, the leading constant is obtained from a continuous limit of a time-dependent product of nonunitary maps.
These two tools suggest a general framework for analysing expected hitting times of quantum-channel loops, beyond the special two-dimensional Grover plane.

A third direction is optimal schedule design.
For unknown $p$, we identify $\kappa(t)=\Theta\rbra{1/t}$ as the critical scale within the regular schedules of \cref{lem:kappa_order}, and evaluate the leading constant for the family $\kappa_b(t)=\min\cbra{\frac{1}{2},\frac{b}{t}}$.
This family is simple and analytically tractable, but we do not prove that it is globally optimal among all adaptive schedules.
A more complete theory would formulate the choice of measurement strengths as an optimal-control problem for the limiting dynamics, possibly allowing time-dependent schedules more general than $b/t$.

Finally, our work suggests a different way to think about quantum algorithm design with intermediate measurements.
Weak measurements need not be viewed only as noise, nor only as a way to extract partial information.
They can also shape the future trajectory of the quantum state.
In amplitude amplification, this turns Grover search from a fixed-length rotation into a measurement-conditioned loop with provable constant-factor oracle savings in the known-$p$ regime and at the promise boundary $p = p_0$ in the lower-bounded regime.
Understanding when similar savings are available for other quantum algorithms remains an open problem.

\subsection{Organisation}
The remainder of the paper is organised as follows.
\cref{sec:pi_over_4_and_weak_measurement} illustrates the $\pi/4$ growth observation and introduces the weak measurement used throughout the paper.
\cref{sec:exact_for_known_p} studies the known-$p$ case, first through a continuous-angle heuristic and then through an exact discrete schedule.
\cref{sec:fixed_point_for_lower_bounded_p} derives the expected cost for a fixed strength, optimises that strength, and establishes the expected-cost bound for $p\geq p_0$.
Finally, \cref{sec:adaptive_for_unknown_p} treats completely unknown $p$, develops an adaptive weak measurement schedule, identifies the critical $\Theta(1/t)$ scale within the regular schedule classes considered there, evaluates the asymptotic constant $C(b)$, and compares the resulting schedule numerically with previous weakly measured and dissipative search procedures.

\section{The \texorpdfstring{$\pi/4$}{π/4} observation and weak measurements}\label{sec:pi_over_4_and_weak_measurement}
Quantum amplitude amplification rotates the initial state $\ket{\phi_\theta} = \cA\ket{0^n}$ from the angle $\theta$ to nearly $\pi/2$, achieving the target state with high probability after $O(1/\sqrt{p}) = O(1/\theta)$ iterations.
More precisely, the number of iterations required is approximately $\pi/(4\theta) - 1/2$. Consequently, one can apply $G$ approximately $\pi/(4\theta)$ times to bring the state close to $\ket{1}\ket{\psi_1}$.

\paragraph{The \texorpdfstring{$\pi/4$}{pi/4} growth region.}
Now consider the state after $k$ iterations:
\begin{equation*}
    G^k\cA\ket{0^n} = \ket{\phi_{\rbra{2k+1}\theta}}.
\end{equation*}
Let $P(k) = \sin^2\rbra{\rbra{2k+1}\theta}$ denote the probability of measuring the target state after $k$ iterations.
Extending $k$ to a real variable $x$, we consider the smooth function
\begin{equation*}
    P(x) = \sin^2\rbra{\rbra{2x+1}\theta}, \qquad x \in \sbra*{0, \tfrac{\pi}{4\theta} - \tfrac{1}{2}}.
\end{equation*}
Its derivative is $P'(x) = 2\theta\sin\rbra{2\rbra{2x+1}\theta}$.
As illustrated in \cref{fig:probability_derivative}, as $x$ increases from $0$ to $\frac{\pi}{4\theta}-\frac{1}{2}$, $P'(x)$ first increases, reaches a maximum, and then decreases.
The maximum occurs at the point $x^*$ satisfying $(2x^*+1)\theta = \pi/4$.
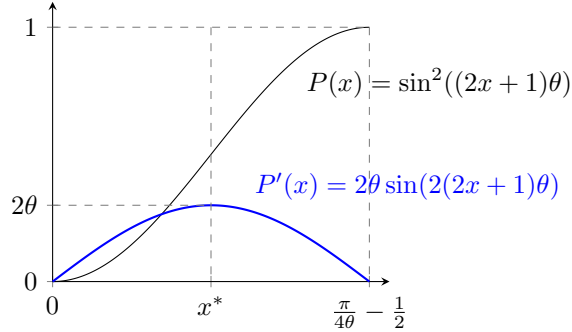
\begin{figure}[htb]
    \centering
    \begin{tikzpicture}
        \begin{axis}[
            scale=.65,
            clip=false,
            xmin=0,xmax=pi/2+.1,
            ymin=0,ymax=1.1,
            axis lines=left,
            ytick={0,0.3,1},
            yticklabels={$0$, $2\theta$, $1$},
            xtick={0,pi/4,pi/2},
            xticklabels={$0$,$x^*$,$\tfrac{\pi}{4\theta} - \tfrac{1}{2}$}]
            \addplot[domain=0:pi/2,samples=200] {sin(deg(x))^2} node[below right,pos=0.8,overlay]{$P(x) = \sin^2\rbra{\rbra{2x+1}\theta}$};
            \addplot[thick,domain=0:pi/2,samples=200,blue] {0.3*sin(2*deg(x))} node[above right,pos=0.6,overlay]{$P'(x) = 2\theta\sin\rbra{2\rbra{2x+1}\theta}$};
            \begin{scope}[dashed,gray]
                \draw (axis cs:pi/4, 1) -- (axis cs:pi/4,0);
                \draw (axis cs:pi/2, 1) -- (axis cs:pi/2,0);
                \draw (axis cs:0, 1) -- (axis cs:pi/2,1);
                \draw (axis cs:0, 0.3) -- (axis cs:pi/4,0.3);
            \end{scope}
        \end{axis}
    \end{tikzpicture}
    \caption{The target probability $P(x)$ and its derivative $P'(x)$, whose maximum $x^*$ satisfies $(2x^*+1)\theta = \pi/4$.}\label{fig:probability_derivative}
\end{figure}

Thus, the Grover operator has the greatest impact on the target probability when the angle is near $\pi/4$, i.e., when $(2x+1)\theta \approx \pi/4$.
When the state is far from this angle, the effect of each Grover iteration becomes significantly less beneficial.
\emph{This observation suggests keeping the angle close to $\pi/4$ throughout the amplification process by weak measurements~\cite{B02}, thereby maximising the effectiveness of each application of $G$}.

\paragraph{Limit the angle to near $\pi/4$ using weak measurements.}
Our weak measurement always acts on the first qubit and has the form
\begin{equation}
    M = \cbra*{M_0 = \ketbra{0}{0}+\sqrt{1-\kappa}\ketbra{1}{1}, M_1 = \sqrt{\kappa}\ketbra{1}{1}}, \qquad 0<\kappa \leq 1,
\end{equation}
where $\kappa$ controls the measurement strength.
For a state $\ket{\phi_\alpha}$ with angle $\alpha \in [0,\pi/2]$:
\begin{itemize}
    \item the $M_1$ branch corresponds to a successful measurement, occurring with probability
    \begin{equation}\label{eq:success_p}
        p_1 = \bra{\phi_\alpha}M_1^{\dagger}M_1\ket{\phi_\alpha} = \kappa\sin^2\rbra{\alpha},
    \end{equation}
    collapsing the state into the desired state $M_1\ket{\phi_\alpha}/\sqrt{p_1} = \ket{1}\ket{\psi_1}$, and terminating the process;
    \item the $M_0$ branch corresponds to a failure, occurring with probability
    \begin{equation}\label{eq:failure_p}
        p_0 = \bra{\phi_\alpha}M_0^{\dagger}M_0\ket{\phi_\alpha} = 1-\kappa\sin^2\rbra{\alpha},
    \end{equation}
    and leading to the post-measurement state
    \begin{equation}\label{eq:failure_state}
        M_0\ket{\phi_\alpha}/\sqrt{p_0} = \ket*{\phi_{\arctan\rbra*{\sqrt{1-\kappa}\tan\rbra{\alpha}}}} \propto \cos\rbra{\alpha}\ket{0}\ket{\psi_0}+\sqrt{1-\kappa}\sin\rbra{\alpha}\ket{1}\ket{\psi_1}.
    \end{equation}
    Unlike projective measurements, the state does not collapse directly to  $\ket{0}\ket{\psi_0}$.
    Instead, if $\kappa$ is small, it retains useful geometric information about the original angle $\alpha$.
\end{itemize}

To illustrate the benefit of weak measurements as simply as possible, suppose we have rotated the angle to $\pi/3$ using Grover iterations and choose $\kappa = 2/3$.
The success branch $M_1$ then occurs with probability $\kappa\sin^2\rbra{\pi/3} = 1/2$, while the failure branch $M_0$ also occurs with probability $1/2$ and, via \cref{eq:failure_state}, results in restoring the state with angle $\pi/4$, since
\begin{equation*}
    \ket{\phi_{\pi/4}} \propto \cos\rbra{\pi/3}\ket{0}\ket{\psi_0}+\sqrt{1/3}\sin\rbra{\pi/3}\ket{1}\ket{\psi_1} = \frac{1}{2}\ket{0}\ket{\psi_0} + \frac{1}{2}\ket{1}\ket{\psi_1}.
\end{equation*}
If we repeat the above procedure at the angle $\pi/3$ until success, and use the rotation angle as a cost metric, the expected cost follows a geometric series:
The first attempt costs $\pi/3$, and each subsequent failure requires rotating by an additional angle of $\pi/3 - \pi/4 = \pi/12$ using Grover iterations to restore the checkpoint before retrying.
Hence, the expected cost is
\begin{equation}\label{eq:cost_pi_over_3}
    \frac{\pi}{3}+\sum_{m\geq 1}\rbra*{\frac{1}{2}}^m \frac{\pi}{12} = \frac{5\pi}{12} = \frac{5}{6}\cdot \frac{\pi}{2}.
\end{equation}
This already improves upon the standard $\pi/2$ cost associated with rotating the state from near $0$ to nearly $\pi/2$ using Grover iterations alone.
This example is simple, yet it accurately demonstrates the advantage of employing weak measurements in amplitude amplification.

\section{Exact schedule for known \texorpdfstring{$p$}{p}}\label{sec:exact_for_known_p}
We start with the most favorable setting, where the target probability $p$ is known exactly.

\subsection{The continuous-angle heuristic}\label{sec:continuous_heuristic}
{We now generalise the construction in \cref{sec:pi_over_4_and_weak_measurement} by allowing the checkpoint and failure angles to vary continuously, while using the same measurement strength on each retry.
This heuristic measures the initial rotation from angle $0$ to identify the leading cost as $\theta \to 0$.
Suppose more generally that the state has been rotated to an angle $\alpha \in (0,\pi/2]$, and the weak measurement is chosen so that the failure branch results in a state at a smaller angle $\beta \in \rbra{0, \alpha}$.
At $\alpha = \pi/2$, a projective measurement succeeds with certainty, with rotation cost $\pi/2$; below we therefore consider $\alpha < \pi/2$.
Then \cref{eq:failure_state} gives $\tan\rbra{\beta} = \sqrt{1-\kappa}\tan\rbra{\alpha}$, and  the required measurement strength $\kappa$ is
\begin{equation}\label{eq:kappa_alpha_beta}
    \kappa\rbra{\alpha, \beta} = 1 - \frac{\tan^2\rbra{\beta}}{\tan^2\rbra{\alpha}}.
\end{equation}
Using \cref{eq:failure_p,eq:success_p}, the probabilities of failure and success branches become
\begin{equation*}
    p_0\rbra{\alpha,\beta} = 1-\kappa\rbra{\alpha,\beta}\sin^2\rbra{\alpha} = 1- \sin^2\rbra{\alpha}+\frac{\cos^2\rbra{\alpha}\sin^2\rbra{\beta}}{\cos^2\rbra{\beta}} = \frac{\cos^2\rbra{\alpha}}{\cos^2\rbra{\beta}} \text{ and }
    p_1\rbra{\alpha,\beta} = 1 - \frac{\cos^2\rbra{\alpha}}{\cos^2\rbra{\beta}},
\end{equation*}
respectively.
Analogously to \cref{eq:cost_pi_over_3}, the expected rotation angle becomes
\begin{equation}\label{eq:cost_alpha_beta}
    E\rbra{\alpha, \beta} = \alpha + \sum_{m \geq 1} \rbra{p_0\rbra{\alpha,\beta}}^m\rbra{\alpha - \beta} = \alpha + \frac{p_0\rbra{\alpha,\beta}}{p_1\rbra{\alpha,\beta}}\rbra{\alpha - \beta} = \alpha + \frac{\cos^2\rbra{\alpha}\rbra{\alpha - \beta}}{\cos^2\rbra{\beta} - \cos^2\rbra{\alpha}}.
\end{equation}
Since $\cos^2\rbra{\beta} - \cos^2\rbra{\alpha} = \sin^2\rbra{\alpha}\cos^2\rbra{\beta} - \cos^2\rbra{\alpha}\sin^2\rbra{\beta} = \sin\rbra{\alpha+\beta}\sin\rbra{\alpha - \beta}$, we have
\begin{equation*}
    E\rbra{\alpha, \beta} = \alpha + \frac{\cos^2\rbra{\alpha}\rbra{\alpha - \beta}}{\sin\rbra{\alpha+\beta}\sin\rbra{\alpha - \beta}} > \alpha + \frac{\cos^2\rbra{\alpha}}{\sin\rbra{\alpha+\beta}},
\end{equation*}
where the inequality follows from $\rbra{\alpha-\beta}/{\sin\rbra{\alpha - \beta}} > 1$ for $\alpha-\beta > 0$.
We now distinguish two cases:
\begin{itemize}
    \item If $0 < \alpha \leq \pi/4$, then $0 < \alpha+\beta < 2\alpha \leq \pi/2$, so $E\rbra{\alpha,\beta} > \alpha+\frac{\cos^2\rbra{\alpha}}{\sin\rbra{2\alpha}} = \alpha+\frac{1}{2}\cot\rbra{\alpha}$. Define $f\rbra{\alpha} = \alpha+\frac{1}{2}\cot\rbra{\alpha}$. Its derivative is $f'\rbra{\alpha} = 1 - \frac{1}{2}\csc^2\rbra{\alpha} \leq 0$ for $0 < \alpha \leq \pi/4$, so $f$ is decreasing on this interval. Thus $f\rbra{\alpha} \geq f\rbra{\pi/4} = \pi/4+1/2$. Hence $E\rbra{\alpha,\beta} > \pi/4+1/2$.
    \item If $\pi/4 \leq \alpha < \pi/2$, then $\sin\rbra*{\alpha+\beta} \leq 1$, so $E\rbra{\alpha,\beta} > \alpha + \cos^2\rbra{\alpha}$. Define $g\rbra{\alpha} = \alpha+\cos^2\rbra{\alpha}$. Its derivative is $g'\rbra{\alpha} = 1 - \sin\rbra{2\alpha} \geq 0$ for $\pi/4 \leq \alpha \leq \pi/2$, so $g$ is increasing on this interval. Hence $g\rbra{\alpha} \geq g\rbra{\pi/4} = \pi/4+1/2$, and again $E\rbra{\alpha, \beta} > \pi/4+1/2$.
\end{itemize}
Therefore, for all admissible $\rbra{\alpha, \beta}$,
\begin{equation}\label{eq:expected_angle_lower_bound}
    E\rbra{\alpha,\beta} > \frac{\pi}{4}+\frac{1}{2}.
\end{equation}
Taking $\alpha = \pi/4$ and letting $\beta \to \alpha^-$, we obtain the limiting infimum
\begin{equation*}
    E\rbra{\pi/4,\beta} \to \frac{\pi}{4}+\frac{1}{2}.
\end{equation*}
Normalised by the standard Grover rotation angle $\pi/2$, this corresponds to
\begin{equation}\label{eq:inf_ratio}
    \frac{\pi/4+1/2}{\pi/2} = \frac{2+\pi}{2\pi} \approx 0.818309886,
\end{equation}
which implies that employing weak measurements reduces the limiting expected rotation angle by roughly $18\%$.
This calculation is motivational and validates our $\pi/4$ observation: rotating the state into the neighbourhood of $\pi/4$ is beneficial and near-optimal within this family of repeated checkpoint measurements.

\subsection{The exact weak measurement schedule}
\label{sec:exact_schedule}

Building on the analysis in \cref{sec:continuous_heuristic}, we now present \cref{algo:exact_known_p} for the case where the target probability $p$ is known.
\cref{thm:exact_known_p} gives the oracle count of this algorithm.

\begin{algorithm*}[ht]
    \caption{Quantum amplitude amplification with an exact weak measurement schedule for known $p$.}
    \label{algo:exact_known_p}
    \begin{algorithmic}[1]
        \Require A quantum oracle $\cA$ (with inverse $\cA^{\dagger}$) satisfying \cref{eq:oracle_A}, along with a target probability $0 < p \leq \sin^2\rbra{\frac{\pi}{8}}$.

        \Ensure The target state $\ket{1}\ket{\psi_1}$.
        
        \State Let $G$ be the Grover rotation as in \cref{eq:grover_operator}.

        \State $\theta \gets \arcsin\rbra{\sqrt{p}}$, $K \gets \ceil*{\frac{\pi}{8\theta} - \frac{1}{2}}$, $\alpha^* \gets \rbra{2K+1}\theta$, $\beta^* \gets \rbra{2K-1}\theta$, $\kappa^* \gets 1 - \frac{\tan^2\rbra{\beta^*}}{\tan^2\rbra{\alpha^*}}$.

        \State Initialise an $n$-qubit register to the state $\cA\ket{0^n}$.
        \State Apply $G^{K}$ to the register. \Comment{Register in state $\ket{\phi_{\alpha^*}}$}

        \Repeat
            \State Perform the weak measurement $\cbra*{M_0 = \ketbra{0}{0}+\sqrt{1-\kappa^*}\ketbra{1}{1}, M_1 = \sqrt{\kappa^*}\ketbra{1}{1}}$ on the first qubit.
            \If {outcome $= 1$}
                \State \Return \Comment{Success: register in state $\ket{1}\ket{\psi_1}$}
            \Else \Comment{Register in state $\ket{\phi_{\beta^*}}$}
                \State Apply $G$ to the register.
            \EndIf
        \Until{success}
    \end{algorithmic}
\end{algorithm*}

\begin{theorem}
    [Oracle counts for known $p$]\label{thm:exact_known_p}
    For any target probability $0 < p \leq \sin^2\rbra{\frac{\pi}{8}}$, let $\theta = \arcsin\rbra{\sqrt{p}}$.
    The expected number of calls to $\cA$ (and $\cA^{\dagger}$) required by \cref{algo:exact_known_p} is less than
    \begin{equation*}
        \rbra*{\frac{\pi}{4}+\frac{2\theta}{\sin\rbra{4\theta}}+\theta}\cdot \frac{1}{\theta}.
    \end{equation*}
\end{theorem}
\begin{proof}
    With $p \leq \sin^2\rbra{\frac{\pi}{8}}$, we have $K \geq 1$.
    The choice of $K$ gives $0 < \beta^* < \pi/4 \leq \alpha^* < \pi/4+2\theta \leq \pi/2$, so $0 < \kappa^* < 1$ and each weak measurement has a positive probability of success.
    Since the $\kappa^*$ in \cref{algo:exact_known_p} satisfies \cref{eq:kappa_alpha_beta} for $\alpha^*$ and $\beta^*$, the analysis of the expected cost in \cref{eq:cost_alpha_beta} applies here.
    The expected number of times the algorithm executes a Grover rotation $G$ is
    \begin{equation*}
        T_G = K + \sum_{m \geq 1} \rbra*{p_0\rbra{\alpha^*,\beta^*}}^m \cdot 1 = K + \frac{p_0\rbra{\alpha^*,\beta^*}}{p_1\rbra{\alpha^*,\beta^*}} = K + \frac{\cos^2\rbra{\alpha^*}}{\cos^2\rbra{\beta^*} - \cos^2\rbra{\alpha^*}},
    \end{equation*}
    where $p_0\rbra{\alpha^*,\beta^*} = \cos^2\rbra{\alpha^*}/\cos^2\rbra{\beta^*}$ and $p_1\rbra{\alpha^*,\beta^*} = 1-\cos^2\rbra{\alpha^*}/\cos^2\rbra{\beta^*}$ are the probabilities of obtaining outcomes $0$ and $1$ in the weak measurement, respectively; $K$ is the number of Grover rotations $G$ used to reach the state $\ket{\phi_{\alpha^*}}$ and for each failure branch, an additional $G$ is used to bring $\ket{\phi_{\beta^*}}$ back to $\ket{\phi_{\alpha^*}}$.
    Then, the same trick of $\cos^2\rbra{\beta^*} - \cos^2\rbra{\alpha^*} = \sin\rbra{\alpha^*+\beta^*}\sin\rbra{\alpha^* - \beta^*}$ gives
    \begin{equation*}
        T_G = K + \frac{\cos^2\rbra{\alpha^*}}{\sin\rbra{\alpha^*+\beta^*}\sin\rbra{\alpha^* - \beta^*}}.
    \end{equation*}
    Substituting $\alpha^* = \rbra{2K+1}\theta$ and $\beta^* = \rbra{2K-1}\theta$ gives
    \begin{align*}
        T_G = K + \frac{\cos^2\rbra{\rbra{2K+1}\theta}}{\sin\rbra{4K\theta}\sin\rbra{2\theta}}.
    \end{align*}
    From $K = \ceil*{\frac{\pi}{8\theta} - \frac{1}{2}}$, we have $\rbra{2K-1}\theta < \frac{\pi}{4} \leq \rbra{2K+1}\theta$ and $\frac{\pi}{2}-2\theta \leq 4K\theta < \frac{\pi}{2}+2\theta$.
    In particular, $\sin\rbra{4K\theta} \geq \cos\rbra{2\theta}$ and $K < \frac{\pi}{8\theta}+\frac{1}{2}$.
    If $K \leq \frac{\pi}{8\theta}$, then
    \begin{equation*}
        T_G \leq K+\frac{1}{\sin\rbra{4\theta}} \leq \frac{\pi}{8\theta}+\frac{1}{\sin\rbra{4\theta}}.
    \end{equation*}
    If $K > \frac{\pi}{8\theta}$, expand $\cos^2\rbra{\rbra{2K+1}\theta}$ to obtain
    \begin{equation*}
        T_G = K-\frac{1}{2}+\frac{1}{2\sin\rbra{4K\theta}\sin\rbra{2\theta}}+\frac{1}{2}\cot\rbra{4K\theta}\cot\rbra{2\theta}
        < \frac{\pi}{8\theta}+\frac{1}{\sin\rbra{4\theta}},
    \end{equation*}
    because $\cot\rbra{4K\theta} < 0$.
    Since each Grover rotation $G$ makes $2$ calls to $\cA$ (and $\cA^{\dagger}$), together with the initial call to prepare $\cA\ket{0^n}$, the total expected number of oracle calls is
    \begin{equation*}
        T_\cA = 2T_G+1 < \frac{\pi}{4\theta} + \frac{2}{\sin\rbra{4\theta}}+1 = \rbra*{\frac{\pi}{4}+\frac{2\theta}{\sin\rbra{4\theta}}+\theta}\cdot \frac{1}{\theta}.
        \qedhere
    \end{equation*}
\end{proof}

Since $\theta T_\cA = E\rbra{\alpha^*,\beta^*}$, the lower bound in \cref{eq:expected_angle_lower_bound} and the upper bound in \cref{thm:exact_known_p} imply $\theta T_\cA \to \pi/4+1/2$ as $\theta \to 0$.
\cref{thm:exact_known_p} provides a rigorous improvement in the leading oracle cost over the standard Grover approach for amplitude amplification~\cite{BHM+02}, for which the oracle count of $\cA$ (and $\cA^{\dagger}$) required to obtain the target state with high probability is approximately $\frac{\pi}{2}\cdot \frac{1}{\theta}$:
In the limit $p \to 0$ (or equivalently $\theta \to 0$), our weak measurement schedule improves by a factor of
\begin{equation}\label{eq:exact_ratio}
    \frac{\frac{\pi}{4}+\frac{2\theta}{\sin\rbra{4\theta}}+\theta}{\frac{\pi}{2}} \to \frac{2+\pi}{2\pi}.
\end{equation}
This is optimal, in the sense that it matches exactly the infimum ratio in the continuous heuristic of \cref{eq:inf_ratio}.

\paragraph{Comparison with restart Grover.}
A related strategy known as restart Grover appears in the literature~\cite[Section~3]{BBH+98}.
In that approach, one selects an integer $m < \frac{\pi}{4\theta}$, prepares the state $G^m\cA\ket{0^n}$, and then performs a projective measurement $\cbra{M_0 = \ketbra{0}{0},M_1 = \ketbra{1}{1}}$ on the first qubit.
If the outcome is $0$, the entire procedure is restarted from the beginning; if the outcome is $1$, the target state is obtained.
By choosing $m$ optimally, the expected number of oracle calls can be reduced compared to the standard Grover algorithm.
The asymptotic ratio of this expected cost to that of standard Grover is approximately $0.88$ (see \cite{BBH+98} for details), which is larger (i.e., worse) than the infimum $\frac{2+\pi}{2\pi} \approx 0.8183$ achieved by our exact weak measurement schedule (see \cref{eq:exact_ratio}).
In the continuous-angle limit, the restart Grover procedure corresponds to the boundary $\beta \to 0$ in \cref{sec:continuous_heuristic}, for which $E\rbra{\alpha,\beta} \to \alpha/\sin^2\rbra{\alpha}$.
The comparison above shows that \cref{algo:exact_known_p} has a smaller leading expected oracle cost in the small-$p$ regime.

\section{Fixed-point schedule for lower-bounded \texorpdfstring{$p$}{p}}\label{sec:fixed_point_for_lower_bounded_p}
We next move to the scenario where the target probability $p$ is not known exactly, but is known to satisfy a lower bound $p \geq p_0$ for some $p_0 > 0$.
In this regime, the exact weak measurement schedule developed in \cref{sec:exact_for_known_p} is no longer directly applicable, as the lack of precise knowledge of $p$ prevents us from reliably steering the state to an angle close to $\pi/4$.
Instead, we prepare $\cA\ket{0^n}$ and repeatedly perform a weak measurement of fixed strength $\kappa$.
A successful measurement returns the target state; after each failure, we apply one Grover rotation and continue.
The procedure is summarised in \cref{algo:fixed_point_lower_bounded_p}.
This general approach, which interleaves weak measurements with Grover iterations, has been explored in several works~\cite{Mizel09,Wang18,AMH22,FYW23}.

\begin{algorithm*}[ht]
    \caption{Quantum amplitude amplification with fixed-point weak measurement schedule}
    \label{algo:fixed_point_lower_bounded_p}
    \begin{algorithmic}[1]
        \Require A quantum oracle $\cA$ (and its inverse $\cA^{\dagger}$) satisfying \cref{eq:oracle_A} and a weak measurement strength $0 < \kappa \leq 1$.
        \Ensure The target state $\ket{1}\ket{\psi_1}$.
        \State Let $G$ be the Grover rotation as in \cref{eq:grover_operator}.
        \State Prepare an $n$-qubit register in the state $\cA\ket{0^n}$.
        \Repeat
            \State Perform the weak measurement $\cbra*{M_0 = \ketbra{0}{0}+\sqrt{1-\kappa}\ketbra{1}{1}, M_1 = \sqrt{\kappa}\ketbra{1}{1}}$ on the first qubit.
            \If {outcome $= 1$}
                \State \Return \Comment{Success: register in state $\ket{1}\ket{\psi_1}$}
            \EndIf
            \State Apply $G$ to the register.
        \Until{success}
    \end{algorithmic}
\end{algorithm*}

It is clear that a poorly chosen measurement strength can degrade or even eliminate the quantum speedup.
For example, setting $\kappa = 1$ reduces the weak measurement to a projective measurement, giving expected oracle cost $1+1/\rbra{2p} = \Theta\rbra{1/p}$ as $p \to 0$.
Conversely, an appropriate choice of $\kappa$ can preserve the quadratic speedup.
In particular, Mizel~\cite{Mizel09} derived a suitable $\kappa$ via eigenvalue analysis, while subsequent works~\cite{AMH22,Wang18} proved that with a proper choice, the expected oracle complexity remains $O\rbra{1/\sqrt{p}}$.
The complexity matches the standard Grover algorithm but does not demonstrate any explicit advantage arising from the weak measurements themselves.
More recently, in~\cite{FYW23}, the near-optimal choice of $\kappa$ is found by gradient-based numerical search, and numerical experiments then show that the resulting weak-measurement-assisted protocol actually requires fewer oracle calls than standard Grover.
Nevertheless, a theoretical derivation that identifies the optimal measurement strength and rigorously quantifies the resulting advantage remains absent from the literature.
This section aims to address this gap by providing a fixed-point schedule with provable improvement over the conventional approach.

\subsection{Expected cost analysis for a given measurement strength}
\label{sec:expected_cost_lower_bounded}
We begin by deriving the exact expected number of calls to $\cA$ and $\cA^{\dagger}$ in \cref{algo:fixed_point_lower_bounded_p}, including the initial preparation, for an arbitrary fixed measurement strength $\kappa$.
This expression will later allow us to choose the optimal $\kappa$ that minimises the expected cost.

In the basis $\cbra{\ket{0}\ket{\psi_0},\ket{1}\ket{\psi_1}}$, the prepared state, Grover rotation, and failure operator of the weak measurement are represented by
\begin{equation*}
    \psi_\theta = \begin{pmatrix}
        \cos\rbra{\theta} \\
        \sin\rbra{\theta}
    \end{pmatrix}, \qquad
    R_\theta = \begin{pmatrix}
        \cos\rbra{2\theta} & -\sin\rbra{2\theta} \\
        \sin\rbra{2\theta} & \cos\rbra{2\theta}
    \end{pmatrix}, \qquad
    D_\kappa = \begin{pmatrix}
        1 & 0 \\
        0 & \sqrt{1-\kappa}
    \end{pmatrix}.
\end{equation*}
An unsuccessful round first measures and then rotates, so its unnormalised state evolves by $B_{\theta,\kappa} = R_\theta D_\kappa$.
To use a convenient Lyapunov representation, we also define
\begin{equation*}
    A_{\theta,\kappa} = D_\kappa R_\theta, \qquad
    u_0 = R_\theta^{\dagger}\psi_\theta = \begin{pmatrix}
        \cos\rbra{\theta} \\
        -\sin\rbra{\theta}
    \end{pmatrix}.
\end{equation*}
The identity $B_{\theta,\kappa} = R_\theta A_{\theta,\kappa}R_\theta^{\dagger}$ shows that
$B_{\theta,\kappa}^m\psi_\theta = R_\theta A_{\theta,\kappa}^m u_0$.
Thus $u_0$ is only an auxiliary coordinate vector for the analysis; the algorithm prepares $\psi_\theta$ directly.

Let $T \in \cbra{1,2,\ldots}\cup\cbra{\infty}$ denote the index of the first successful weak measurement.
The survival probability after $m$ unsuccessful measurements is
\begin{equation}\label{eq:m_fixed_point_failures_probability}
    S_m\rbra{\theta,\kappa} = \Pr\rbra{T > m}
    = \Abs*{B_{\theta,\kappa}^m\psi_\theta}^2
    = \Abs*{A_{\theta,\kappa}^m u_0}^2.
\end{equation}
Each unsuccessful measurement is followed by one Grover iteration, while the successful measurement ends the algorithm immediately.
Consequently, $N_G = T-1$ and the tail-sum identity gives
\begin{equation}\label{eq:cost_round_lower_bounded_p}
    \EE\sbra{N_G} = \sum_{m \geq 1}\Pr\rbra{T > m}
    = \sum_{m \geq 1}S_m\rbra{\theta,\kappa},\qquad
    Q_{\mathrm{total}}\rbra{\theta,\kappa}
    = 1+2\sum_{m \geq 1}S_m\rbra{\theta,\kappa}.
\end{equation}
Here the initial $1$ counts the preparation call.
The probability of success at measurement $m$ is $S_{m-1}-S_m$.

To establish convergence, we examine the eigenvalues of $A_{\theta,\kappa}$.
Since
\begin{equation*}
    \det\rbra{A_{\theta,\kappa}} = \sqrt{1-\kappa}, \qquad \tr\rbra{A_{\theta,\kappa}} = \rbra*{1+\sqrt{1-\kappa}}\cos\rbra{2\theta},
\end{equation*}
the characteristic polynomial of $A_{\theta,\kappa}$ is
\begin{equation}\label{eq:characteristic_polynomial}
    \chi\rbra{\lambda} = \lambda^2 - \rbra*{1+\sqrt{1-\kappa}}\cos\rbra{2\theta}\lambda + \sqrt{1-\kappa}.
\end{equation}
Assuming $0 < \kappa < 1$ and $\theta \in \rbra{0,\pi/4}$, we have
\begin{align*}
    \chi\rbra{1} &= 1+\sqrt{1-\kappa} - \rbra*{1+\sqrt{1-\kappa}}\cos\rbra{2\theta} > 0, \\
    \chi\rbra{-1} &= 1+\sqrt{1-\kappa} + \rbra*{1+\sqrt{1-\kappa}}\cos\rbra{2\theta} > 0, \\
    \abs{\chi\rbra{0}} &= \sqrt{1-\kappa} < 1.
\end{align*}
By the Jury stability criterion~\cite{Jury62}, all eigenvalues of $A_{\theta,\kappa}$ lie strictly inside the unit disc.\footnote{For a monic quadratic polynomial $\chi\rbra{\lambda} = \lambda^2+a_1\lambda +a_0$, if $\chi\rbra{\pm 1} > 0$ and $\abs{\chi\rbra{0}} < 1$, then the real roots must lie strictly between $-1$ and $1$.
If the roots are complex conjugates $\lambda_1 = \bar{\lambda}_2$, then $\abs{\lambda_1}^2 = \lambda_1\lambda_2 = \chi\rbra{0}$.
The condition $\abs{\chi\rbra{0}} < 1$ ensures $\abs{\lambda_1},\abs{\lambda_2} < 1$.}
Consequently, the expected cost in \cref{eq:cost_round_lower_bounded_p} is finite, and the matrix series
\begin{equation}\label{eq:cost_q}
    Q_{\theta,\kappa} = \sum_{m \geq 0} \rbra[\big]{A_{\theta,\kappa}^m}^{\dagger}A_{\theta,\kappa}^m
\end{equation}
also converges to the closed form given in the following lemma.

\begin{lemma}\label{lem:cost_q_closed_form}
    For any $0 < \kappa < 1$ and $\theta \in \rbra{0,\pi/4}$, the matrix series $Q_{\theta,\kappa}$ in \cref{eq:cost_q} converges and admits the closed form
    \begin{equation}
        Q_{\theta,\kappa} = \begin{pmatrix}
            \dfrac{2\sqrt{1-\kappa}}{\kappa} + \dfrac{\kappa}{\rbra*{1+\sqrt{1-\kappa}}^2\sin^2\rbra{2\theta}} & \dfrac{-\cos\rbra{2\theta}}{\rbra*{1+\sqrt{1-\kappa}}\sin\rbra{2\theta}} \\
            \dfrac{-\cos\rbra{2\theta}}{\rbra*{1+\sqrt{1-\kappa}}\sin\rbra{2\theta}} & \dfrac{2}{\kappa}
        \end{pmatrix}.
    \end{equation}
\end{lemma}
\begin{proof}
    For any $0 < \kappa < 1$ and $\theta \in \rbra{0,\pi/4}$, we have already shown that all eigenvalues of $A_{\theta,\kappa}$ have modulus less than $1$.
    Consequently, there exist constants $C > 0$ and $r \in \rbra{0,1}$ such that for all $m \geq 0$, $\Abs*{A_{\theta,\kappa}^m} \leq Cr^m$.
    It follows that
    \begin{equation*}
        \Abs[\big]{\rbra[\big]{A_{\theta,\kappa}^m}^{\dagger}A_{\theta,\kappa}^m} \leq \Abs[\big]{\rbra[\big]{A_{\theta,\kappa}^m}^{\dagger}}\Abs[\big]{A_{\theta,\kappa}^m} \leq C^2r^{2m}
    \end{equation*}
    and therefore the series defining $Q_{\theta,\kappa}$ converges absolutely:
    \begin{equation*}
        \sum_{m \geq 0} \Abs*{\rbra[\big]{A_{\theta,\kappa}^m}^{\dagger}A_{\theta,\kappa}^m} \leq C^2\sum_{m \geq 0} r^{2m} < \infty.
    \end{equation*}
    Multiplying the series on the left by $A_{\theta,\kappa}^{\dagger}$ and on the right by $A_{\theta,\kappa}$ shifts the index by one:
    \begin{equation*}
        A_{\theta,\kappa}^{\dagger}Q_{\theta,\kappa} A_{\theta,\kappa} = \sum_{m \geq 0} A_{\theta,\kappa}^{\dagger}\rbra[\big]{A_{\theta,\kappa}^m}^{\dagger}A_{\theta,\kappa}^mA_{\theta,\kappa} = \sum_{m \geq 1} \rbra[\big]{A_{\theta,\kappa}^m}^{\dagger}A_{\theta,\kappa}^m = Q_{\theta,\kappa} - I.
    \end{equation*}
    Hence $Q_{\theta,\kappa}$ satisfies the equation
    \begin{equation}\label{eq:q_equation}
        X - A_{\theta,\kappa}^{\dagger} X A_{\theta,\kappa} = I.
    \end{equation}
    Because all eigenvalues of $A_{\theta,\kappa}$ have modulus less than $1$, the solution to \cref{eq:q_equation} is unique.\footnote{Suppose that there are two solutions $X_1$ and $X_2$. Then their difference $\Delta = X_1 - X_2$ satisfies $\Delta = A_{\theta,\kappa}^{\dagger} \Delta A_{\theta,\kappa} = \rbra[\big]{A_{\theta,\kappa}^{\dagger}}^m \Delta A_{\theta,\kappa}^m$ for any $m \geq 1$. Taking $m \to \infty$ forces $\Delta = 0$.}

    To derive the entries of $Q_{\theta,\kappa}$ explicitly, we write
    \begin{equation*}
        Q_{\theta,\kappa} = \begin{pmatrix}
            x & y \\
            y & z
        \end{pmatrix}.
    \end{equation*}
    Substituting this into \cref{eq:q_equation} yields
    \begin{equation*}
        \begin{pmatrix}
            x & y \\
            y & z
        \end{pmatrix} -
        \begin{pmatrix}
        \cos\rbra{2\theta} & \sqrt{1-\kappa}\sin\rbra{2\theta} \\
        -\sin\rbra{2\theta} & \sqrt{1-\kappa}\cos\rbra{2\theta}
    \end{pmatrix}
    \begin{pmatrix}
            x & y \\
            y & z
        \end{pmatrix}
    \begin{pmatrix}
        \cos\rbra{2\theta} & -\sin\rbra{2\theta} \\
        \sqrt{1-\kappa}\sin\rbra{2\theta} & \sqrt{1-\kappa}\cos\rbra{2\theta}
    \end{pmatrix} =
    \begin{pmatrix}
        1 & 0 \\
        0 & 1
    \end{pmatrix}
    \end{equation*}
    Carrying out the matrix multiplications gives the following system of scalar equations:
    \begin{gather}
        \label{eq:q_equation_1} x\sin^2\rbra{2\theta} - 2y\sqrt{1-\kappa}\sin\rbra{2\theta}\cos\rbra{2\theta} - z\rbra{1-\kappa}\sin^2\rbra{2\theta} = 1, \\
        \label{eq:q_equation_2} x\sin\rbra{2\theta}\cos\rbra{2\theta} + y\rbra*{1-\sqrt{1-\kappa}\rbra{\cos^2\rbra{2\theta} - \sin^2\rbra{2\theta}}} - z\rbra{1-\kappa}\sin\rbra{2\theta}\cos\rbra{2\theta} = 0, \\
        \label{eq:q_equation_3} -x\sin^2\rbra{2\theta} + 2y\sqrt{1-\kappa}\cos\rbra{2\theta}\sin\rbra{2\theta}+z\rbra{1-\rbra{1-\kappa}\cos^2\rbra{2\theta}} = 1
    \end{gather}
    Adding \cref{eq:q_equation_1} and \cref{eq:q_equation_3} simplifies to $z\rbra{1-\rbra{1-\kappa}} = 2$, i.e.,
    \begin{equation*}
        z = \frac{2}{\kappa}.
    \end{equation*}
    Next, taking $\cos\rbra{2\theta}$ times \cref{eq:q_equation_1} minus $\sin\rbra{2\theta}$ times \cref{eq:q_equation_2} leaves $-y\rbra{1+\sqrt{1-\kappa}}\sin\rbra{2\theta} = \cos\rbra{2\theta}$, so
    \begin{equation*}
        y = \dfrac{-\cos\rbra{2\theta}}{\rbra*{1+\sqrt{1-\kappa}}\sin\rbra{2\theta}}.
    \end{equation*}
    Finally, substituting the expressions of $y$ and $z$ into \cref{eq:q_equation_2} gives
    \begin{equation*}
        x = \frac{2\rbra{1-\kappa}}{\kappa} + \dfrac{\rbra*{1-\sqrt{1-\kappa}\rbra{\cos^2\rbra{2\theta} - \sin^2\rbra{2\theta}}}}{\rbra*{1+\sqrt{1-\kappa}}\sin^2\rbra{2\theta}}.
    \end{equation*}
    Using $\cos^2\rbra{2\theta} - \sin^2\rbra{2\theta} = 1-2\sin^2\rbra{2\theta}$, this simplifies to
    \begin{align*}
        x &= \frac{2\rbra{1-\kappa}}{\kappa} + \frac{2\sqrt{1-\kappa}}{1+\sqrt{1-\kappa}} + \frac{1-\sqrt{1-\kappa}}{\rbra*{1+\sqrt{1-\kappa}}\sin^2\rbra{2\theta}} \\
        &= \frac{2\sqrt{1-\kappa}}{\kappa} + \frac{\kappa}{\rbra*{1+\sqrt{1-\kappa}}^2\sin^2\rbra{2\theta}}.
    \end{align*}
    Thus the claimed closed form for $Q_{\theta,\kappa}$ follows.
\end{proof}

\begin{remark}
    \cref{eq:q_equation} is a discrete-time Lyapunov equation for the linear system $x\rbra{t+1} = A_{\theta,\kappa}x\rbra{t}$.
    Such equations arise naturally in stability analysis, control theory, and the study of linear recurrences, see, for example, \cite[Section~3]{BCS18}.
    The uniqueness of the solution follows from the fact that $A_{\theta,\kappa}$ is Schur stable (that is, all eigenvalues lie inside the unit disc).
\end{remark}

The closed form in \cref{lem:cost_q_closed_form}, together with $S_0 = \Abs{u_0}^2 = 1$, gives
\begin{equation*}
    \EE\sbra{N_G} = u_0^{\dagger}Q_{\theta,\kappa}u_0-1,
    \qquad Q_{\mathrm{total}}\rbra{\theta,\kappa}
    = 2u_0^{\dagger}Q_{\theta,\kappa}u_0-1.
\end{equation*}

\begin{theorem}\label{thm:fixed_point_expected_cost}
    For any $0 < \kappa < 1$ and $\theta \in \rbra{0,\pi/4}$, put $p = \sin^2\rbra{\theta}$. The expected total number of calls to $\cA$ and $\cA^{\dagger}$ in \cref{algo:fixed_point_lower_bounded_p}, including preparation, is
    \begin{equation}\label{eq:fixed_point_expected_cost}
        Q_{\mathrm{total}}\rbra{\theta,\kappa}
        = \frac{\rbra*{1+\sqrt{1-\kappa}}^2}{\kappa}
        +\frac{\kappa}{2p\rbra*{1+\sqrt{1-\kappa}}^2}.
    \end{equation}
    For fixed $\theta$, this cost is minimised uniquely by
    \begin{equation}\label{eq:optimal_kappa}
        \kappa^*\rbra{\theta}
        = \frac{4\sqrt{2}\sin\rbra{\theta}}{\rbra*{1+\sqrt{2}\sin\rbra{\theta}}^2}
        = \frac{4\sqrt{2p}}{\rbra*{1+\sqrt{2p}}^2} \in \rbra{0,1},
    \end{equation}
    and the corresponding minimum is
    \begin{equation*}
        Q_{\mathrm{total}}\rbra{\theta,\kappa^*\rbra{\theta}}
        = \sqrt{\frac{2}{p}},\qquad
        \EE\sbra{N_G} = \frac{1}{\sqrt{2p}}-\frac{1}{2}.
    \end{equation*}
\end{theorem}
\begin{proof}
    Introduce the substitution
    \begin{equation}\label{eq:kappa_x_substitution}
        x = \frac{1-\sqrt{1-\kappa}}{1+\sqrt{1-\kappa}} \in \rbra{0,1},
        \qquad \kappa = \frac{4x}{\rbra{1+x}^2}.
    \end{equation}
    In this variable, the entries from \cref{lem:cost_q_closed_form} become
    \begin{equation*}
        Q_{\theta,\kappa} = \begin{pmatrix}
            \dfrac{1-x^2}{2x}+\dfrac{x}{4p\rbra{1-p}} & -\dfrac{\rbra{1-2p}\rbra{1+x}}{4\sqrt{p\rbra{1-p}}} \\[1ex]
            -\dfrac{\rbra{1-2p}\rbra{1+x}}{4\sqrt{p\rbra{1-p}}} & \dfrac{\rbra{1+x}^2}{2x}
        \end{pmatrix}.
    \end{equation*}
    Since $u_0 = \rbra*{\sqrt{1-p},-\sqrt{p}}^{\dagger}$, we obtain
    \begin{equation*}
        u_0^{\dagger}Q_{\theta,\kappa}u_0 = \rbra{1-p}\rbra*{\frac{1-x^2}{2x}+\frac{x}{4p\rbra{1-p}}} + \frac{\rbra{1-2p}\rbra{1+x}}{2}+\frac{p\rbra{1+x}^2}{2x} = \frac{1}{2}+\frac{1}{2x}+\frac{x}{4p}.
    \end{equation*}
    Thus
    \begin{equation*}
        Q_{\mathrm{total}}\rbra{\theta,\kappa}
        = 2u_0^{\dagger}Q_{\theta,\kappa}u_0-1
        = \frac{1}{x}+\frac{x}{2p},
    \end{equation*}
    which is \cref{eq:fixed_point_expected_cost}. By the arithmetic--geometric mean inequality,
    \begin{equation*}
        \frac{1}{x}+\frac{x}{2p} \geq \sqrt{\frac{2}{p}},
    \end{equation*}
    with equality if and only if
    \begin{equation}\label{eq:x_optimal}
        x = \sqrt{2p} = \sqrt{2}\sin\rbra{\theta}.
    \end{equation}
    This value lies in $\rbra{0,1}$ for $0 < p < 1/2$; substituting it into \cref{eq:kappa_x_substitution} gives \cref{eq:optimal_kappa}.
    Finally, $Q_{\mathrm{total}} = 1+2\EE\sbra{N_G}$ gives the claimed Grover iteration count.
\end{proof}

In \cite[Lemma~6.3]{Wang18}, the measurement strength is set to $\kappa = 1-e^{-\lambda}$ for $\lambda \in \sbra{2\sqrt{p},4\sqrt{p}}$.
For sufficiently small $p$, this yields $\kappa \sim \lambda \in \sbra{2\sqrt{p},4\sqrt{p}}$.
In contrast, our optimal strength $\kappa^*\rbra{\theta}$ from \cref{thm:fixed_point_expected_cost} behaves asymptotically as $\kappa^*\rbra{\theta} \sim 4\sqrt{2}\theta \sim 4\sqrt{2}\sqrt{p}$.
At this strength, the exact total expected oracle cost is $\sqrt{2/p}$. Its ratio to the leading standard Grover cost $\pi/\rbra{2\sqrt{p}}$ is $2\sqrt{2}/\pi \approx 0.9003$, a reduction of roughly $10\%$ in the leading constant.
Although this saving is more modest than the $18\%$ achieved by \cref{algo:exact_known_p} when $p$ is known exactly, \cref{algo:fixed_point_lower_bounded_p} itself does not require the actual value of $p$ as input.
By setting $\kappa = \kappa^*\rbra{\theta_0}$ with $\theta_0 = \arcsin\rbra{\sqrt{p_0}}$, we can obtain a universal schedule that works for any actual target probability $p \geq p_0$, as established in the following subsection.

\subsection{The fixed-point weak measurement schedule}
The previous theorem, \cref{thm:fixed_point_expected_cost}, provides the optimal measurement strength when $\theta$ is known exactly.
In particular, however, we may only have a guarantee that $p \geq p_0$ (equivalently $\theta \geq \theta_0$).
The following theorem gives the exact expected oracle cost when the strength is chosen from the lower bound, $\kappa = \kappa^*\rbra{\theta_0}$.
The cost is $O\rbra{1/\sqrt{p_0}}$ uniformly over $p \in [p_0,1/2)$ and decreases as the actual $p$ increases.

\begin{theorem}[Oracle counts for lower-bounded $p \geq p_0$]\label{thm:fixed_point_lower_bounded}
    For any $0 < p_0 \leq p < 1/2$, let $\theta_0 = \arcsin\rbra{\sqrt{p_0}}$ and choose $\kappa = \kappa^*\rbra{\theta_0}$ as in \cref{eq:optimal_kappa}. The expected total number of oracle calls required by \cref{algo:fixed_point_lower_bounded_p} is
    \begin{equation*}
        Q_{\mathrm{total}}\rbra{\theta,\kappa^*\rbra{\theta_0}}
        = \frac{1}{\sqrt{2p_0}}+\frac{\sqrt{p_0}}{\sqrt{2}\,p}
        \leq \frac{1}{\sqrt{2}}\rbra*{\frac{1}{\sqrt{p_0}}+\frac{1}{\sqrt{p}}}
        \leq \sqrt{\frac{2}{p_0}}.
    \end{equation*}
    The exact cost is strictly decreasing in $p$ for fixed $p_0$.
\end{theorem}
\begin{proof}
    The chosen strength corresponds to $x_0 = \sqrt{2p_0}$ in \cref{eq:kappa_x_substitution}. Substituting into \cref{eq:fixed_point_expected_cost} gives
    \begin{equation*}
        Q_{\mathrm{total}} = \frac{1}{x_0}+\frac{x_0}{2p}
        = \frac{1}{\sqrt{2p_0}}+\frac{\sqrt{p_0}}{\sqrt{2}\,p}.
    \end{equation*}
    The first inequality follows from $\sqrt{p_0}/p \leq 1/\sqrt{p}$, and the second from $p \geq p_0$. The term proportional to $1/p$ decreases strictly as $p$ increases.
\end{proof}

The conventional fixed-point construction of~\cite{YLC14} chooses a predetermined query count from $p_0$ and a target failure probability; that count does not decrease when the actual $p$ is larger.
This fixed-horizon guarantee and our repeat-until-success expected cost are different measures.
To examine the former, \cref{algo:truncated_fixed_point_lower_bounded_p} performs at most $m$ weak measurements, each followed by a Grover rotation on failure, and then measures projectively if no earlier measurement has succeeded.

\begin{algorithm*}[ht]
    \caption{Truncated quantum amplitude amplification with fixed-point weak measurement schedule}
    \label{algo:truncated_fixed_point_lower_bounded_p}
    \begin{algorithmic}[1]
        \Require A quantum oracle $\cA$ (and its inverse $\cA^{\dagger}$) satisfying \cref{eq:oracle_A}, a weak measurement strength $0 < \kappa \leq 1$, and a non-negative integer number $m$ of rounds.
        \Ensure The target state $\ket{1}\ket{\psi_1}$ on success, or a failure flag.
        \State Let $G$ be the Grover rotation as in \cref{eq:grover_operator}.
        \State Prepare an $n$-qubit register in the state $\cA\ket{0^n}$.
        \For {$j\gets1,\ldots,m$}
            \State Perform the weak measurement $\cbra*{M_0 = \ketbra{0}{0}+\sqrt{1-\kappa}\ketbra{1}{1},M_1 = \sqrt{\kappa}\ketbra{1}{1}}$ on the first qubit.
            \If {outcome $= 1$}
                \State \Return \Comment{Register in state $\ket{1}\ket{\psi_1}$}
            \EndIf
            \State Apply $G$ to the register.
        \EndFor
        \State Perform the projective measurement $\cbra*{M_0 = \ketbra{0}{0},M_1 = \ketbra{1}{1}}$ on the first qubit.
        \State \Return the target state if the outcome is $1$, and a failure flag otherwise.
    \end{algorithmic}
\end{algorithm*}

\begin{theorem}\label{thm:truncated_fixed_point_probability}
    For $0 < p_0 \leq p < 1/2$, let $\theta_0 = \arcsin\rbra{\sqrt{p_0}}$ and choose $\kappa = \kappa^*\rbra{\theta_0}$.
    For any non-negative integer $m$, the failure probability of \cref{algo:truncated_fixed_point_lower_bounded_p} satisfies
    \begin{equation}\label{eq:truncated_fixed_point_probability}
        P_{\mathrm{fail}}\rbra{m} < 2\rbra*{\frac{1-\sqrt{2p_0}}{1+\sqrt{2p_0}}}^{m} \leq 2\exp\rbra*{-2\sqrt{2p_0}\,m}.
    \end{equation}
    In particular, for any desired failure probability $\epsilon \in \rbra{0,1/2}$,
    \begin{equation}\label{eq:truncated_rounds}
        m \geq \frac{\ln\rbra{2/\epsilon}}{2\sqrt{2p_0}}
    \end{equation}
    suffices to guarantee success probability at least $1-\epsilon$ uniformly for $p \in [p_0,1/2)$.
    The algorithm uses at most $1+2m$ oracle calls, including preparation.
\end{theorem}
\begin{proof}
    Recall from \cref{eq:characteristic_polynomial} that the characteristic polynomial of $A_{\theta,\kappa}$ is $\chi\rbra{\lambda} = \lambda^2 - \rbra*{1+\sqrt{1-\kappa}}\cos\rbra{2\theta}\lambda + \sqrt{1-\kappa}$.
    Its discriminant is
    \begin{equation*}
        \Delta = \rbra*{1+\sqrt{1-\kappa}}^2\cos^2\rbra{2\theta} - 4\sqrt{1-\kappa} = \rbra*{1-\sqrt{1-\kappa}}^2 - \rbra*{1+\sqrt{1-\kappa}}^2\sin^2\rbra{2\theta}.
    \end{equation*}
    Substituting $\kappa = \kappa^*\rbra{\theta_0}$ from \cref{thm:fixed_point_expected_cost} yields
    \begin{align*}
        \Delta &= \rbra*{1-\sqrt{1-\kappa^*\rbra{\theta_0}}}^2 - \rbra*{1+\sqrt{1-\kappa^*\rbra{\theta_0}}}^2\sin^2\rbra{2\theta} \\
        &= \rbra*{1+\sqrt{1-\kappa^*\rbra{\theta_0}}}^2\rbra*{\rbra*{\frac{1-\sqrt{1-\kappa^*\rbra{\theta_0}}}{1+\sqrt{1-\kappa^*\rbra{\theta_0}}}}^2 - \sin^2\rbra{2\theta}}.
    \end{align*}
    With the optimality condition (see \cref{eq:kappa_x_substitution,eq:x_optimal}), this becomes
    \begin{equation}\label{eq:delta}
            \Delta = \rbra*{1+\sqrt{1-\kappa^*\rbra{\theta_0}}}^2\rbra*{2\sin^2\rbra{\theta_0}-\sin^2\rbra{2\theta}} = \rbra*{1+\sqrt{1-\kappa^*\rbra{\theta_0}}}^2\rbra*{2p_0-4p\rbra{1-p}}.
    \end{equation}
    Because $0 < p_0 \leq p < 1/2$, we have $4p\rbra{1-p} > 2p \geq 2p_0$, and hence
    \begin{equation*}
        \Delta < \rbra*{1+\sqrt{1-\kappa^*\rbra{\theta_0}}}^2\rbra{2p_0-2p} \leq 0.
    \end{equation*}
    Therefore, the eigenvalues of $A_{\theta,\kappa^*\rbra{\theta_0}}$ are distinct complex conjugates
    \begin{equation}\label{eq:fixed_point_eigenvalues}
        \lambda_1 = \frac{\rbra*{1+\sqrt{1-\kappa^*\rbra{\theta_0}}}\cos\rbra{2\theta} + i\sqrt{-\Delta}}{2}, \qquad \lambda_2 = \frac{\rbra*{1+\sqrt{1-\kappa^*\rbra{\theta_0}}}\cos\rbra{2\theta} - i\sqrt{-\Delta}}{2}.
    \end{equation}
    Since $\lambda_1\lambda_2 = \sqrt{1-\kappa^*\rbra{\theta_0}}$, we write
    \begin{equation}\label{eq:s}
        s = \sqrt{1-\kappa^*\rbra{\theta_0}}, \qquad \lambda_1 = \sqrt{s}e^{i\phi}, \quad \lambda_2 = \sqrt{s}e^{-i\phi}
    \end{equation}
    with
    \begin{equation}\label{eq:phi_cos_sin}
        \cos\rbra{\phi} = \frac{\rbra*{1+s}\cos\rbra{2\theta}}{2\sqrt{s}}, \qquad \sin\rbra{\phi} = \frac{\sqrt{-\Delta}}{2\sqrt{s}}.
    \end{equation}
    We now compute $A_{\theta,\kappa^*\rbra{\theta_0}}^m u_0$.
    Using the spectral decomposition, there exist $\alpha_1,\alpha_2,\beta_1,\beta_2 \in \CC$ such that
    \begin{equation*}
        A^m_{\theta,\kappa^*\rbra{\theta_0}}u_0 = s^{m/2}\begin{pmatrix}
            \alpha_1 \cos\rbra{m\phi}+ i\alpha_2\sin\rbra{m\phi} \\
            \beta_1 \cos\rbra{m\phi} + i\beta_2\sin\rbra{m\phi}
        \end{pmatrix}.
    \end{equation*}
    Here $u_0 = \rbra*{\cos\rbra{\theta},-\sin\rbra{\theta}}^{\dagger}$ and $A_{\theta,\kappa^*\rbra{\theta_0}}u_0 = \rbra*{\cos\rbra{\theta},s\sin\rbra{\theta}}^{\dagger}$.
    From the initial conditions $m = 0$ and $m = 1$ we obtain
    \begin{align*}
        \alpha_1 &= \cos\rbra{\theta}, &
        s^{1/2}\rbra{\alpha_1\cos\rbra{\phi}+i\alpha_2\sin\rbra{\phi}} &= \cos\rbra{\theta}, \\
        \beta_1 &= -\sin\rbra{\theta}, &
        s^{1/2}\rbra{\beta_1\cos\rbra{\phi}+i\beta_2\sin\rbra{\phi}} &= s\sin\rbra{\theta}.
    \end{align*}
    Solving them gives
    \begin{equation}\label{eq:fixed_point_decompsitions}
        \begin{aligned}
            \alpha_1 &= \cos\rbra{\theta}, &
            \alpha_2 &= -i\cos\rbra{\theta}\,\frac{1-s^{1/2}\cos\rbra{\phi}}{s^{1/2}\sin\rbra{\phi}}, \\
            \beta_1 &= -\sin\rbra{\theta}, &
            \beta_2 &= -i\sin\rbra{\theta}\,\frac{s+s^{1/2}\cos\rbra{\phi}}{s^{1/2}\sin\rbra{\phi}}.
        \end{aligned}
    \end{equation}
    Thus, after $m$ rounds of failures, the unnormalised state in the rotated coordinates is
    \begin{align*}
        A^m_{\theta,\kappa^*\rbra{\theta_0}}u_0
        &= s^{m/2}\begin{pmatrix}
            \cos\rbra{\theta}\rbra*{\cos\rbra{m\phi}+\dfrac{1-s^{1/2}\cos\rbra{\phi}}{s^{1/2}\sin\rbra{\phi}}\sin\rbra{m\phi}}\\[1ex]
            \sin\rbra{\theta}\rbra*{-\cos\rbra{m\phi}+\dfrac{s+s^{1/2}\cos\rbra{\phi}}{s^{1/2}\sin\rbra{\phi}}\sin\rbra{m\phi}}
        \end{pmatrix} \\
        &= s^{m/2}\begin{pmatrix}
            \cos\rbra{\theta}\rbra*{\cos\rbra{m\phi}+\dfrac{\sqrt{2p_0}+2p}{\sqrt{4p\rbra{1-p}-2p_0}}\sin\rbra{m\phi}}\\[1ex]
            \sin\rbra{\theta}\rbra*{-\cos\rbra{m\phi}+\dfrac{2-\sqrt{2p_0}-2p}{\sqrt{4p\rbra{1-p}-2p_0}}\sin\rbra{m\phi}}
        \end{pmatrix}.
    \end{align*}
    The last equality uses \cref{eq:phi_cos_sin,eq:delta} and $\rbra{1-s}/\rbra{1+s} = \sqrt{2p_0}$.
    The physical register state is $R_\theta A^m_{\theta,\kappa^*\rbra{\theta_0}}u_0$.
    Since the final measurement of \cref{algo:truncated_fixed_point_lower_bounded_p} checks for $\ket{0}\ket{\psi_0}$ and $\ket{1}\ket{\psi_1}$, its failure probability is
    \begin{equation}\label{eq:probability_failure_def}
            P_{\mathrm{fail}}\rbra{m} = \abs*{e_0^{\dagger}R_\theta A^m_{\theta,\kappa^*\rbra{\theta_0}}u_0}^2 = s^m\rbra{1-p}\abs*{\cos\rbra{m\phi}+\frac{\sqrt{2p_0}-2p}{\sqrt{4p\rbra{1-p}-2p_0}}\sin\rbra{m\phi}}^2,
    \end{equation}
    where $e_0 = \rbra{1,0}^{\dagger}$.
    The last equality uses $\cos\rbra{2\theta} = 1-2p$, $\sin\rbra{2\theta}\sin\rbra{\theta} = 2p\cos\rbra{\theta}$, and
    \begin{equation*}
        \rbra{1-2p}\rbra*{\sqrt{2p_0}+2p}-2p\rbra*{2-\sqrt{2p_0}-2p} = \sqrt{2p_0}-2p.
    \end{equation*}
    Applying the Cauchy--Schwarz inequality to \cref{eq:probability_failure_def}, we obtain
    \begin{align}\label{eq:probability_failure1}
        P_{\mathrm{fail}}\rbra{m}
        & \leq s^m\rbra{1-p}\rbra*{\cos^2\rbra{m\phi}+\sin^2\rbra{m\phi}}\rbra*{1+\frac{\rbra*{\sqrt{2p_0}-2p}^2}{4p\rbra{1-p}-2p_0}}\notag\\
        &= s^m\rbra{1-p}\rbra*{1+\frac{\rbra*{\sqrt{2p_0}-2p}^2}{4p\rbra{1-p}-2p_0}}.
    \end{align}
    Simplifying the coefficient gives
    \begin{equation*}
        \rbra{1-p}\rbra*{1+\frac{\rbra*{\sqrt{2p_0}-2p}^2}{4p\rbra{1-p}-2p_0}}
        = \rbra{1-p}\frac{4p\rbra{1-p}-4p\sqrt{2p_0}+4p^2}{4p\rbra{1-p}-2p_0} = \frac{4p\rbra{1-p}\rbra{1-\sqrt{2p_0}}}{4p\rbra{1-p}-2p_0}.
    \end{equation*}
    Since $4p\rbra{1-p}$ increases on $\rbra{0,1/2}$ and $z/\rbra{z-2p_0}$ decreases for $z > 2p_0$, the condition $p \geq p_0$ yields
    \begin{equation*}
        \frac{4p\rbra{1-p}\rbra{1-\sqrt{2p_0}}}{4p\rbra{1-p}-2p_0}
        \leq \frac{4p_0\rbra{1-p_0}\rbra{1-\sqrt{2p_0}}}{4p_0\rbra{1-p_0}-2p_0}
        = \frac{\rbra{2-2p_0}\rbra{1-\sqrt{2p_0}}}{1-2p_0}
        = \frac{2-2p_0}{1+\sqrt{2p_0}} < 2.
    \end{equation*}
    Applying this to \cref{eq:probability_failure1} yields
    \begin{equation}\label{eq:probability_failure2}
        P_{\mathrm{fail}}\rbra{m} \leq \frac{4p\rbra{1-p}\rbra{1-\sqrt{2p_0}}}{4p\rbra{1-p}-2p_0}s^m < 2s^m.
    \end{equation}
    Finally, from the optimality condition (\cref{eq:optimal_kappa}) for $\kappa^*\rbra{\theta_0}$ and \cref{eq:s}, we have
    \begin{equation*}
        s = \sqrt{1-\kappa^*\rbra{\theta_0}}
        = \frac{1-\sqrt{2}\sin\rbra{\theta_0}}{1+\sqrt{2}\sin\rbra{\theta_0}}
        = \frac{1-\sqrt{2p_0}}{1+\sqrt{2p_0}}.
    \end{equation*}
    Thus,
    \begin{equation*}
        P_{\mathrm{fail}}\rbra{m} < 2\rbra*{\frac{1-\sqrt{2p_0}}{1+\sqrt{2p_0}}}^m.
    \end{equation*}
    To guarantee $P_{\mathrm{fail}}\rbra{m} \leq \epsilon \in \rbra{0,1/2}$, it is enough to choose
    \begin{equation*}
        m \geq \frac{\ln\rbra{2/\epsilon}}{\ln\rbra*{\frac{1+\sqrt{2p_0}}{1-\sqrt{2p_0}}}}.
    \end{equation*}
    Using $\ln\rbra*{\frac{1+x}{1-x}} \geq 2x$ for $0 < x < 1$, with $x = \sqrt{2p_0}$, we obtain
    \begin{equation}\label{eq:log_1k_1k}
        -\ln\rbra{s} = \ln\rbra*{\frac{1+\sqrt{2p_0}}{1-\sqrt{2p_0}}} \geq 2\sqrt{2p_0}.
    \end{equation}
    This also gives $2s^m \leq 2\exp\rbra*{-2\sqrt{2p_0}\,m}$, completing \cref{eq:truncated_fixed_point_probability}.
    Consequently, a simpler sufficient condition is
    \begin{equation*}
        m \geq \frac{\ln\rbra{2/\epsilon}}{2\sqrt{2p_0}}
        \geq \frac{\ln\rbra{2/\epsilon}}{\ln\rbra*{\frac{1+\sqrt{2p_0}}{1-\sqrt{2p_0}}}},
    \end{equation*}
    which is \cref{eq:truncated_rounds}. Each unsuccessful round uses one Grover rotation, so there are at most $m$ rotations and one preparation, giving at most $1+2m$ oracle calls.
\end{proof}

To compare it with fixed-point quantum search, set $\epsilon = \delta^2$ and choose
$m = \ceil*{\frac{\ln\rbra{2/\delta^2}}{2\sqrt{2p_0}}}$.
Then, the complete algorithm uses at most $1+2m$ oracle calls.
For fixed $\delta$ and $p_0 \to 0$, this sufficient budget has leading term $\frac{\sqrt{2}\ln\rbra{\sqrt{2}/\delta}}{\sqrt{p_0}}$.
The familiar small-error query expression for fixed-point search~\cite{YLC14} is $\ln\rbra{2/\delta}/\sqrt{p_0}$, with a natural logarithm.
The advantage established in \cref{thm:fixed_point_lower_bounded} concerns expected cost and its dependence on the actual $p$.
Early success can shorten a run of the truncated algorithm, whose expected cost is at most that of the untruncated loop including preparation.

\section{Adaptive schedule for unknown \texorpdfstring{$p$}{p}}
\label{sec:adaptive_for_unknown_p}
We now consider the most challenging setting in which the target probability $p$ is completely unknown.
In this case, neither the exact schedule of \cref{sec:exact_for_known_p} nor the fixed-point schedule of \cref{sec:fixed_point_for_lower_bounded_p} (which requires a known lower bound $p_0$) is directly applicable.
We therefore adopt a general adaptive strategy shown in \cref{algo:adaptive_unknown_p}.

\begin{algorithm*}[ht]
    \caption{Quantum amplitude amplification with adaptive weak measurement schedule}
    \label{algo:adaptive_unknown_p}
    \begin{algorithmic}[1]
        \Require A quantum oracle $\cA$ (and its inverse $\cA^{\dagger}$) satisfying \cref{eq:oracle_A} and a weak measurement strength function $\kappa:\NN \to (0,1)$.

        \Ensure The target state $\ket{1}\ket{\psi_1}$ upon termination.
        
        \State Let $G$ be the Grover rotation as in \cref{eq:grover_operator}.

        \State $t \gets 0$

        \State Initialise an $n$-qubit register to the state $\cA\ket{0^n}$.
        \Repeat
            \State $t \gets t+1$
            \State Perform the weak measurement $\cbra*{M_0 = \ketbra{0}{0}+\sqrt{1-\kappa\rbra{t}}\ketbra{1}{1}, M_1 = \sqrt{\kappa\rbra{t}}\ketbra{1}{1}}$ on the first qubit.
            \If {outcome $= 1$}
                \State \Return \Comment{Success: register in state $\ket{1}\ket{\psi_1}$}
            \EndIf
            \State Apply $G$ to the register.
        \Until{success}
    \end{algorithmic}
\end{algorithm*}
To find a suitable strength function $\kappa:\NN \to (0,1)$, we first analyse the survival probability of \cref{algo:adaptive_unknown_p} after $m$ consecutive failures, and then use this analysis to design an effective control law.

\subsection{Survival probability under adaptive weak measurement strengths}
We work in the same two-dimensional subspace $\cbra{\ket{0}\ket{\psi_0}, \ket{1}\ket{\psi_1}}$ as in \cref{sec:fixed_point_for_lower_bounded_p}.
Let
\begin{gather*}
    u_0 = \begin{pmatrix}
        \cos\rbra{\theta} \\
        -\sin\rbra{\theta}
    \end{pmatrix}, \qquad
    R_\theta = \begin{pmatrix}
        \cos\rbra{2\theta} & -\sin\rbra{2\theta} \\
        \sin\rbra{2\theta} & \cos\rbra{2\theta}
    \end{pmatrix}, \qquad D_{\kappa\rbra{t}} = \begin{pmatrix}
        1 & 0 \\
        0 & \sqrt{1-\kappa\rbra{t}}
    \end{pmatrix}, \qquad
    A_{\theta,\kappa\rbra{t}} = D_{\kappa\rbra{t}}R_{\theta},
\end{gather*}
then the unnormalised state after $m$ consecutive failures of \cref{algo:adaptive_unknown_p} is $R_\theta u_m$, where
\begin{equation}\label{eq:u_m_adaptive}
    u_{m} = A_{\theta,\kappa\rbra{m}}\cdots A_{\theta,\kappa\rbra{1}}u_0,
\end{equation}
and the survival probability is $S_m\rbra{\theta, \kappa} = \Abs*{u_m}^2$.

To analyse this probability recursively, for each round $t \geq 1$, define $\alpha_t$ as the angle of the state just before the weak measurement in round $t$ (to reach this round, the algorithm must have experienced $t-1$ consecutive failures).
From \cref{eq:failure_state}, after a failure outcome, the post-measurement state has angle
$\arctan\rbra*{\sqrt{1-\kappa\rbra{t}}\tan\rbra{\alpha_{t}}}$.
Applying the subsequent Grover rotation adds $2\theta$ to the angle, giving the recurrence\footnote{The branch of the arctangent is chosen continuously with $\alpha_t$, starting from $\arctan\rbra{0} = 0$. The formula extends by continuity when $\alpha_t = \pi/2+k\pi$ for an integer $k$.}
\begin{equation}\label{eq:adaptive_alpha_t}
    \alpha_{t+1} = \arctan\rbra*{\sqrt{1-\kappa\rbra{t}}\tan\rbra{\alpha_{t}}} +2\theta, \qquad t\geq 1
\end{equation}
with the initial angle $\alpha_1 = \theta$ (since the register starts in $\cA\ket{0^n}$).
The following proposition expresses the survival probability after $m$ rounds in terms of the product of conditional failure probabilities.

\begin{proposition}\label{prop:survival_product}
    For any $m \geq 0$, $S_m\rbra{\theta, \kappa} = \prod_{t=1}^{m}\rbra*{1 - \kappa\rbra{t}\sin^2\rbra*{\alpha_t}}$.
\end{proposition}
\begin{proof}
    In round $t \geq 1$, the conditional probability of a failure (given that no success occurred earlier) is, by \cref{eq:failure_p}, $ s_t = 1-\kappa\rbra{t}\sin^2\rbra{\alpha_t}$.
    The survival probabilities therefore satisfy
    \begin{equation*}
        S_{t}\rbra{\theta,\kappa} = s_t S_{t-1}\rbra{\theta,\kappa} = \rbra*{1-\kappa\rbra{t}\sin^2\rbra{\alpha_t}}S_{t-1}\rbra{\theta,\kappa}.
    \end{equation*}
    Since the base case $S_{0}\rbra{\theta,\kappa} = 1 = \prod_{t=1}^{0}\rbra*{1 - \kappa\rbra{t}\sin^2\rbra*{\alpha_t}}$ holds, with the empty product equal to $1$, we obtain the product formula by induction on $m$.
\end{proof}

From \cref{prop:survival_product}, the survival probability can be estimated via the following exponential bounds.

\begin{lemma}\label{lem:survival_probability_bound}
    For any $m \geq 0$,
    \begin{equation}\label{eq:survival_probability_bound1}
        \exp\rbra*{-\sum_{t=1}^m\frac{\kappa\rbra{t}\sin^2\rbra{\alpha_t}}{1-\kappa\rbra{t}\sin^2\rbra{\alpha_t}}} \leq S_m\rbra{\theta, \kappa} \leq \exp\rbra*{-\sum_{t=1}^m\kappa\rbra{t}\sin^2\rbra{\alpha_t}}.
    \end{equation}
    In particular, for any $0 \leq m \leq \floor*{\pi/\rbra{8\theta}}$,
    \begin{equation}\label{eq:survival_probability_bound2}
        \exp\rbra*{-2\sum_{t=1}^m\kappa\rbra{t}\sin^2\rbra{\alpha_t}} \leq S_m\rbra{\theta, \kappa} \leq \exp\rbra*{-\sum_{t=1}^m\kappa\rbra{t}\sin^2\rbra{\alpha_t}}.
    \end{equation}
\end{lemma}
\begin{proof}
    It is easy to verify that for all $x \in [0,1)$,
    \begin{equation*}
        -\frac{x}{1-x} \leq \ln\rbra{1-x} \leq -x.
    \end{equation*}
    Applying this inequality with $x = \kappa\rbra{t}\sin^2\alpha_t$, which lies in $[0,1)$, gives
    \begin{equation*}
        -\frac{\kappa\rbra{t}\sin^2\rbra{\alpha_t}}{1-\kappa\rbra{t}\sin^2\rbra{\alpha_t}} \leq \ln\rbra*{1-\kappa\rbra{t}\sin^2\rbra{\alpha_t}} \leq - \kappa\rbra{t}\sin^2\rbra{\alpha_t}.
    \end{equation*}
    Summing over $t = 1,\dots,m$ yields
    \begin{equation}\label{eq:survival_probability_bound_ln}
        -\sum_{t=1}^m\frac{\kappa\rbra{t}\sin^2\rbra{\alpha_t}}{1-\kappa\rbra{t}\sin^2\rbra{\alpha_t}} \leq \sum_{t=1}^m\ln\rbra*{1-\kappa\rbra{t}\sin^2\rbra{\alpha_t}} \leq - \sum_{t=1}^m\kappa\rbra{t}\sin^2\rbra{\alpha_t}.
    \end{equation}
    Exponentiating and recalling from \cref{prop:survival_product} give \cref{eq:survival_probability_bound1}.

    Now assume $1\leq t \leq \floor*{\pi/\rbra{8\theta}}$.
    From the recurrence \cref{eq:adaptive_alpha_t}, we have $\alpha_{t+1} \leq \alpha_t +2\theta$ while $0\leq\alpha_t\leq\pi/4$, with $\alpha_1 = \theta$.
    Thus, we obtain by induction
    \begin{equation}\label{eq:alpha_t_bound}
        0\leq\alpha_t \leq \rbra{2t-1}\theta \leq 2t\theta \leq \frac{\pi}{4} \quad \text{for $1\leq t \leq \floor*{\pi/\rbra{8\theta}}$.}
    \end{equation}
    Consequently $\sin^2\rbra{\alpha_t} \leq \sin^2\rbra{\pi/4} = 1/2$.
    Therefore,
    \begin{equation*}
        \frac{\kappa\rbra{t}\sin^2\rbra{\alpha_t}}{1-\kappa\rbra{t}\sin^2\rbra{\alpha_t}} \leq \frac{\kappa\rbra{t}\sin^2\rbra{\alpha_t}}{1-1\cdot \frac{1}{2}} = 2\kappa\rbra{t}\sin^2\rbra{\alpha_t}.
    \end{equation*}
    Hence, for any $0\leq m \leq \floor*{\pi/\rbra{8\theta}}$,
    \begin{equation*}
        -\sum_{t=1}^m\frac{\kappa\rbra{t}\sin^2\rbra{\alpha_t}}{1-\kappa\rbra{t}\sin^2\rbra{\alpha_t}} \geq -2\sum_{t=1}^m\kappa\rbra{t}\sin^2\rbra{\alpha_t}.
    \end{equation*}
    Combining this with \cref{eq:survival_probability_bound_ln} and exponentiating gives \cref{eq:survival_probability_bound2}.
\end{proof}

With \cref{lem:survival_probability_bound}, we focus on the special case $m = \floor*{\pi/\rbra{8\theta}}$.
To guide the choice of schedule, we seek a constant probability of success within $\floor*{\pi/\rbra{8\theta}}$ rounds, i.e., $1-S_{\floor*{\pi/\rbra{8\theta}}}\rbra{\theta, \kappa} = \Theta\rbra{1}$.
From the bound in \cref{eq:survival_probability_bound2}, this condition translates into
\begin{equation*}
    \sum_{t=1}^{\floor*{\pi/\rbra{8\theta}}}\kappa\rbra{t}\sin^2\rbra{\alpha_t} = \Omega\rbra{1}.
\end{equation*}
Then, we can characterise the needed order of $\kappa\rbra{t}$ as follows.

\begin{lemma}\label{lem:kappa_order}
    Let $N_\theta=\floor*{\pi/\rbra{8\theta}}$.
    The scale $1/t$ is critical in the following sense.
    As $\theta\to 0$, if $\kappa\rbra{t}=\Theta\rbra{1/t}$, then $\sum_{t=1}^{N_\theta}\kappa\rbra{t}\sin^2\rbra{\alpha_t}=\Theta\rbra{1}$; if $\kappa\rbra{t}=o\rbra{1/t}$ or $\kappa\rbra{t}=\omega\rbra{1/t}$, then $\sum_{t=1}^{N_\theta}\kappa\rbra{t}\sin^2\rbra{\alpha_t}=o\rbra{1}$.
    Hence, among the three asymptotic regimes considered above, achieving a $\Theta\rbra{1}$ sum requires $\kappa\rbra{t}=\Theta\rbra{1/t}$.
\end{lemma}
\begin{proof}
    Recall from \cref{eq:alpha_t_bound} that for any $1\leq t \leq \floor*{\pi/\rbra{8\theta}}$, we have $\alpha_t \leq 2t\theta \leq \pi/4$.
    Using $\sqrt{2}x/\pi \leq \sin\rbra{x} \leq x$ for $x \in [0, \pi/4]$, we obtain
    \begin{equation}\label{eq:sum_sin_bound}
        \frac{2}{\pi^2}\sum_{t=1}^{\floor*{\pi/\rbra{8\theta}}}\kappa\rbra{t}\alpha_t^2 \leq \sum_{t=1}^{\floor*{\pi/\rbra{8\theta}}}\kappa\rbra{t}\sin^2\rbra{\alpha_t} \leq \sum_{t=1}^{\floor*{\pi/\rbra{8\theta}}}\kappa\rbra{t}\alpha_t^2.
    \end{equation}
    \paragraph{Sufficiency: $\kappa\rbra{t} \in \Theta\rbra{1/t}$ implies the sum is $\Theta\rbra{1}$.}
    Assume that there exist $c_1\geq c_2 > 0$ with $c_1\geq1$ such that
    \begin{equation}\label{eq:kappa_order}
        c_2/t \leq \kappa\rbra{t} \leq c_1/t.
    \end{equation}
    From an upper bound, using $\alpha_t \leq 2\theta t$, we get
    \begin{equation*}
        \sum_{t=1}^{\floor*{\pi/\rbra{8\theta}}}\kappa\rbra{t}\alpha_t^2 \leq 4\theta^2\sum_{t=1}^{\floor*{\pi/\rbra{8\theta}}}\kappa\rbra{t}t^2 \leq 4c_1\theta^2\sum_{t=1}^{\floor*{\pi/\rbra{8\theta}}}t = 2c_1\theta^2\floor*{\pi/\rbra{8\theta}}\rbra*{\floor*{\pi/\rbra{8\theta}}+1} \leq \frac{\pi^2c_1}{16},
    \end{equation*}
    where the last inequality holds because $\floor*{\pi/\rbra{8\theta}}\rbra*{\floor*{\pi/\rbra{8\theta}}+1} \leq \pi/\rbra{8\theta}\cdot \rbra{\pi/\rbra{8\theta}+1} \leq 2\rbra{\pi/\rbra{8\theta}}^2$.
    For a lower bound, note that for $0\leq x\leq 1$ and $\alpha \in \sbra{0,\pi/4}$, $\arctan\rbra{x\tan\rbra{\alpha}} \geq x\alpha$.
    With the recurrence of \cref{eq:adaptive_alpha_t}, we have
        \begin{equation*}
            \alpha_{t+1} = \arctan\rbra*{\sqrt{1-\kappa\rbra{t}}\tan\rbra{\alpha_{t}}} +2\theta \geq \sqrt{1-\kappa\rbra{t}}\alpha_t + 2\theta \geq \rbra*{1-\kappa\rbra{t}}\alpha_t + 2\theta.
        \end{equation*}
        Using the second inequality of \cref{eq:kappa_order}, we obtain
        \begin{equation*}
            \alpha_{t+1} \geq \rbra*{1-\frac{c_1}{t}}\alpha_t +2\theta.
        \end{equation*}
        Then, we can prove that $\alpha_{t} \geq \frac{2\theta}{1+c_1}t$ by induction, since if $\alpha_t \geq \frac{2\theta}{1+c_1}t$ and $t\geq c_1$, we have
        \begin{equation*}
            \alpha_{t+1} - \frac{2\theta}{1+c_1}\rbra{t+1} \geq \rbra*{1-\frac{c_1}{t}}\alpha_t +2\theta - \frac{2\theta}{1+c_1}\rbra{t+1} \geq \rbra*{1-\frac{c_1}{t}}\frac{2\theta}{1+c_1}t +2\theta - \frac{2\theta}{1+c_1}\rbra{t+1} = 0.
        \end{equation*}
        For $1\leq t<c_1$, the recurrence gives $\alpha_{t+1}\geq2\theta\geq\frac{2\theta}{1+c_1}\rbra{t+1}$, and the base case holds because $\alpha_1=\theta\geq\frac{2\theta}{1+c_1}$.
        Consequently,
        \begin{equation*}
            \sum_{t=1}^{\floor*{\pi/\rbra{8\theta}}}\kappa\rbra{t}\alpha_t^2 \geq \sum_{t=1}^{\floor*{\pi/\rbra{8\theta}}}\frac{c_2}{t}\cdot \rbra*{\frac{2\theta}{1+c_1}t}^2 = \frac{2c_2}{\rbra{1+c_1}^2}\theta^2\floor*{\pi/\rbra{8\theta}}\rbra*{\floor*{\pi/\rbra{8\theta}}+1} > \frac{\pi^2c_2}{128\rbra{1+c_1}^2},
        \end{equation*}
    where the last inequality holds for sufficiently small $\theta$ because $\floor*{\pi/\rbra{8\theta}}\rbra*{\floor*{\pi/\rbra{8\theta}}+1} > \rbra*{\pi/\rbra{8\theta}-1}^2 \geq \rbra*{\pi/\rbra{16\theta}}^2$.
    Thus, by \cref{eq:sum_sin_bound}, $\sum_{t=1}^{\floor*{\pi/\rbra{8\theta}}}\kappa\rbra{t}\sin^2\rbra{\alpha_t} = \Theta\rbra{1}$.
    \paragraph{Necessity within the three regimes.}
    Suppose $\sum_{t=1}^{\floor*{\pi/\rbra{8\theta}}}\kappa\rbra{t}\sin^2\rbra{\alpha_t} = \Theta\rbra{1}$, we conclude the proof by contradiction to rule out the following two cases.
    \begin{itemize}
        \item \underline{Case $\kappa\rbra{t} = o\rbra{1/t}$ as $t \to \infty$:} In this case, for any $\epsilon > 0$, there exists $T_\epsilon \in \NN_+$ such that
        \begin{equation*}
            \kappa\rbra{t} \leq \epsilon/t \quad \text{for any $t > T_\epsilon$.}
        \end{equation*}
        Then, using the second inequality of \cref{eq:sum_sin_bound}, we have
        \begin{align*}
            \sum_{t=1}^{\floor*{\pi/\rbra{8\theta}}}\kappa\rbra{t}\sin^2\rbra{\alpha_t} \leq \sum_{t=1}^{\floor*{\pi/\rbra{8\theta}}}\kappa\rbra{t}\alpha_t^2 &\leq 4\theta^2\sum_{t=1}^{\floor*{\pi/\rbra{8\theta}}}\kappa\rbra{t}t^2 
            = 4\theta^2\rbra*{\sum_{t=1}^{T_\epsilon}\kappa\rbra{t}t^2 + \sum_{t=T_\epsilon+1}^{\floor*{\pi/\rbra{8\theta}}}\kappa\rbra{t}t^2}.
        \end{align*}
        With
        \begin{equation*}
            \sum_{t=1}^{T_\epsilon}\kappa\rbra{t}t^2 + \sum_{t=T_\epsilon+1}^{\floor*{\pi/\rbra{8\theta}}}\kappa\rbra{t}t^2
            \leq \sum_{t=1}^{T_\epsilon}1\cdot t^2 + \sum_{t=T_\epsilon+1}^{\floor*{\pi/\rbra{8\theta}}}\epsilon t
            \leq \frac{\rbra*{T_\epsilon+1}^3}{3}+\epsilon\cdot \frac{\rbra*{\floor*{\pi/\rbra{8\theta}}+1}^2}{2},
        \end{equation*}
        we obtain
        \begin{equation*}
            \sum_{t=1}^{\floor*{\pi/\rbra{8\theta}}}\kappa\rbra{t}\sin^2\rbra{\alpha_t} \leq 4\theta^2\rbra*{\frac{\rbra*{T_\epsilon+1}^3}{3}+\epsilon\cdot \frac{\rbra*{\floor*{\pi/\rbra{8\theta}}+1}^2}{2}} \to \epsilon\cdot\frac{\pi^2}{32}
        \end{equation*}
        as $\theta \to 0$.
        Since this holds for any $\epsilon > 0$, we have $\sum_{t=1}^{\floor*{\pi/\rbra{8\theta}}}\kappa\rbra{t}\sin^2\rbra{\alpha_t} = o\rbra{1}$ as $\theta \to 0$, which contradicts the assumption.
        \item \underline{Case $\kappa\rbra{t} = \omega\rbra{1/t}$ as $t \to \infty$:}
        Considering that for $x \in \sbra{0,1}$ and $\alpha \in \sbra{0,\pi/4}$, we have
        \begin{equation}\label{eq:alpha_alphat}
            \alpha - \arctan\rbra*{\sqrt{1-x}\tan\rbra{\alpha}} = \int_{\sqrt{1-x}}^1 \frac{\tan\rbra{\alpha}}{1+s^2\tan^2\rbra{\alpha}}\mathrm{d}s \geq \frac{\tan\rbra{\alpha}}{2}\rbra*{1-\sqrt{1-x}},
        \end{equation}
        where the last inequality holds because $\frac{1}{1+s^2\tan^2\rbra{\alpha}} \geq \frac{1}{1+s^2\tan^2\rbra{\pi/4}} \geq \frac{1}{2}$ for any $s \in \sbra{0,1}$ and $\alpha \in \sbra{0,\pi/4}$.
        Using $\rbra*{1-\sqrt{1-x}} = \frac{x}{1+\sqrt{1-x}} \geq \frac{x}{2}$ for $x \in \sbra{0,1}$ and $\tan\rbra{\alpha} \geq \alpha$ for $\alpha \in \sbra{0,\pi/4}$, we get
        \begin{equation*}
            \alpha - \arctan\rbra*{\sqrt{1-x}\tan\rbra{\alpha}} \geq \frac{\tan\rbra{\alpha}}{2}\rbra*{1-\sqrt{1-x}} \geq \frac{\alpha}{2}\cdot \frac{x}{2} = \frac{x\alpha}{4}.
        \end{equation*}
        So
        \begin{equation*}
            \arctan\rbra*{\sqrt{1-x}\tan\rbra{\alpha}} \leq \rbra*{1-x/4}\alpha
        \end{equation*}
        for any $x \in \sbra{0,1}$ and $\alpha \in \sbra{0,\pi/4}$.
        Hence, with the recurrence in \cref{eq:adaptive_alpha_t} and $\kappa\rbra{t} \in \rbra{0,1}, \alpha_t \leq \pi/4$, we have
        \begin{equation*}
            \alpha_{t+1} = \arctan\rbra*{\sqrt{1-\kappa\rbra{t}}\tan\rbra{\alpha_{t}}} +2\theta \leq \rbra*{1 - \frac{\kappa\rbra{t}}{4}}\alpha_t +2\theta.
        \end{equation*}
        Then,
        \begin{align*}
            \alpha_{t+1}^2 &\leq \rbra*{1 - \frac{\kappa\rbra{t}}{4}}^2\alpha_t^2 + 4\rbra*{1 - \frac{\kappa\rbra{t}}{4}}\alpha_t\theta + 4\theta^2 \\
            &\leq \rbra*{1 - \frac{\kappa\rbra{t}}{4}}^2\alpha_t^2 + \rbra*{1 - \frac{\kappa\rbra{t}}{4}}\rbra*{\frac{\kappa\rbra{t}}{4}\alpha^2_t + \frac{16}{\kappa\rbra{t}}\theta^2} + 4\theta^2 \tag{by the mean inequality}\\
            &= \rbra*{1-\frac{\kappa\rbra{t}}{4}}\alpha^2_t +\frac{16}{\kappa\rbra{t}}\theta^2.
        \end{align*}
        Rearranging this gives $\frac{\kappa\rbra{t}}{4}\alpha_t^2 \leq \alpha_t^2 - \alpha_{t+1}^2  +\frac{16}{\kappa\rbra{t}}\theta^2$, thus
        \begin{equation}\label{eq:kappa_t_alpha_t2_bound}
            \kappa\rbra{t}\alpha_t^2 \leq 4\alpha_t^2 - 4\alpha_{t+1}^2  +\frac{64}{\kappa\rbra{t}}\theta^2.
        \end{equation}
        Now because $\kappa\rbra{t} = \omega\rbra{1/t}$ as $t \to \infty$, for any $M > 0$, there exists $T_M \in \NN_+$ such that
        \begin{equation*}
            \kappa\rbra{t} \geq M/t \quad \text{for any $t > T_M$}.
        \end{equation*}
        Summing \cref{eq:kappa_t_alpha_t2_bound} from $t = T_M+1$ to $\floor*{\pi/\rbra{8\theta}}$, we obtain
        \begin{equation*}
            \sum_{t=T_M+1}^{\floor*{\pi/\rbra{8\theta}}}\kappa\rbra{t}\alpha_t^2 \leq 4\alpha_{T_M+1}^2 +  \sum_{t=T_M+1}^{\floor*{\pi/\rbra{8\theta}}}\frac{64}{\kappa\rbra{t}}\theta^2 \leq 4\alpha_{T_M+1}^2 + \frac{64\theta^2}{M}\sum_{t=T_M+1}^{\floor*{\pi/\rbra{8\theta}}}t.
        \end{equation*}
        Using $\alpha_{T_M+1} \leq 2\theta \rbra{T_M+1}$, we get
        \begin{equation*}
            \sum_{t=T_M+1}^{\floor*{\pi/\rbra{8\theta}}}\kappa\rbra{t}\alpha_t^2 \leq 16\theta^2\rbra*{T_{M}+1}^2 + \frac{64\theta^2}{M}\sum_{t=T_M+1}^{\floor*{\pi/\rbra{8\theta}}}t \leq 16\theta^2\rbra*{T_{M}+1}^2+ \frac{32\theta^2}{M}\rbra*{\floor*{\pi/\rbra{8\theta}}+1}^2 \to \frac{\pi^2}{2M}
        \end{equation*}
        as $\theta \to 0$.
        The contribution from the initial terms $t \leq T_M$ is
        \begin{equation*}
            \sum_{t=1}^{T_M}\kappa\rbra{t}\alpha_t^2 \leq \sum_{t=1}^{T_M}\alpha_t^2 \leq 4\theta^2\sum_{t=1}^{T_M}t^2 \to 0
        \end{equation*}
        as $\theta \to 0$.
        Therefore,
        \begin{equation*}
            \limsup_{\theta\to 0} \sum_{t=1}^{\floor*{\pi/\rbra{8\theta}}}\kappa\rbra{t}\alpha_t^2 =  \lim_{\theta\to 0} \sum_{t=1}^{T_M}\kappa\rbra{t}\alpha_t^2 +  \limsup_{\theta\to 0} \sum_{t=T_M+1}^{\floor*{\pi/\rbra{8\theta}}}\kappa\rbra{t}\alpha_t^2 \leq \frac{\pi^2}{2M}.
        \end{equation*}
        With the second inequality of \cref{eq:sum_sin_bound}, we have
        \begin{equation*}
            \limsup_{\theta\to 0} \sum_{t=1}^{\floor*{\pi/\rbra{8\theta}}}\kappa\rbra{t}\sin^2\rbra{\alpha_t} \leq \limsup_{\theta\to 0} \sum_{t=1}^{\floor*{\pi/\rbra{8\theta}}}\kappa\rbra{t}\alpha_t^2 \leq \frac{\pi^2}{2M}.
        \end{equation*}
        Since this holds for any $M > 0$, we have $\sum_{t=1}^{\floor*{\pi/\rbra{8\theta}}}\kappa\rbra{t}\sin^2\rbra{\alpha_t} = o\rbra{1}$ as $\theta \to 0$, which contradicts the assumption.
    \end{itemize}
    Thus, among the three asymptotic regimes considered in the lemma, only $\kappa\rbra{t}=\Theta\rbra{1/t}$ yields a $\Theta\rbra{1}$ sum.
\end{proof}

\begin{remark}\label{rem:critical_scale}
    We note that the dyadic measurement schedule in \cite[Algorithm~6.1]{Wang18} roughly chooses $\kappa\rbra{t} = 1-e^{-1/2^k}$ if $t \in (c(2^{k}-1), c(2^{k+1}-1)]$ for some constant $c > 0$ depending on the desired failure probability.
    With $t \in (c(2^{k}-1), c(2^{k+1}-1)]$, we have $\frac{c}{t+c} < \frac{1}{2^k} \leq \frac{2c}{t+c}$, hence $\frac{1}{2^k} = \Theta\rbra{1/t}$.
    For sufficiently large $t$, $\frac{1}{2^k}$ is small, then, $\kappa\rbra{t} = 1-e^{-1/2^k} \approx \frac{1}{2^k} = \Theta\rbra{1/t}$.
    Thus, this particular schedule also satisfies $\kappa\rbra{t} = \Theta\rbra{1/t}$, in agreement with our \cref{lem:kappa_order}.
\end{remark}

The following intuition explains \cref{lem:kappa_order}.
For a known target probability $p_0$, \cref{thm:fixed_point_expected_cost} shows that the optimal measurement strength for \cref{algo:fixed_point_lower_bounded_p} scales as $\kappa^* = \Theta\rbra{\sqrt{p_0}}$.
In the adaptive setting, after $t$ consecutive failures, the algorithm has effectively performed $t$ rounds of weak measurements and Grover rotations.
At this stage, suppose that an estimate of the target probability is roughly $p_{\text{guess}}\rbra{t}$, and the number of failures $t$ should be comparable to the number of iterations that a standard Grover search would need to amplify such a probability, i.e., $t = \Theta\rbra{1/\sqrt{p_{\text{guess}}\rbra{t}}}$, so $p_{\text{guess}}\rbra{t} = \Theta\rbra{1/t^2}$.
Substituting this estimate into the optimal strength $\kappa^* = \Theta\rbra{\sqrt{p_0}}$ for \cref{algo:fixed_point_lower_bounded_p} yields
\begin{equation*}
    \kappa\rbra{t} = \Theta\rbra*{\sqrt{p_{\text{guess}}\rbra{t}}} = \Theta\rbra*{\frac{1}{t}},
\end{equation*}
which is precisely the scaling proved in \cref{lem:kappa_order}.
Thus, the lemma identifies the $\Theta\rbra{1/t}$ scale for obtaining a constant probability of success within $\floor*{\pi/\rbra{8\theta}}$ rounds.
This interpretation bridges the gap between the fixed-point schedule for lower-bounded $p$ (presented in \cref{algo:fixed_point_lower_bounded_p}) and the adaptive schedule for completely unknown $p$ (given in \cref{algo:adaptive_unknown_p}).

\subsection{The adaptive weak measurement schedule}
Motivated by the critical scale in \cref{lem:kappa_order}, we now focus on a concrete one-parameter family of strength functions:
\begin{equation}\label{eq:kappa_family}
    \kappa_b\rbra{t} = \min\cbra*{\frac{1}{2},\, \frac{b}{t}}, \qquad t > 0.
\end{equation}
The truncation at $1/2$ ensures that the measurement does not become too aggressive in the early rounds.
In this subsection, we optimise the constant $b$ in this family to minimise the expected number of oracle calls for any sufficiently small target probability $p$.
The analysis will confirm that the $1/t$ decay indeed yields the optimal $O\rbra{1/\sqrt{p}}$ complexity with a favourable constant factor.

We study the constant factor for the expected number of rounds in \cref{algo:adaptive_unknown_p} compared to $1/\theta$, i.e.,
\begin{equation*}
    \theta\sum_{t\geq 0}S_t\rbra{\theta,\kappa_b} = \theta\sum_{t\geq 0} \Abs*{u_t}^2,
\end{equation*}
where $u_t$ is defined in \cref{eq:u_m_adaptive}.
In particular,
\begin{equation}\label{eq:u_t_update_b}
    u_{t+1} = A_{\theta, \kappa_b\rbra{t+1}}u_t = D_{\kappa_b\rbra{t+1}}R_{\theta} u_t, \qquad u_0 = \begin{pmatrix}
        \cos\rbra{\theta} \\ -\sin\rbra{\theta}
    \end{pmatrix}.
\end{equation}
For any fixed $R > 0$, consider the partial sum
\begin{equation*}
    \theta\sum_{t = \ceil{2b}}^{\floor{R/\theta}} \Abs*{u_t}^2 = \sum_{t = \ceil{2b}}^{\floor{R/\theta}} \theta\Abs*{u_t}^2,
\end{equation*}
which is a Riemann sum
\begin{equation*}
    \sum_{t = \ceil{2b}}^{\floor{R/\theta}} \theta\Abs*{z\rbra{t\theta}}^2,
\end{equation*}
if there is a function $z:\sbra{0,R} \to \RR^2$ satisfying $z\rbra{t\theta} = u_t$.
Let $\tau = t\theta$; such a function $z$ satisfies
\begin{equation*}
    z\rbra*{\tau + \theta} = z\rbra*{\rbra{t+1}\theta} = u_{t+1} = D_{\kappa_b\rbra{t+1}}R_{\theta} u_{t} = D_{\kappa_b\rbra{\tau/\theta+1}}R_{\theta} z\rbra{\tau}.
\end{equation*}
For small $\theta$ and fixed $\tau > 0$, with $\tau \geq \ceil{2b} \theta$, we have the following Taylor expansions, whose remainders are uniform on compact intervals bounded away from $\tau=0$:
\begin{gather*}
    D_{\kappa_b\rbra{\tau/\theta+1}} =
    \begin{pmatrix}
        1 & 0 \\
        0 & \sqrt{1-\kappa_b\rbra{\tau/\theta+1}}
    \end{pmatrix} = \begin{pmatrix}
        1 & 0 \\
        0 & 1
    \end{pmatrix} + \theta \begin{pmatrix}
        0 & 0 \\
        0 & -\frac{b}{2\tau}
    \end{pmatrix} + O\rbra{\theta^2}, \\
    R_{\theta} = \begin{pmatrix}
        \cos\rbra{2\theta} & -\sin\rbra{2\theta} \\
        \sin\rbra{2\theta} & \cos\rbra{2\theta}
    \end{pmatrix} = \begin{pmatrix}
        1 & 0 \\
        0 & 1
    \end{pmatrix} + \theta\begin{pmatrix}
        0 & -2 \\
        2 & 0
    \end{pmatrix} + O\rbra{\theta^2}.
\end{gather*}
Then,
\begin{equation}\label{eq:A_approximation}
    D_{\kappa_b\rbra{\tau/\theta+1}}R_{\theta} = \begin{pmatrix}
        1 & 0 \\
        0 & 1
    \end{pmatrix} + \theta\begin{pmatrix}
        0 & -2 \\
        2 & -\frac{b}{2\tau}
    \end{pmatrix} + O\rbra{\theta^2}.
\end{equation}
Formally taking $\theta \to 0$, $z\rbra{\tau}$ satisfies the differential equation
\begin{equation}\label{eq:continuous_limiting}
    z'\rbra{\tau} = \begin{pmatrix}
        0 & -2 \\
        2 & -\frac{b}{2\tau}
    \end{pmatrix} z\rbra{\tau}.
\end{equation}

We now analyse the relationship between the discrete recurrence and this continuous counterpart.
The following lemma establishes the existence and uniqueness of a regular solution for the differential equation, which will serve as a reference for approximating the discrete dynamics when $\theta$ is small.
Write $e_0 = \lim_{\theta\to 0}u_0 = \begin{psmallmatrix}1 \\ 0\end{psmallmatrix}$.

\begin{lemma}\label{lem:ode_solution}
    Fix $R > 0$ and $b > 2$.
    Consider the differential equation \cref{eq:continuous_limiting} on the interval $(0,R]$ with initial condition
    \begin{equation*}
        z\rbra{0} = \lim_{\tau \to 0^+}z\rbra{\tau} = e_0 = \begin{pmatrix}
            1 \\ 0
        \end{pmatrix}.
    \end{equation*}
    Then, there exists a unique regular solution $z^0:\sbra{0,R} \to \RR^2$ that satisfies the equation for $\tau > 0$ and is continuous at $\tau = 0$ with $z^0\rbra{0} = e_0$.
\end{lemma}
\begin{proof}
    Write $z\rbra{\tau} = \begin{psmallmatrix}
        x\rbra{\tau} \\ y\rbra{\tau}
    \end{psmallmatrix}$.
    Then, \cref{eq:continuous_limiting} becomes
    \begin{equation}\label{eq:continuous_limiting_xy}
        x'\rbra{\tau} = -2y\rbra{\tau}, \qquad y'\rbra{\tau} = 2x\rbra{\tau} - \frac{b}{2\tau}y\rbra{\tau}.
    \end{equation}
    The term $ \frac{b}{2\tau}y\rbra{\tau}$ is singular at $\tau = 0$.
    Rewrite the second equation as
    \begin{equation*}
        y'\rbra{\tau} + \frac{b}{2\tau}y\rbra{\tau} = 2x\rbra{\tau}.
    \end{equation*}
    Multiplying by $\tau^{b/2}$ gives
    \begin{equation*}
        \rbra*{\tau^{b/2}y\rbra{\tau}}' = 2\tau^{b/2}x\rbra{\tau}.
    \end{equation*}
    Since we require $y\rbra{0} = 0$ (from the limit $z\rbra{0} = e_0$), integrating it from $0$ to $\tau$ yields
    \begin{equation}\label{eq:y_integral_solution}
        y\rbra{\tau} = 2\tau^{-b/2}\int_{0}^{\tau} s^{b/2}x\rbra{s}\mathrm{d}s.
    \end{equation}
    Using $x' = -2y$ and the initial condition $x\rbra{0} = 1$, we also have
    \begin{equation*}
        x\rbra{\tau} = 1 - 2\int_{0}^\tau y\rbra{r}\mathrm{d}r.
    \end{equation*}
    Substituting \cref{eq:y_integral_solution} into this, we obtain a single equation for $x$:
    \begin{equation}\label{eq:x_integral_solution}
        x\rbra{\tau} = 1 - 4\int_{0}^\tau r^{-b/2}\rbra*{\int_{0}^r s^{b/2}x\rbra{s}\mathrm{d}s}\mathrm{d}r.
    \end{equation}
    By Fubini's theorem (applicable because the integrand is absolutely integrable for $b > 2$), we have
    \begin{align*}
        \int_{0}^\tau r^{-b/2}\rbra*{\int_{0}^r s^{b/2}x\rbra{s}\mathrm{d}s}\mathrm{d}r &= \int_{0}^\tau \rbra*{\int_{s}^\tau r^{-b/2}\mathrm{d}r}s^{b/2}x\rbra{s}\mathrm{d}s \\
        &= \int_{0}^\tau \frac{s - \tau^{1-b/2}s^{b/2}}{b/2-1} x\rbra{s}\mathrm{d}s. \tag{$b > 2$}
    \end{align*}
    Now fix a small $h$ (to be chosen).
    For $0 \leq s \leq \tau \leq h$, we have
    \begin{equation}\label{eq:kernel_bound}
        0 \leq \frac{s - \tau^{1-b/2}s^{b/2}}{b/2-1} \leq \frac{s}{b/2-1} \leq \frac{\tau}{b/2-1} \leq \frac{h}{b/2-1}.
    \end{equation}
    This kernel is non-negative and bounded.
    Define an operator $T$ on the Banach space $C\rbra{\sbra{0,h};\RR}$ (continuous functions from $\sbra{0,h}$ to $\RR$ with infinity norm $\Abs*{f}_{\infty} = \sup_{\tau \in \sbra{0,h}}\abs*{f(\tau)}$) by
    \begin{equation*}
        \rbra{T x}\rbra{\tau} = 1 - 4\int_{0}^\tau \frac{s - \tau^{1-b/2}s^{b/2}}{b/2-1} x\rbra{s}\mathrm{d}s.
    \end{equation*}
    For any $x_1,x_2\in C\rbra{\sbra{0,h};\RR}$, with \cref{eq:kernel_bound}, we have
    \begin{align*}
        \abs*{\rbra{T x_1}\rbra{\tau} - \rbra{T x_2}\rbra{\tau}} \leq \frac{4h}{b/2-1}\int_{0}^\tau \abs*{x_1\rbra{s} - x_2\rbra{s}}\mathrm{d}s \leq \frac{4h\tau}{b/2-1}\Abs*{x_1 - x_2}_{\infty} \leq \frac{4h^2}{b/2-1}\Abs*{x_1 - x_2}_{\infty}.
    \end{align*}
    Taking the supremum over $0\leq \tau \leq h$,
    \begin{equation*}
        \Abs*{T x_1 - T x_2}_{\infty} \leq \frac{4h^2}{b/2-1}\Abs*{x_1 - x_2}_{\infty}.
    \end{equation*}
    Choose $h$ small enough so that $\frac{4h^2}{b/2-1} < 1$.
    Then, $T$ is a contraction.
    By the Banach fixed-point theorem, there exists a unique $x\in C\rbra{\sbra{0,h};\RR}$ such that $Tx = x$, which is exactly the desired integral equation of \cref{eq:x_integral_solution}.
    Hence $x$ is well defined and continuous on $\sbra{0,h}$.

    Once $x$ exists, we define $y$ via \cref{eq:y_integral_solution}.
    Because $x$ is continuous, it is bounded near $0$.
    Therefore,
    \begin{equation*}
        y\rbra{\tau} = 2\tau^{-b/2}\int_{0}^{\tau} s^{b/2}x\rbra{s}\mathrm{d}s = 2\tau^{-b/2} O\rbra{\tau^{b/2+1}} = O\rbra{\tau}.
    \end{equation*}
    Hence $y$ is continuous at $\tau = 0$ with $y\rbra{0} = 0$.
    Thus, $z\rbra{\tau} = \begin{psmallmatrix}
        x\rbra{\tau} \\ y\rbra{\tau}
    \end{psmallmatrix}$
    is a continuous solution of the integral equations \cref{eq:x_integral_solution,eq:y_integral_solution}, and differentiating them shows that $x\rbra{\tau}$ and $y\rbra{\tau}$ satisfy the original differential equations for $\tau > 0$.
    This gives a unique local solution on $\sbra{0,h}$.

    Finally, after reaching any small positive time $h > 0$, the coefficient $-\frac{b}{2\tau}$ is no longer singular on $\sbra{h,R}$.
    So the ordinary ODE theory extends the solution uniquely from $\sbra{0,h}$ to $\sbra{0,R}$.
    The resulting function $z^0$ is the desired global solution on $\sbra{0,R}$.
\end{proof}

Since \cref{eq:continuous_limiting} is the continuous limit for $\tau = t\theta > \ceil{2b}\theta$ (recall that $\kappa_b\rbra{t}$ is defined piecewise), the behaviour of the discrete dynamics near the origin requires separate attention.
The following lemma shows that for small $\delta$, the discrete state  $u_t$ remains close to $e_0$ when $t \leq \ceil{\delta/\theta}$.

\begin{lemma}\label{lem:small_deviation}
    Let $\delta \in (0, \pi/32]$ and assume $\theta < \delta$.
    Then
    \begin{equation*}
        \sup_{0 \leq t \leq \ceil{\delta/\theta}} \Abs*{u_t - e_0} \leq 4\delta+12b\delta^2.
    \end{equation*}
\end{lemma}
\begin{proof}
    Recall the recurrence of $\alpha_t$ in \cref{eq:adaptive_alpha_t}, we have
    \begin{equation*}
        R_{\theta}u_t = \Abs*{u_t}\begin{pmatrix}
            \cos\rbra{\alpha_{t+1}} \\
            \sin\rbra{\alpha_{t+1}}
        \end{pmatrix},
    \end{equation*}
    so that
    \begin{equation}\label{eq:u_t_alpha_t}
        u_t = \Abs*{u_t}\begin{pmatrix}
            \cos\rbra{\alpha_{t+1}-2\theta} \\
            \sin\rbra{\alpha_{t+1}-2\theta}
        \end{pmatrix}.
    \end{equation}
    From \cref{prop:survival_product},
    \begin{equation*}
        \Abs*{u_t}^2 = S_t\rbra{\theta,\kappa_b} = \prod_{j=1}^{t}\rbra*{1-\kappa_b\rbra{j}\sin^2\rbra{\alpha_j}}
    \end{equation*}
    with the empty product equal to $1$ when $t=0$.
    Using the inequality $1 - \prod_{j=1}^t\rbra{1-a_j} \leq \sum_{j=1}^ta_j$ for $a_j \in \sbra{0,1}$, we obtain
    \begin{equation}\label{eq:u_t_norm_bound}
        1-\Abs*{u_t}^2 = 1 - \prod_{j = 1}^{t}\rbra*{1-\kappa_b\rbra{j}\sin^2\rbra{\alpha_j}} \leq \sum_{j=1}^t\kappa_b\rbra{j}\sin^2\rbra{\alpha_j} \leq b\sum_{j=1}^t\frac{\alpha_j^2}{j}.
    \end{equation}

    Now, for any $t \leq \ceil{\delta/\theta}$, we estimate
    \begin{equation*}
        \Abs{u_t - e_0} \leq \Abs*{u_t - \Abs*{u_t}\cdot e_0} + \Abs*{\Abs*{u_t}\cdot e_0 - e_0} = \Abs*{u_t}\Abs*{\begin{psmallmatrix} \cos\rbra{\alpha_{t+1}-2\theta} \\ \sin\rbra{\alpha_{t+1}-2\theta} \end{psmallmatrix} - \begin{psmallmatrix} \cos\rbra{0} \\ \sin\rbra{0} \end{psmallmatrix}} + \abs{1-\Abs*{u_t}}.
    \end{equation*}
    Since $\Abs*{u_t} \leq 1$ and $\abs{1-\Abs*{u_t}} \leq 1-\Abs*{u_t}^2$ (because $\Abs*{u_t} \leq 1$), we get
    \begin{equation}\label{eq:u_t_u_0_bound}
        \Abs{u_t - e_0} \leq \abs*{\alpha_{t+1}-2\theta} + \rbra*{1-\Abs*{u_t}^2}.
    \end{equation}

    Because $\theta< \delta \leq \pi/32$, we have $\ceil{\delta/\theta}+1 < 2\ceil{\delta/\theta} < 4\delta/\theta \leq \pi/\rbra{8\theta}$.
    Hence, for all $1\leq t \leq \ceil{\delta/\theta}+1$, we apply \cref{eq:alpha_t_bound}, to obtain
    \begin{equation}\label{eq:alpha_t_bound2}
        0\leq\alpha_t \leq 2t\theta.
    \end{equation}
    In particular, for $1\leq t\leq\ceil{\delta/\theta}$, the recurrence gives $\alpha_{t+1}\geq 2\theta$, so $\abs{\alpha_{t+1}-2\theta}\leq 2t\theta<4\delta$. For $t=0$, $\abs{\alpha_1-2\theta}=\theta<\delta$.
    Moreover, from \cref{eq:u_t_norm_bound} and \cref{eq:alpha_t_bound2}, we have
    \begin{equation*}
        1-\Abs*{u_t}^2 \leq b\sum_{j=1}^t\frac{\alpha_j^2}{j} \leq 2b\theta^2t\rbra{t+1} \leq 2b\theta^2\rbra{\delta/\theta+1}\rbra{\delta/\theta+2} < 2b\theta^2\rbra{2\delta/\theta}\rbra*{3\delta/\theta} = 12b\delta^2,
    \end{equation*}
    where $t \leq \ceil{\delta/\theta} < \delta/\theta+1 < 2\delta/\theta$ since $\theta<\delta$.
    Then, substituting the estimates into \cref{eq:u_t_u_0_bound} gives
    \begin{equation*}
        \Abs{u_t - e_0} \leq 4\delta +12b\delta^2
    \end{equation*}
    for $0 \leq t \leq \ceil{\delta/\theta}$, which completes the proof.
\end{proof}

For $\tau > \delta > 0$, \cref{eq:continuous_limiting} is no longer singular, and standard numerical analysis guarantees that the discrete recurrence approximates the continuous solution when the step size $\theta$ is small.
This is formalised in the next lemma. 

\begin{lemma}\label{lem:fixed_interval_error}
    Fix $\delta \in (0,\pi/32]$, $R > \delta$ and $b > 2$.
    Let $z^0$ be the solution from \cref{lem:ode_solution}.
    Then,
    \begin{equation*}
        \limsup_{\theta\to 0} \sup_{\ceil{\delta/\theta} \leq t \leq \floor{R/\theta}}\Abs*{u_{t} - z^0\rbra{t\theta}} \leq 4\delta+12b\delta^2 +\Abs*{z^0\rbra{\delta} - z^0\rbra{0}}.
    \end{equation*}
\end{lemma}
\begin{proof}
    For $\tau \in \sbra*{\delta,R}$, the coefficient matrix
    \begin{equation*}
        B\rbra{\tau} = \begin{pmatrix}
            0 & -2 \\ 2 & -\frac{b}{2\tau}
        \end{pmatrix}
    \end{equation*}
    is bounded and Lipschitz.
    Moreover, for sufficiently small $\theta$ and $\tau \in \sbra*{\delta,R}$, we have $\tau/\theta \geq \ceil{2b}$.
    Hence, with \cref{eq:A_approximation}, we obtain
    \begin{equation*}
        D_{\kappa_b\rbra{\tau/\theta+1}}R_{\theta} = I + \theta B\rbra{\tau} + O\rbra{\theta^2}.
    \end{equation*}
    Consequently, the recurrence 
    \begin{equation*}
        \hat{u}_{t+1} = D_{\kappa_b\rbra{t+1}}R_{\theta} \hat{u}_t = D_{\kappa_b\rbra{\tau/\theta+1}}R_{\theta} \hat{u}_t \quad \text{for $\ceil{\delta/\theta} \leq t \leq \floor{R/\theta}-1$ with $\hat{u}_{\ceil{\delta/\theta}} = z^0\rbra{\ceil{\delta/\theta}\theta}$}
    \end{equation*}
    agrees with the forward Euler discretisation of $z'\rbra{\tau} = B\rbra{\tau}z\rbra{\tau}$ up to a uniform $O\rbra{\theta^2}$ error per step (see, e.g., \cite{Robinson04}).
    The local truncation error is $O\rbra{\theta^2}$.
    Since the number of steps is $O\rbra{\rbra{R-\delta}/\theta} = O\rbra{1/\theta}$, the global error accumulates to $O\rbra{\theta}$, provided the initial condition $\hat{u}_{\ceil{\delta/\theta}} = z^0\rbra{\ceil{\delta/\theta}\theta}$.
    Hence,
    \begin{equation}\label{eq:euler_error}
        \sup_{\ceil{\delta/\theta} \leq t \leq \floor{R/\theta}}\Abs*{\hat{u}_t - z^0\rbra{t\theta}} = O\rbra{\theta}.
    \end{equation}
    Since the recurrence of $u_t$ is the same as that of $\hat{u}_t$, for $\ceil{\delta/\theta} \leq t \leq \floor{R/\theta}-1$, we have
    \begin{equation*}
        \Abs*{u_{t+1} - \hat{u}_{t+1}} = \Abs*{D_{\kappa_b\rbra{t+1}}R_{\theta} u_t - D_{\kappa_b\rbra{t+1}}R_{\theta} \hat{u}_t} \leq \Abs*{D_{\kappa_b\rbra{t+1}}R_{\theta}}\Abs*{u_t - \hat{u}_t} \leq \Abs*{u_t - \hat{u}_t},
    \end{equation*}
    where $\Abs*{D_{\kappa_b\rbra{t+1}}R_{\theta}}$ is the operator norm of $D_{\kappa_b\rbra{t+1}}R_{\theta}$.
    Then,
    \begin{equation*}\label{eq:hat_u_error}
        \sup_{\ceil{\delta/\theta} \leq t \leq \floor{R/\theta}}\Abs*{u_t - \hat{u}_t} = \Abs*{u_{\ceil{\delta/\theta}} - \hat{u}_{\ceil{\delta/\theta}}}.
    \end{equation*}
    Combining with \cref{eq:euler_error}, we obtain
    \begin{align*}
        \sup_{\ceil{\delta/\theta} \leq t \leq \floor{R/\theta}}\Abs*{u_t - z^0\rbra{t\theta}} &\leq \sup_{\ceil{\delta/\theta} \leq t \leq \floor{R/\theta}}\Abs*{u_t - \hat{u}_t} + \sup_{\ceil{\delta/\theta} \leq t \leq \floor{R/\theta}}\Abs*{\hat{u}_t - z^0\rbra{t\theta}} \\
        &= \Abs*{u_{\ceil{\delta/\theta}} - \hat{u}_{\ceil{\delta/\theta}}} + O\rbra{\theta} \\
        &= \Abs*{u_{\ceil{\delta/\theta}} - z^0\rbra{\ceil{\delta/\theta}\theta}} + O\rbra{\theta}.
    \end{align*}
    Recalling from \cref{lem:small_deviation}, for $\delta \in (0,\pi/32]$ and $\theta < \delta$, we have
    \begin{equation*}
        \Abs*{u_{\ceil{\delta/\theta}} - z^0\rbra{0}} = \Abs*{u_{\ceil{\delta/\theta}} - e_0} \leq 4\delta+12b\delta^2.
    \end{equation*}
    Then,
    \begin{equation*}
        \Abs*{u_{\ceil{\delta/\theta}} - z^0\rbra{\ceil{\delta/\theta}\theta}} \leq \Abs*{u_{\ceil{\delta/\theta}} - z^0\rbra{0}} + \Abs*{z^0\rbra{0} - z^0\rbra{\ceil{\delta/\theta}\theta}} \leq 4\delta+12b\delta^2 + \Abs*{z^0\rbra{\ceil{\delta/\theta}\theta} - z^0\rbra{0}}.
    \end{equation*}
    Therefore,
    \begin{equation*}
        \sup_{\ceil{\delta/\theta} \leq t \leq \floor{R/\theta}}\Abs*{u_t - z^0\rbra{t\theta}} \leq \Abs*{u_{\ceil{\delta/\theta}} - z^0\rbra{\delta}} + O\rbra{\theta} \leq 4\delta+12b\delta^2 + \Abs*{z^0\rbra{\ceil{\delta/\theta}\theta} - z^0\rbra{0}} + O\rbra{\theta}.
    \end{equation*}
    Finally, taking $\theta \to 0$ and using the continuity of $z^0$ (so that $\lim_{\theta\to 0} z^0\rbra{\ceil{\delta/\theta}\theta} = z^0\rbra{\delta}$), we obtain the desired bound.
\end{proof}

We now combine the previous estimates to establish the convergence of the discrete dynamics to the continuous solution on any fixed interval $\sbra{0, R}$.
The following lemma confirms that as the step $\theta$ vanishes, $u_t$ converges uniformly to the solution $z^0$ in \cref{lem:ode_solution}.

\begin{lemma}\label{lem:ode_approximation}
    Fix $R > 0$ and $b > 2$.
    Let $z^0:\sbra{0,R} \to \RR^2$ be the solution from \cref{lem:ode_solution}.
    Then,
    \begin{equation*}
        \lim_{\theta\to 0} \sup_{0 \leq t \leq \floor{R/\theta}}\Abs*{u_{t} - z^0\rbra{t\theta}} = 0.
    \end{equation*}
\end{lemma}
\begin{proof}
    Fix an arbitrary $\delta \in (0,\pi/32]$ with $\delta < R$ and assume $\theta < \min\cbra{\delta, R-\delta}$.
    We split the index range into two parts: the initial segment $0\leq t\leq \ceil{\delta/\theta}$ and the remaining part $\ceil{\delta/\theta} \leq t \leq \floor{R/\theta}$.

    For the initial segment, using \cref{lem:small_deviation}, we obtain
    \begin{align*}
        \sup_{0 \leq t \leq \ceil{\delta/\theta}}\Abs*{u_{t} - z^0\rbra{t\theta}} &\leq  \sup_{0 \leq t \leq \ceil{\delta/\theta}}\Abs*{u_{t} - e_0} +  \sup_{0 \leq t \leq \ceil{\delta/\theta}}\Abs*{e_0 - z^0\rbra{t\theta}} \\
        &\leq 4\delta+12b\delta^2 +  \sup_{0 \leq t \leq \ceil{\delta/\theta}}\Abs*{z^0\rbra{0} - z^0\rbra{t\theta}}.
    \end{align*}
    Because $z^0$ is continuous on $\sbra{0,R}$, we have
    \begin{align*}
        & \limsup_{\theta \to 0} \sup_{0 \leq t \leq \ceil{\delta/\theta}}\Abs*{z^0\rbra{0} - z^0\rbra{t\theta}} \\
        \leq{}& \max\cbra*{\limsup_{\theta \to 0}\sup_{0 \leq \tau \leq \delta}\Abs*{z^0\rbra{0} - z^0\rbra{\tau}}, \limsup_{\theta \to 0}\sup_{\delta \leq \tau \leq \ceil{\delta/\theta}\theta}\Abs*{z^0\rbra{0} - z^0\rbra{\tau}}} \\
        ={}& \max\cbra*{\sup_{0 \leq \tau \leq \delta}\Abs*{z^0\rbra{0} - z^0\rbra{\tau}}, \Abs*{z^0\rbra{0} - z^0\rbra{\delta}}} \\
        ={}& \sup_{0 \leq \tau \leq \delta}\Abs*{z^0\rbra{0} - z^0\rbra{\tau}}.
    \end{align*}
    Consequently,
    \begin{equation*}
        \limsup_{\theta \to 0} \sup_{0 \leq t \leq \ceil{\delta/\theta}}\Abs*{u_{t} - z^0\rbra{t\theta}} \leq 4\delta+12b\delta^2 + \sup_{0 \leq \tau \leq \delta}\Abs*{z^0\rbra{\tau} - z^0\rbra{0}}
    \end{equation*}
    For the second segment, \cref{lem:fixed_interval_error} gives
    \begin{equation*}
        \limsup_{\theta \to 0} \sup_{\ceil{\delta/\theta} \leq t \leq \floor{R/\theta}}\Abs*{u_{t} - z^0\rbra{t\theta}} \leq 4\delta+12b\delta^2 + \Abs*{z^0\rbra{\delta} - z^0\rbra{0}}.
    \end{equation*}
    Combining the two parts, the overall supremum over the full range satisfies
    \begin{align*}
        &\limsup_{\theta \to 0} \sup_{0 \leq t \leq \floor{R/\theta}}\Abs*{u_{t} - z^0\rbra{t\theta}} \\
        \leq{}& \max\cbra*{\limsup_{\theta \to 0}\sup_{0 \leq t \leq \ceil{\delta/\theta}}\Abs*{u_{t} - z^0\rbra{t\theta}}, \limsup_{\theta \to 0}\sup_{\ceil{\delta/\theta} \leq t \leq \floor{R/\theta}}\Abs*{u_{t} - z^0\rbra{t\theta}}} \\
        \leq{}& \max\cbra*{4\delta+12b\delta^2 + \sup_{0 \leq \tau \leq \delta}\Abs*{z^0\rbra{0} - z^0\rbra{\tau}}, 4\delta+12b\delta^2 + \Abs*{z^0\rbra{\delta} - z^0\rbra{0}}} \\
        ={}& 4\delta+12b\delta^2 + \sup_{0 \leq \tau \leq \delta}\Abs*{z^0\rbra{\tau} - z^0\rbra{0}}.
    \end{align*}
    Finally, because $z^0$ is continuous at $\tau = 0$, the right-hand side tends to $0$ as $\delta \to 0$, so we conclude
    \begin{equation*}
        \lim_{\theta \to 0} \sup_{0 \leq t \leq \floor{R/\theta}}\Abs*{u_{t} - z^0\rbra{t\theta}} = 0. \qedhere
    \end{equation*}
\end{proof}

With \cref{lem:ode_approximation}, it is easy to see that
\begin{equation*}
    \lim_{\theta \to 0} \sum_{t=0}^{\floor{R/\theta}}\theta\Abs*{u_t}^2 = \int_{0}^R \Abs*{z^0\rbra{\tau}}^2\mathrm{d}\tau.
\end{equation*}
To obtain the limit of $\sum_{t=0}^{\infty}\theta\Abs*{u_t}^2$ as $\theta \to 0$, we need to control the tail of the series.
The following lemma shows that $\Abs*{u_t}^2$ has an algebraic decay in $t$.

\begin{lemma}\label{lem:algebraic_decay}
    Assume $b > 2$ and fix $p \in \rbra{1,b/2}$.
    There exist constants $C = C\rbra{b, p} > 0$ and $R_0 = R_0\rbra{b,p} > 0$ such that for all $\theta \in \rbra{0,\pi/16}$ and all $t \geq R_0/\theta$,
    \begin{equation*}
        \Abs*{u_t}^2 \leq C \cdot \rbra{t\theta}^{-p}.
    \end{equation*}
\end{lemma}
\begin{proof}
    Recalling the recurrence of $\alpha_t$ in \cref{eq:adaptive_alpha_t} and \cref{prop:survival_product}, we have
    \begin{equation*}
        \Abs*{u_{t+1}}^2 = \Abs*{u_{t}}^2\rbra*{1-\kappa_b\rbra{t+1}\sin^2\rbra{\alpha_{t+1}}} \leq \Abs*{u_{t}}^2\exp\rbra*{-\kappa_b\rbra{t+1}\sin^2\rbra{\alpha_{t+1}}}.
    \end{equation*}
    For any integer $T \geq 2b$ and any integer $q > 0$, we have
    \begin{equation}\label{eq:ut_q_steps_bound}
        \Abs*{u_{T+q}}^2 \leq \Abs*{u_{T}}^2\exp\rbra*{-\sum_{j=0}^{q-1}\kappa_b\rbra{T+j+1}\sin^2\rbra{\alpha_{T+j+1}}} \leq \Abs*{u_{T}}^2\exp\rbra*{-\frac{b}{T+q}\sum_{j=0}^{q-1}\sin^2\rbra{\alpha_{T+j+1}}},
    \end{equation}
    where the last inequality holds because $\kappa_b\rbra{T+j+1} = b/\rbra{T+j+1} \geq b/\rbra{T+q}$.

    We need a lower bound for $\sum_{j=0}^{q-1}\sin^2\rbra{\alpha_{T+j+1}}$.
    Recalling from \cref{eq:alpha_alphat} that
    \begin{equation*}
        \alpha - \arctan\rbra*{\sqrt{1-x}\tan\rbra{\alpha}} = \int_{\sqrt{1-x}}^1 \frac{\tan\rbra{\alpha}}{1+s^2\tan^2\rbra{\alpha}}\mathrm{d}s 
    \end{equation*}
    actually holds for $x \in \sbra{0,1}$ and $\alpha \in \rbra{-\pi/2, \pi/2}$, we obtain
    \begin{align*}
        \abs*{\alpha - \arctan\rbra*{\sqrt{1-x}\tan\rbra{\alpha}}} &\leq \int_{\sqrt{1-x}}^1 \frac{\abs*{\tan\rbra{\alpha}}}{1+s^2\tan^2\rbra{\alpha}}\mathrm{d}s \\
        &= \int_{\sqrt{1-x}}^1 \frac{1}{\frac{1}{\abs*{\tan\rbra{\alpha}}}+s^2\abs*{\tan\rbra{\alpha}}}\mathrm{d}s \\
        &\leq \int_{\sqrt{1-x}}^1 \frac{1}{2s}\mathrm{d}s \\
        &= -\frac{1}{4}\ln\rbra{1-x}.
    \end{align*}
    By restricting $x \in \sbra{0,1/2}$, we have $-\frac{1}{4}\ln\rbra{1-x} \leq x$, thus
    \begin{equation*}
        \abs*{\alpha - \arctan\rbra*{\sqrt{1-x}\tan\rbra{\alpha}}} \leq x.
    \end{equation*}
    By continuity and $\pi$-periodicity, this bound extends to all $\alpha\in\RR$, with the arctangent interpreted continuously as $\alpha$ varies.
    Consequently, for $T \geq 2b$ (so $\kappa_b\rbra{T+j} \leq 1/2$),
    \begin{equation*}
        \abs*{\alpha_{T+j+1} - \rbra*{\alpha_{T+j} + 2\theta}} = \abs*{\arctan\rbra*{\sqrt{1-\kappa_b\rbra{T+j}}\tan\rbra{\alpha_{T+j}}} - \alpha_{T+j}} \leq \kappa_b\rbra{T+j} = \frac{b}{T+j} \leq \frac{b}{T}.
    \end{equation*}
    Then, by induction starting at $T+1$,
    \begin{equation*}
        \abs*{\alpha_{T+j+1} - \rbra*{\alpha_{T+1} + 2j\theta}} \leq \frac{bj}{T}.
    \end{equation*}
    Since $\alpha \to \sin^2\rbra{\alpha}$ is a $1$-Lipschitz function, we have
    \begin{equation}\label{eq:sum_alpha_j}
    \begin{aligned}
        \abs*{\sum_{j=0}^{q-1}\sin^2\rbra{\alpha_{T+j+1}} - \sum_{j=0}^{q-1}\sin^2\rbra{\alpha_{T+1}+2j\theta}} &\leq \sum_{j=0}^{q-1} \abs*{\sin^2\rbra{\alpha_{T+j+1}} - \sin^2\rbra{\alpha_{T+1}+2j\theta}} \\
        &\leq \sum_{j=0}^{q-1} \abs*{\alpha_{T+j+1} - \rbra*{\alpha_{T+1}+2j\theta}} \\
        &\leq \sum_{j=0}^{q-1}\frac{bj}{T} =\frac{bq\rbra{q-1}}{2T}.
    \end{aligned}
    \end{equation}
    Now evaluate the reference sum:
    \begin{equation}\label{eq:sum_alpha_j_theta}
    \begin{aligned}
        \sum_{j=0}^{q-1}\sin^2\rbra{\alpha_{T+1}+2j\theta} &= \sum_{j=0}^{q-1}\rbra*{\frac{1}{2} - \frac{1}{2}\cos\rbra{2\alpha_{T+1}+4j\theta}} \\
        &= \frac{q}{2} - \frac{1}{2}\sum_{j=0}^{q-1}\cos\rbra{2\alpha_{T+1}+4j\theta} \\
        &= \frac{q}{2} - \frac{1}{2}\frac{\sin\rbra*{2q\theta}}{\sin\rbra{2\theta}}\cos\rbra*{2\alpha_{T+1}+2\rbra{q-1}\theta}.
    \end{aligned}
    \end{equation}
    Hence,
    \begin{align*}
        \abs*{\sum_{j=0}^{q-1}\sin^2\rbra{\alpha_{T+j+1}} - \frac{q}{2}} &\leq \abs*{\sum_{j=0}^{q-1}\sin^2\rbra{\alpha_{T+j+1}} - \sum_{j=0}^{q-1}\sin^2\rbra{\alpha_{T+1}+2j\theta}} + \abs*{\sum_{j=0}^{q-1}\sin^2\rbra{\alpha_{T+1}+2j\theta} - \frac{q}{2}} \\
        &\leq \frac{bq\rbra{q-1}}{2T} + \abs*{\frac{\sin\rbra*{2q\theta}}{\sin\rbra{2\theta}}\cos\rbra*{2\alpha_{T+1}+2\rbra{q-1}\theta}} \\
        &\leq \frac{bq\rbra{q-1}}{2T} + \frac{1}{\sin\rbra{2\theta}}.
    \end{align*}
    Combining \cref{eq:sum_alpha_j,eq:sum_alpha_j_theta}, we have
    \begin{equation*}
        \sum_{j=0}^{q-1}\sin^2\rbra{\alpha_{T+j+1}} \geq \frac{q}{2} - \frac{bq\rbra{q-1}}{2T} - \frac{1}{\sin\rbra{2\theta}}.
    \end{equation*}
    With \cref{eq:ut_q_steps_bound}, we obtain
    \begin{equation}\label{eq:ut_q_steps_bound2}
        \begin{aligned}
            \Abs*{u_{T+q}}^2 &\leq \Abs*{u_{T}}^2\exp\rbra*{-\frac{b}{T+q}\rbra*{\frac{q}{2} - \frac{bq\rbra{q-1}}{2T} - \frac{1}{\sin\rbra{2\theta}}}} \\
            &= \Abs*{u_{T}}^2\exp\rbra*{-\frac{q}{T+q}\rbra*{\frac{b}{2} - \frac{b^2\rbra{q-1}}{2T} - \frac{b}{q\sin\rbra{2\theta}}}}.
        \end{aligned}
    \end{equation}
    Now choose parameters so that the exponent is at most $-\frac{pq}{T+q}$.
    We require
    \begin{equation*}
        \frac{b}{2} - \frac{b^2\rbra{q-1}}{2T} - \frac{b}{q\sin\rbra{2\theta}} \geq p.
    \end{equation*}
    A sufficient condition is
    \begin{equation}\label{eq:Tq}
        \frac{b^2q}{2T} \leq \frac{1}{2}\rbra*{\frac{b}{2} - p} \quad \text{and} \quad \frac{b}{q\sin\rbra{2\theta}} \leq \frac{1}{2}\rbra*{\frac{b}{2} - p},
    \end{equation}
    which requires
    \begin{equation*}
        T \geq \frac{2b^2}{b-2p} \cdot q \quad \text{and} \quad q \geq \frac{4b}{b-2p} \cdot \frac{1}{\sin\rbra{2\theta}}.
    \end{equation*}
    Then, since $\theta \in \rbra{0,\pi/16}$, we choose
    \begin{equation*}
        q_0 = \ceil*{\frac{4b}{b-2p} \cdot \frac{\pi}{4\theta}} > \frac{4b}{b-2p} \cdot \frac{1}{\sin\rbra{2\theta}} \quad \text{and} \quad T_0 = \ceil*{\frac{2b^2}{b-2p} \cdot q_0} \geq \frac{2b^2}{b-2p} \cdot q_0,
    \end{equation*}
    to meet \cref{eq:Tq}.
    Let $T_k = T_0+kq_0$ for a non-negative integer $k$; then $T_k$ and $q_0$ meet the condition \cref{eq:Tq}.
    Thus
    \begin{equation*}
        \frac{b}{2} - \frac{b^2\rbra{q_0-1}}{2T_k} - \frac{b}{q_0\sin\rbra{2\theta}} > \frac{b}{2} - \rbra*{\frac{b^2q_0}{2T_k} + \frac{b}{q_0\sin\rbra{2\theta}}} \geq \frac{b}{2} - \rbra*{\frac{b}{2} - p} = p.
    \end{equation*}
    Then, \cref{eq:ut_q_steps_bound2} gives
    \begin{equation*}
        \Abs*{u_{T_{k+1}}}^2 = \Abs*{u_{T_k+q_0}}^2 \leq \Abs*{u_{T_k}}^2\exp\rbra*{-\frac{pq_0}{T_k+q_0}}.
    \end{equation*}
    Iterating,
    \begin{equation*}
        \Abs*{u_{T_{k}}}^2 \leq \Abs*{u_{T_{0}}}^2\exp\rbra*{-p\sum_{j=0}^{k-1}\frac{q_0}{T_j+q_0}} = \Abs*{u_{T_{0}}}^2\exp\rbra*{-p\sum_{j=0}^{k-1}\frac{q_0}{T_0+\rbra{j+1}q_0}}.
    \end{equation*}
    Since $\Abs*{u_{T_{0}}}^2 \leq 1$ and
    \begin{equation*}
        \sum_{j=0}^{k-1}\frac{q_0}{T_0+\rbra{j+1}q_0} \geq \int_{T_0/q_0+1}^{T_0/q_0+\rbra*{k+1}}\frac{1}{x}\mathrm{d}x = \ln\rbra*{\frac{T_0+\rbra{k+1}q_0}{T_0+q_0}},
    \end{equation*}
    we obtain
    \begin{equation*}
        \Abs*{u_{T_k}}^2 \leq \rbra*{\frac{T_0+\rbra{k+1}q_0}{T_0+q_0}}^{-p}.
    \end{equation*}
    For any $t \geq T_0$, let $k_t = \floor{\rbra{t - T_0}/q_0}$.
    Then, $T_{k_t} \leq t < T_{k_t+1} = T_0 +\rbra{k_t+1}q_0$, and by the monotonicity of $\Abs*{u_t}^2$, we have
    \begin{equation}\label{eq:u_t_bound}
        \Abs*{u_{t}}^2 \leq \Abs*{u_{T_{k_t}}}^2 \leq \rbra*{\frac{T_0+\rbra{k_t+1}q_0}{T_0+q_0}}^{-p} \leq \rbra*{\frac{t}{T_0+q_0}}^{-p}.
    \end{equation}
    To choose constants independent of $\theta$, set
    \begin{equation*}
        q_* = \frac{b\pi}{b-2p}+\frac{\pi}{16}, \qquad
        R_0 = \frac{2b^2}{b-2p}q_*+\frac{\pi}{16}, \qquad C = \rbra{R_0+q_*}^p.
    \end{equation*}
    The ceiling bounds and $\theta < \pi/16$ give $q_0\theta < q_*$ and $T_0\theta < R_0$.
    Therefore, for $t \geq R_0/\theta$, \cref{eq:u_t_bound} yields
    \begin{equation*}
        \Abs*{u_t}^2 \leq \rbra*{T_0\theta+q_0\theta}^p\rbra{t\theta}^{-p} < C\rbra{t\theta}^{-p}.
    \end{equation*}
    This completes the proof.
\end{proof}

We are ready to connect the discrete sum of squared norms to the integral of the limiting continuous solution.
The following lemma establishes that as the step size $\theta \to 0$, the sum $\sum_{t=0}^{\infty}\theta\Abs*{u_t}^2$ converges to the integral of $\Abs*{z^0\rbra{\tau}}^2$ over $\tau \in [0,\infty)$.

\begin{lemma}\label{lem:integral_convergence}
    Fix $b > 2$.
    Let $z^0$ be the solution from \cref{lem:ode_solution}.
    Then,
    \begin{equation*}
        \lim_{\theta \to 0}\sum_{t=0}^{\infty}\theta\Abs*{u_t}^2 = \int_{0}^\infty \Abs*{z^0\rbra{\tau}}^2\mathrm{d}\tau.
    \end{equation*}
\end{lemma}
\begin{proof}
    \allowdisplaybreaks
    We first show the convergence on any fixed interval $\sbra{L, R}$ with $0 \leq L < R$.
    From \cref{lem:ode_approximation}, we have
    \begin{equation}\label{eq:ode_approximation_interval}
        \limsup_{\theta\to 0} \sup_{\ceil{L/\theta} \leq t \leq \floor{R/\theta}}\Abs*{u_{t} - z^0\rbra{t\theta}} \leq \limsup_{\theta\to 0} \sup_{0 \leq t \leq \floor{R/\theta}}\Abs*{u_{t} - z^0\rbra{t\theta}} = 0.
    \end{equation}
    Then,
    \begin{align*}
        & \abs*{\sum_{t = \ceil{L/\theta}}^{\floor{R/\theta}}\theta \Abs*{u_t}^2 - \sum_{t = \ceil{L/\theta}}^{\floor{R/\theta}}\theta \Abs*{z^0\rbra{t\theta}}^2} \\
        \leq{}& \sum_{t = \ceil{L/\theta}}^{\floor{R/\theta}} \theta \abs*{\Abs*{u_t}^2 - \Abs*{z^0\rbra{t\theta}}^2} \\
        ={}& \sum_{t = \ceil{L/\theta}}^{\floor{R/\theta}} \theta \abs*{\ave*{u_t - z^0\rbra{t\theta}, u_t + z^0\rbra{t\theta}}} \\
        \leq{}& \sum_{t = \ceil{L/\theta}}^{\floor{R/\theta}} \theta \Abs*{u_t - z^0\rbra{t\theta}} \cdot \Abs*{u_t + z^0\rbra{t\theta}} \\
        \leq{}& \sum_{t = \ceil{L/\theta}}^{\floor{R/\theta}} \theta \Abs*{u_t - z^0\rbra{t\theta}} \cdot \rbra*{2\Abs*{u_t} + \Abs*{u_t - z^0\rbra{t\theta}}} \\ 
        \leq{}& \rbra*{\floor{R/\theta} - \ceil{L/\theta}+1}\theta \cdot \sup_{\ceil{L/\theta} \leq t \leq \floor{R/\theta}}\Abs*{u_{t} - z^0\rbra{t\theta}} \cdot \sup_{\ceil{L/\theta} \leq t \leq \floor{R/\theta}}\rbra*{2\Abs*{u_t} + \Abs*{u_t - z^0\rbra{t\theta}}} \\
        \leq{}& \rbra*{R-L +\theta} \cdot \epsilon_{\theta} \cdot \rbra*{2 + \epsilon_\theta},
    \end{align*}
    where $\epsilon_\theta = \sup_{\ceil{L/\theta} \leq t \leq \floor{R/\theta}} \Abs*{u_t - z^0\rbra{t\theta}}$.
    By \cref{eq:ode_approximation_interval}, $\epsilon_\theta \to 0$ as $\theta \to 0$.
    Hence,
    \begin{equation*}
        \lim_{\theta \to 0} \abs*{\sum_{t = \ceil{L/\theta}}^{\floor{R/\theta}}\theta \Abs*{u_t}^2 - \sum_{t = \ceil{L/\theta}}^{\floor{R/\theta}}\theta \Abs*{z^0\rbra{t\theta}}^2} = 0.
    \end{equation*}
    Because $\Abs*{z^0\rbra{\tau}}^2$ is continuous on $\sbra{L,R}$, the Riemann sum converges to the integral:
    \begin{equation*}
        \lim_{\theta \to 0} \sum_{t = \ceil{L/\theta}}^{\floor{R/\theta}}\theta \Abs*{z^0\rbra{t\theta}}^2 = \int_{L}^R \Abs*{z^0\rbra{\tau}}^2\mathrm{d}\tau.
    \end{equation*}
    Thus,
    \begin{equation}\label{lem:integral_convergence1}
        \lim_{\theta \to 0}\sum_{t = \ceil{L/\theta}}^{\floor{R/\theta}}\theta \Abs*{u_t}^2 = \lim_{\theta \to 0} \sum_{t = \ceil{L/\theta}}^{\floor{R/\theta}}\theta \Abs*{z^0\rbra{t\theta}}^2 = \int_{L}^R \Abs*{z^0\rbra{\tau}}^2\mathrm{d}\tau.
    \end{equation}

    Fix $p \in \rbra{1,b/2}$, by \cref{lem:algebraic_decay}, there exists $C = C\rbra{b,p} > 0$ and $R_0 = R_0\rbra{b,p} > 0$ such that for all $\theta \in \rbra{0,\pi/16}$ and all $t \geq R_0/\theta$,
    \begin{equation*}
        \Abs*{u_t}^2 \leq C\cdot \rbra*{t\theta}^{-p}.
    \end{equation*}
    Hence, for any $R \geq R_0$ and $\theta < \min\cbra{R_0, \pi/16}$,
    \begin{equation}\label{eq:tail_sum}
        \begin{aligned}
        \sum_{t = \ceil{R/\theta}}^{\infty} \theta\Abs*{u_t}^2
        &\leq C\theta^{1-p}\sum_{t = \ceil{R/\theta}}^{\infty}t^{-p}
        \leq C\theta^{1-p}\int_{\ceil{R/\theta} - 1}^\infty x^{-p}\mathrm{d}x \\
        &= \frac{C}{\rbra{p-1}\rbra*{\theta\ceil{R/\theta} - \theta}^{p-1}}
        \leq \frac{C}{\rbra{p-1}\rbra{R-\theta}^{p-1}}.
        \end{aligned}
    \end{equation}
    Then, for any $R' > R \geq R_0$, combining with \cref{lem:integral_convergence1}, we have
    \begin{equation*}
        \int_{R}^{R'} \Abs*{z^0\rbra{\tau}}^2\mathrm{d}\tau = \lim_{\theta \to 0} \sum_{t = \ceil{R/\theta}}^{\floor{R'/\theta}} \theta\Abs*{u_t}^2 \leq \limsup_{\theta \to 0} \sum_{t = \ceil{R/\theta}}^{\infty} \theta\Abs*{u_t}^2 \leq \frac{C}{\rbra*{p-1}R^{p-1}}.
    \end{equation*}
    Taking $R' \to \infty$, we obtain
    \begin{equation}\label{eq:tail_integral}
        \int_{R}^{\infty} \Abs*{z^0\rbra{\tau}}^2\mathrm{d}\tau \leq \frac{C}{\rbra*{p-1}R^{p-1}}.
    \end{equation}
    For any $R > R_0$ and $\theta < \min\cbra{R_0, \pi/16}$, we have
    \begin{align*}
        \abs*{\sum_{t=0}^{\infty}\theta\Abs*{u_t}^2 - \int_{0}^{\infty}\Abs*{z^0\rbra{\tau}}^2\mathrm{d}\tau} &\leq \abs*{\sum_{t=0}^{\floor{R/\theta}}\theta\Abs*{u_t}^2 - \int_{0}^{R}\Abs*{z^0\rbra{\tau}}^2\mathrm{d}\tau} + \abs*{\sum_{t=\floor{R/\theta}+1}^{\infty}\theta\Abs*{u_t}^2 - \int_{R}^{\infty}\Abs*{z^0\rbra{\tau}}^2\mathrm{d}\tau} \\
        &\leq \abs*{\sum_{t=0}^{\floor{R/\theta}}\theta\Abs*{u_t}^2 - \int_{0}^{R}\Abs*{z^0\rbra{\tau}}^2\mathrm{d}\tau} + \abs*{\sum_{t=\floor{R/\theta}+1}^{\infty}\theta\Abs*{u_t}^2} + \abs*{\int_{R}^{\infty}\Abs*{z^0\rbra{\tau}}^2\mathrm{d}\tau} \\
        &\leq \abs*{\sum_{t=0}^{\floor{R/\theta}}\theta\Abs*{u_t}^2 - \int_{0}^{R}\Abs*{z^0\rbra{\tau}}^2\mathrm{d}\tau} + \frac{C}{p-1}\rbra*{\rbra{R-\theta}^{1-p}+R^{1-p}}.
    \end{align*}
    Thus, with \cref{lem:integral_convergence1}, we obtain
    \begin{equation*}
        \limsup_{\theta \to 0}\abs*{\sum_{t=0}^{\infty}\theta\Abs*{u_t}^2 - \int_{0}^{\infty}\Abs*{z^0\rbra{\tau}}^2\mathrm{d}\tau} \leq \frac{2C}{\rbra*{p-1}R^{p-1}}.
    \end{equation*}
    Finally, let $R \to \infty$, with $p > 1$, we have
    \begin{equation*}
        \limsup_{\theta \to 0}\abs*{\sum_{t=0}^{\infty}\theta\Abs*{u_t}^2 - \int_{0}^{\infty}\Abs*{z^0\rbra{\tau}}^2\mathrm{d}\tau} = 0,
    \end{equation*}
    which means
    \begin{equation*}
        \lim_{\theta \to 0} \sum_{t=0}^{\infty}\theta\Abs*{u_t}^2 = \int_{0}^{\infty}\Abs*{z^0\rbra{\tau}}^2\mathrm{d}\tau. \qedhere
    \end{equation*}
\end{proof}

We now arrive at the main theorem of this section.
For the adaptive weak measurement schedule defined by $\kappa_b\rbra{t} = \min\cbra{1/2, b/t}$ with $b > 2$, the expected number of oracle calls can be expressed as $Q_p\rbra{b} = 2\sum_{t\geq 0}\Abs*{u_t}^2-1$ (one preparation call and two calls after each failed measurement).
By the analysis above, as the target probability $p \to 0$ (equivalently $\theta \to 0$), the rescaled sum $\theta\sum_{t\geq 0}\Abs*{u_t}^2$ converges to an integral involving the solution of the limiting differential equation.
Evaluating the integral yields the following asymptotic ratio.

\begin{theorem}
    [Oracle counts for completely unknown $p$]\label{thm:adaptive_unknown_p}
    Let $b > 2$ be the parameter defining the adaptive weak measurement strength $\kappa_b$ in \cref{algo:adaptive_unknown_p}.
    Then, the ratio of expected number of calls to $\cA$ (and $\cA^{\dagger}$) required by \cref{algo:adaptive_unknown_p} compared to $1/\sqrt{p}$, i.e., $\sqrt{p}\rbra*{2\sum_{t\geq 0}\Abs*{u_t}^2-1}$, converges to
    \begin{equation*}
        C\rbra{b} = \frac{b\sqrt{\pi}\Gamma\rbra*{\frac{b+2}{4}}^2\Gamma\rbra*{\frac{b-2}{2}}}{\rbra{b-1}\Gamma\rbra*{\frac{b}{4}}^2\Gamma\rbra*{\frac{b-1}{2}}}.
    \end{equation*}
    as $p \to 0$.
\end{theorem}
\begin{proof}
    With \cref{lem:integral_convergence}, we have
    \begin{equation*}
        \lim_{p\to 0}\sqrt{p} Q_p\rbra{b} = \lim_{p\to 0} 2\sqrt{p}\sum_{t\geq 0}S_t\rbra*{\arcsin\rbra{\sqrt{p}}, \kappa_b} = 2\lim_{\theta \to 0}\sum_{t\geq 0}\theta\Abs*{u_t}^2 = 2\int_{0}^{\infty}\Abs*{z^0\rbra{\tau}}^2\mathrm{d}\tau,
    \end{equation*}
    where $z^0$ is the solution from \cref{lem:ode_solution}.
    Write $z^0\rbra{\tau} = \begin{psmallmatrix}
        x\rbra{\tau} \\ y\rbra{\tau}
    \end{psmallmatrix}$, with \cref{eq:continuous_limiting_xy}, we have
    \begin{equation}\label{eq:continuous_limiting_x}
        x''\rbra{\tau} + \frac{b}{2\tau} x'\rbra{\tau} + 4x\rbra{\tau} = 0
    \end{equation}
    and
    \begin{equation}\label{eq:norm2_integral}
        \int_{0}^{\infty}\Abs*{z^0\rbra{\tau}}^2\mathrm{d}\tau = \int_{0}^{\infty}\rbra*{x\rbra{\tau}^2 + \frac{1}{4}x'\rbra{\tau}^2}\mathrm{d}\tau.
    \end{equation}

    Because $x'\rbra{0} = 0$, we extend $x\rbra{\tau}$ evenly to the whole $\RR$, then use the Fourier transform convention:
    \begin{equation}
        \hat{x}\rbra{\omega} = \cF\sbra*{x}\rbra{\omega} = \int_{-\infty}^{\infty}x\rbra{\tau}e^{-i\omega\tau}\mathrm{d}\tau, \qquad x\rbra{\tau} = \frac{1}{2\pi}\int_{-\infty}^{\infty}\hat{x}\rbra{\omega}e^{i\omega\tau}\mathrm{d}\omega.
    \end{equation}
    With
    \begin{equation*}
        \cF\sbra*{x'}\rbra{\omega} = i\omega\hat{x}\rbra{\omega}, \quad \cF\sbra{x''}\rbra{\omega} = -\omega^2\hat{x}\rbra{\omega}, \quad \text{and} \quad \cF\sbra{\tau f}\rbra{\omega} = i\frac{\mathrm{d}\hat{f}}{\mathrm{d}\omega}\rbra{\omega},
    \end{equation*}
    multiplying \cref{eq:continuous_limiting_x} by $\tau$ gives
    \begin{equation*}
        -i\rbra*{2\omega\hat{x}\rbra{\omega} +\omega^2\frac{\mathrm{d}\hat{x}}{\mathrm{d}\omega}\rbra{\omega}} + \frac{b}{2}i\omega\hat{x}\rbra{\omega} + 4i\frac{\mathrm{d}\hat{x}}{\mathrm{d}\omega}\rbra{\omega} = 0.
    \end{equation*}
    Thus,
    \begin{equation*}
        \frac{1}{\hat{x}\rbra{\omega}}\frac{\mathrm{d}\hat{x}}{\mathrm{d}\omega}\rbra{\omega} = \frac{4-b}{2}\frac{\omega}{4-\omega^2}.
    \end{equation*}
    Integrating gives
    \begin{equation*}
        \hat{x}\rbra{\omega} = C \rbra*{4 - \omega^2}^{\frac{b-4}{4}}
    \end{equation*}
    for a constant $C$ on $\rbra{-2,2}$.
    On either interval outside $\sbra{-2,2}$, the solution is a constant multiple of $\rbra{\omega^2-4}^{\rbra{b-4}/4}$, which must vanish because $\omega\hat{x} \in L^2\rbra{\RR}$ and $b > 2$.
    Since $\hat{x} \in L^2\rbra{\RR}$ also excludes distributions supported at $\omega = \pm 2$, we obtain
    \begin{equation*}
        \hat{x}\rbra{\omega} = C \rbra*{4 - \omega^2}_+^{\frac{b-4}{4}}.
    \end{equation*}
    Then,
    \begin{equation*}
        x\rbra{\tau} = \frac{1}{2\pi}\int_{-\infty}^{\infty}\hat{x}\rbra{\omega}e^{i\omega\tau}\mathrm{d}\omega = \frac{C}{2\pi}\int_{-2}^{2}\rbra*{4 - \omega^2}^{\frac{b-4}{4}}e^{i\omega\tau}\mathrm{d}\omega.
    \end{equation*}
    Using $x\rbra{0} = 1$ and
    \begin{equation}
        \label{eq:beta_gamma}
    \begin{aligned}
        \int_{-2}^{2}\rbra*{4 - \omega^2}^{\frac{b-4}{4}}\mathrm{d}\omega &= \int_{-1}^{1}\rbra*{4 - \rbra{2s}^2}^{\frac{b-4}{4}}\mathrm{d}2s \\
        &= 2^{\frac{b-2}{2}}\int_{-1}^1\rbra{1-s^2}^{\frac{b-4}{4}}\mathrm{d}s \\
        &= 2^{\frac{b}{2}}\int_{0}^1\rbra{1-s^2}^{\frac{b-4}{4}}\mathrm{d}s \\
        &= 2^{\frac{b-2}{2}}\int_{0}^1\rbra{s^2}^{-\frac{1}{2}}\rbra{1-s^2}^{\frac{b-2}{4}-\frac{1}{2}}\mathrm{d}s^2 \\
        &= 2^{\frac{b-2}{2}}\int_{0}^1\rbra{s^2}^{\frac{1}{2}-1}\rbra{1-s^2}^{\frac{b}{4}-1}\mathrm{d}s^2 \\
        &= \frac{2^{\frac{b-2}{2}} \Gamma\rbra*{\frac{1}{2}}\Gamma\rbra*{\frac{b}{4}}}{\Gamma\rbra*{\frac{b+2}{4}}} = \frac{2^{\frac{b-2}{2}} \sqrt{\pi}\Gamma\rbra*{\frac{b}{4}}}{\Gamma\rbra*{\frac{b+2}{4}}},
    \end{aligned}
    \end{equation}
    where the last two equalities use the relation between the Beta and Gamma functions~\cite[Chapter~12.4]{WW21} and $\Gamma\rbra*{\frac{1}{2}} = \sqrt{\pi}$, we have
    \begin{equation*}
        C = \frac{2\sqrt{\pi}\Gamma\rbra*{\frac{b+2}{4}}}{2^{\frac{b-2}{2}}\Gamma\rbra*{\frac{b}{4}}} \quad \text{ and } \quad \hat{x}\rbra{\omega} = \frac{2\sqrt{\pi}\Gamma\rbra*{\frac{b+2}{4}}}{2^{\frac{b-2}{2}}\Gamma\rbra*{\frac{b}{4}}} \rbra*{4 - \omega^2}_+^{\frac{b-4}{4}}.
    \end{equation*}

    Since we have extended $x\rbra{\tau}$ evenly to the whole $\RR$, with Parseval-Plancherel identity, we obtain
    \begin{align*}
        \int_{0}^{\infty}\rbra*{x\rbra{\tau}^2 + \frac{1}{4}x'\rbra{\tau}^2}\mathrm{d}\tau &= \frac{1}{4\pi}\int_{-\infty}^{\infty}\rbra*{\abs*{\cF\rbra{x}\rbra{\omega}}^2 + \frac{1}{4}\abs*{\cF\rbra{x'}\rbra{\omega}}^2}\mathrm{d}\omega \\
        &= \frac{1}{4\pi}\int_{-\infty}^{\infty}\rbra*{1+\frac{\omega^2}{4}}\abs*{\hat{x}\rbra{\omega}}^2\mathrm{d}\omega \\
        &= \frac{\Gamma\rbra*{\frac{b+2}{4}}^2}{2^{b-2}\Gamma\rbra*{\frac{b}{4}}^2}\int_{-2}^{2}\rbra*{1+\frac{\omega^2}{4}}\rbra*{4 - \omega^2}^{\frac{b-4}{2}}\mathrm{d}\omega.
    \end{align*}
    With the similar derivation of \cref{eq:beta_gamma}, we have
    \begin{align*}
        \int_{-2}^{2}\rbra*{4 - \omega^2}^{\frac{b-4}{2}}\mathrm{d}\omega &= 2^{b-2}\int_{0}^1\rbra{1-s^2}^{\frac{b-4}{2}}\mathrm{d}s = 2^{b-3}\int_{0}^1\rbra{s^2}^{\frac{1}{2}-1}\rbra{1-s^2}^{\frac{b-2}{2}-1}\mathrm{d}s^2 = 2^{b-3}\sqrt{\pi}\frac{\Gamma\rbra*{\frac{b-2}{2}}}{\Gamma\rbra*{\frac{b-1}{2}}}. \\
        \int_{-2}^{2}\frac{\omega^2}{4}\rbra*{4 - \omega^2}^{\frac{b-4}{2}}\mathrm{d}\omega &= 2^{b-2}\int_{0}^{1}s^2\rbra*{1 - s^2}^{\frac{b-4}{2}}\mathrm{d}s = 2^{b-3}\int_{0}^{1}\rbra{s^2}^{\frac{3}{2}-1}\rbra*{1 - s^2}^{\frac{b-2}{2}-1}\mathrm{d}s^2 = 2^{b-4}\sqrt{\pi}\frac{\Gamma\rbra*{\frac{b-2}{2}}}{\Gamma\rbra*{\frac{b+1}{2}}}.
    \end{align*}
    Hence,
    \begin{align*}
        \int_{0}^{\infty}\rbra*{x\rbra{\tau}^2 + \frac{1}{4}x'\rbra{\tau}^2}\mathrm{d}\tau &= \frac{\Gamma\rbra*{\frac{b+2}{4}}^2}{2^{b-2}\Gamma\rbra*{\frac{b}{4}}^2}\rbra*{2^{b-3}\sqrt{\pi}\frac{\Gamma\rbra*{\frac{b-2}{2}}}{\Gamma\rbra*{\frac{b-1}{2}}} + 2^{b-4}\sqrt{\pi}\frac{\Gamma\rbra*{\frac{b-2}{2}}}{\Gamma\rbra*{\frac{b+1}{2}}}} \\
        &= \frac{\sqrt{\pi}\Gamma\rbra*{\frac{b+2}{4}}^2}{4\Gamma\rbra*{\frac{b}{4}}^2}\rbra*{2\frac{\Gamma\rbra*{\frac{b-2}{2}}}{\Gamma\rbra*{\frac{b-1}{2}}} + \frac{\Gamma\rbra*{\frac{b-2}{2}}}{\Gamma\rbra*{\frac{b+1}{2}}}} \\
        &= \frac{\sqrt{\pi}\Gamma\rbra*{\frac{b+2}{4}}^2}{4\Gamma\rbra*{\frac{b}{4}}^2}\rbra*{\frac{2\Gamma\rbra*{\frac{b-2}{2}}}{\Gamma\rbra*{\frac{b-1}{2}}} + \frac{2\Gamma\rbra*{\frac{b-2}{2}}}{\rbra{b-1}\Gamma\rbra*{\frac{b-1}{2}}}} \\
        &= \frac{b\sqrt{\pi}\Gamma\rbra*{\frac{b+2}{4}}^2\Gamma\rbra*{\frac{b-2}{2}}}{2\rbra{b-1}\Gamma\rbra*{\frac{b}{4}}^2\Gamma\rbra*{\frac{b-1}{2}}}.
    \end{align*}
    Multiplying by $2$ implies the desired
    \begin{equation*}
        C\rbra{b} = \frac{b\sqrt{\pi}\Gamma\rbra*{\frac{b+2}{4}}^2\Gamma\rbra*{\frac{b-2}{2}}}{\rbra{b-1}\Gamma\rbra*{\frac{b}{4}}^2\Gamma\rbra*{\frac{b-1}{2}}}. \qedhere
    \end{equation*}
\end{proof}

After establishing the asymptotic constant $C\rbra{b}$ for the adaptive weak measurement schedule, we can explore its dependence on the parameter $b$.
Define the normalised expected cost
\begin{equation*}
    R_p\rbra{b} = \frac{2\sqrt{p}\,Q_p\rbra{b}}{\pi}.
\end{equation*}
Then $R_p\rbra{b} \to 2C\rbra{b}/\pi$ as $p \to 0$.
\cref{fig:Cb} plots $2C\rbra{b}/\pi = C\rbra{b}/\rbra{\pi/2}$ for $b \in \rbra{2,10}$, where $\pi/2$ is the corresponding asymptotic constant for standard Grover search.
As $b \to 2^+$, the function diverges because the Gamma factor $\Gamma\rbra*{\frac{b-2}{2}}$ becomes singular.
It then decreases to a minimum near $b \approx 5.2$, before increasing slowly for larger $b$.

\definecolor{paperNavy}{HTML}{345995}
\definecolor{paperOrange}{HTML}{E17C05}
\definecolor{paperGreen}{HTML}{2A9D8F}
\definecolor{paperRed}{HTML}{D1495B}

\begin{figure}[htbp]
\centering
\begin{tikzpicture}
\begin{axis}[
    scale=.9,
    width=9cm,
    height=6cm,
    title={$2C(b)/\pi$},
    xlabel={$b$},
    x label style={overlay},
    x tick label style={
        overlay
    },
    xmin=2-.5, xmax=10+.5,
    ymin=1-.5,
    ytick distance=1,
    xtick distance=1,
    grid=major,
    minor y tick num=2,
    minor x tick num=9,
]

\addplot[
    very thick,
    mark=none,
    color=paperRed,
] table[
    x=b,
    y=C,
] {./figures/Cb.dat};

\addplot[
    only marks,
    mark=*,
    mark size=2.2pt,
    color=paperRed,
] coordinates {(5.2,1.28467878)};

\node[
    yshift=5pt,
    anchor=south west,
] at (axis cs:5.2,1.28467878)
{$(5.2,1.28467878)$};

\end{axis}
\end{tikzpicture}
\vspace{1cm}
\caption{Limiting normalised cost $\lim_{p\to 0}R_p\rbra{b} = 2C\rbra{b}/\pi$ for our adaptive schedule. The analytic expression diverges as $b \to 2^+$ and has a numerical minimum near $b = 5.2$, where $C\rbra{b} \approx 2.02$.}
\label{fig:Cb}
\end{figure}
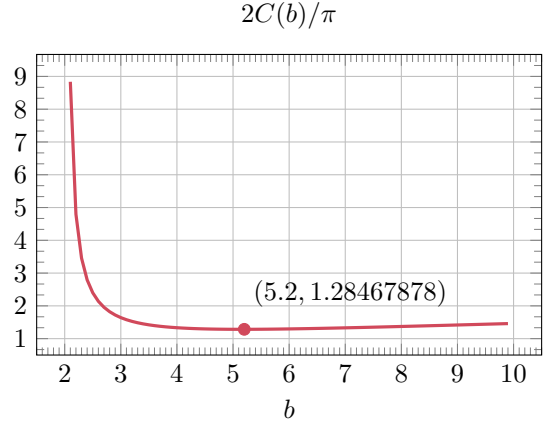

To put our result in context, we compare our schedule with those of \cite{Mizel09} and \cite[Algorithm~6.1]{Wang18}.
Both are designed for the setting where $p$ is completely unknown; in particular, \cite{Wang18} proves the same $1/\sqrt{p}$ scaling.
We use the following strength functions in the same weakly measured loop, preparing $\cA\ket{0^n}$ and measuring before each Grover rotation:
\begin{equation*}
    \kappa^{\text{M}}\rbra{t} = 1 - \rbra*{\frac{1-\sin\rbra{\pi/(2t)}}{1+\sin\rbra{\pi/(2t)}}}^2
\end{equation*}
and
\begin{equation*}
    \kappa_b^{\text{W}}\rbra{t} = 1-e^{-1/2^{\ceil*{\log_2\rbra{t/b}}}}, \qquad t \geq 1,
\end{equation*}
respectively.
The Wang-type family retains the dyadic structure, with the initial blocks specified by the formula above; the original algorithm sets block lengths using a target error. We compare these strength functions under our preparation and query conventions.
For comparison, we re-parameterise the latter schedule by $b$ so that, for large $t$,
\begin{equation*}
    \kappa_b^{\mathrm{W}}\rbra{t} \sim \lambda_b\rbra{t}
    = \frac{1}{2^{\ceil*{\log_2\rbra{t/b}}}} \in \left(\frac{b}{2t},\frac{b}{t}\right].
\end{equation*}
This places both families on the same $1/t$ scale, but equal values of $b$ need not be equally well tuned.

\begin{figure}[htbp]
\centering
\input{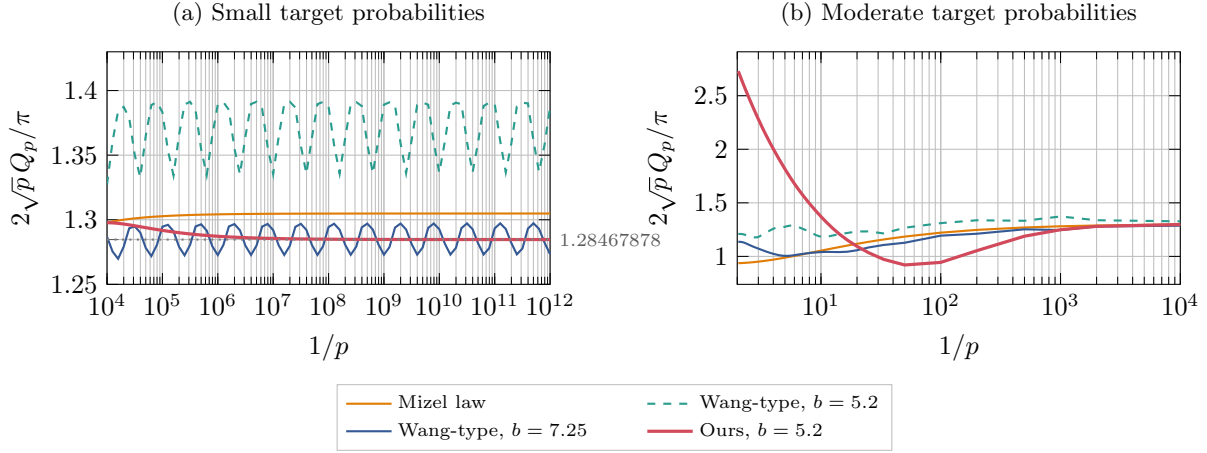}
\begingroup
\pgfplotsset{comparison axis/.style={
    width=1\linewidth,
    height=5cm,
    xlabel={$1/p$},
    ylabel={$2\sqrt{p}\,Q_p/\pi$},
    grid=both,
    axis line style={black},
    tick style={black},
    tick label style={black},
    label style={black},
    x label style={overlay},
    y label style={overlay},
    x tick label style={overlay},
    title style={font=\small},
    legend cell align=left,
    legend style={font=\scriptsize,fill=white,draw=gray!50},
    unbounded coords=discard,
}}
\begin{minipage}[t]{0.49\linewidth}
\vspace{0pt}
\centering
\begin{tikzpicture}
\begin{axis}[
    scale=.9,
    comparison axis,
    title={(a) Small target probabilities},
    xmode=log,
    xmin=1e4,xmax=1e12,
    xtick={1e4,1e5,1e6,1e7,1e8,1e9,1e10,1e11,1e12},
    ymin=1.25,ymax=1.43,
    y label style={yshift=-8mm},
    extra y ticks={1.28467878},
    extra y tick labels={\(1.28467878\)},
    extra y tick style={
        ticklabel pos=right,
        tick label style={font=\scriptsize,gray!70!black,overlay},
        tick style={gray!70!black},
    },
    legend to name=comparisonlegend,
    legend columns=2,
    legend style={/tikz/every even column/.append style={column sep=7mm}},
]
\addplot[paperOrange,thick] table[x=N,y=M] {figures/comparison_small_plot.dat};
\addlegendentry{Mizel law}
\addplot[paperGreen,thick,dashed] table[x=N,y=W52] {figures/comparison_small_plot.dat};
\addlegendentry{Wang-type, $b=5.2$}
\addplot[paperNavy,thick] table[x=N,y=Wbest] {figures/comparison_small_plot.dat};
\addlegendentry{Wang-type, $b=\comparisonWangB$}
\addplot[paperRed,very thick] table[x=N,y=smooth] {figures/comparison_small_plot.dat};
\addlegendentry{Ours, $b=5.2$}
\addplot[black!50,dotted,thick,forget plot] coordinates {(1e4,1.28467878) (1e12,1.28467878)};
\end{axis}
\end{tikzpicture}
\end{minipage}
\hfill
\begin{minipage}[t]{0.49\linewidth}
\vspace{0pt}
\centering
\begin{tikzpicture}
\begin{axis}[
    scale=.9,
    comparison axis,
    title={(b) Moderate target probabilities},
    xmode=log,
    xmin=2,xmax=1e4,
]
\addplot[paperOrange,thick] table[x=N,y=M] {figures/comparison_moderate_plot.dat};
\addplot[paperGreen,thick,dashed] table[x=N,y=W52] {figures/comparison_moderate_plot.dat};
\addplot[paperNavy,thick] table[x=N,y=Wbest] {figures/comparison_moderate_plot.dat};
\addplot[paperRed,very thick] table[x=N,y=smooth] {figures/comparison_moderate_plot.dat};
\end{axis}
\end{tikzpicture}
\end{minipage}
\par\vspace{2.5\baselineskip}
\noindent\makebox[\linewidth][c]{\pgfplotslegendfromname{comparisonlegend}}
\endgroup
\caption{Finite-$p$ normalised expected cost $R_p = 2\sqrt{p}\,Q_p/\pi$ for the stated strength laws, with the same preparation and loop convention. Panel (a) uses 81 logarithmically spaced points from $1/p = 10^4$ to $10^{12}$; panel (b) extends the comparison to $p = 0.49$. The Wang-type value $b = 7.25$ is selected once by the fixed-grid procedure described in the text and held fixed in both panels. Each curve is evaluated using a common tolerance $10^{-6}$ for the analytical bound on the omitted normalised cost, with the cutoff chosen as described below.}
\label{fig:comparison}
\end{figure}

The finite-$p$ costs are compared in \cref{fig:comparison}, using the same asymptotic Grover reference $\pi/\rbra{2\sqrt{p}}$ for normalisation throughout as \cref{fig:Cb}.
We perform a systematic fixed-$b$ comparison on $p_j = 10^{-12+j/10}$, $j = 0,\ldots,80$, with equal weight for each logarithmic grid point.
For the Wang-type family we test $b \in \cbra{3,3.25,\ldots,12}\cup\cbra{5.2}$ and compare the arithmetic means of the normalised costs using the lower and upper bounds defined below.
A candidate is discarded once its lower mean exceeds the upper mean of another candidate; otherwise, its numerical sum is extended.
This selects $b = 7.25$, whose upper mean is below the lower means of all other tested values, and which is then held fixed for every $p$, including the moderate-$p$ panel.

\paragraph{Numerical evaluation and the omitted tail.}
We evaluate $S_t = \Abs*{u_t}^2$ using \cref{eq:u_m_adaptive} and approximate the expected costs by
\begin{equation*}
    Q_p^{[T]} = 2\sum_{t=0}^{T}S_t-1,
    \qquad R_p^{[T]} = \frac{2\sqrt{p}\,Q_p^{[T]}}{\pi}.
\end{equation*}
The cutoff is chosen from an upper bound on the omitted sum.
Write $d_t = \sqrt{1-\kappa\rbra{t}}$ and define
\begin{equation*}
    H\rbra{d} = \begin{pmatrix}
        d & -\dfrac{\rbra{1-d}\cos\rbra{2\theta}}{2\sin\rbra{2\theta}} \\[1ex]
        -\dfrac{\rbra{1-d}\cos\rbra{2\theta}}{2\sin\rbra{2\theta}} & 1
    \end{pmatrix},
    \qquad \lambda\rbra{d} = 1-\frac{1-d}{2}\rbra*{1+\frac{1}{\sin\rbra{2\theta}}}.
\end{equation*}
Here $\lambda\rbra{d}$ is the smallest eigenvalue of $H\rbra{d}$.
Direct calculation gives $A_{\theta,\kappa}^{\dagger}H\rbra{d}A_{\theta,\kappa} = dH\rbra{d}$ for $d = \sqrt{1-\kappa}$.
Since $d_t$ is nondecreasing for all three schedules and the matrices $H\rbra{d}/\lambda\rbra{d}$ decrease in the positive semidefinite order whenever $\lambda\rbra{d} > 0$, induction gives
\begin{equation*}
    S_t \leq \frac{u_T^{\dagger}H\rbra{d_T}u_T}{\lambda\rbra{d_T}}\prod_{j=T+1}^{t}d_j,
    \qquad t > T, \quad \lambda\rbra{d_T} > 0.
\end{equation*}
Consequently, for any $F_T \geq \sum_{t>T}\prod_{j=T+1}^{t}d_j$,
\begin{equation*}
    0 \leq R_p-R_p^{[T]} \leq B_T,
    \qquad B_T = \frac{4\sqrt{p}}{\pi}\frac{u_T^{\dagger}H\rbra{d_T}u_T}{\lambda\rbra{d_T}}F_T.
\end{equation*}
For our schedule with $T \geq 2b$, and for the Mizel law, respectively, we use
\begin{equation*}
    F_T = \frac{T+1}{b/2-1},
    \qquad F_T = \frac{T+1}{\pi-1}.
\end{equation*}
These follow by bounding $d_j$ by $e^{-b/(2j)}$ and $e^{-\pi/j}$, respectively, and summing the resulting power bounds.
For our schedule with $b > 2$ and any fixed $0 < p < 1/2$, we have $d_T \to 1$ and $\lambda\rbra{d_T} \to 1$ as $T \to \infty$. Thus, the bound on $S_t$ with the first choice of $F_T$ gives $\sum_{t\geq 0}S_t < \infty$, proving almost-sure termination and finite expected oracle cost.
For the Wang-type law, at a block endpoint $T = \floor{b2^k}$, we use
\begin{equation*}
    F_T = e^{2^{-k-1}}\rbra*{\frac{b2^k}{1-2e^{-b/4}}+\frac{1}{1-e^{-b/4}}},
    \qquad b > 4\ln\rbra{2}.
\end{equation*}
This follows by summing over the subsequent dyadic blocks, whose lengths differ from $b2^{k+j-1}$ by at most one; the displayed bound includes these rounding errors.
All tested Wang-type parameters satisfy $b > 4\ln\rbra{2}$.
For each plotted value, we stop at the first numerical checkpoint with $\lambda\rbra{d_T} > 0$ and $B_T \leq 10^{-6}$.
Thus $\sbra*{R_p^{[T]},R_p^{[T]}+B_T}$ bounds the infinite normalised cost.

\section*{AI use statement}
\addcontentsline{toc}{section}{AI use statement}
The authors used ChatGPT Pro (GPT-5.5/5.4) and OpenAI Codex to assist with language editing, exploratory coding, preliminary literature searches, and brainstorming of analytical approaches.
The authors verified all results, revised the manuscript as needed, and take full responsibility for its content.

\addcontentsline{toc}{section}{References}

\bibliographystyle{alphaurl}
\bibliography{main}

\end{document}